\documentclass[12pt,a4paper]{article}
\usepackage{amsmath,amsthm}
\usepackage{amsmath,amssymb}
\usepackage{cases}
\numberwithin{equation}{section}

\newcommand {\Cg} {{\cal C}} %E gothique
\newcommand {\Eg} {{\cal E}} %E gothique
\newcommand {\Fg} {{\cal F}}
\newcommand {\Hg} {{\cal H}} %H gothique
\newcommand {\Mg} {{\cal M}} %N gothique
\newcommand {\Ng} {{\cal N}} %N gothique
\newcommand {\Sg} {{\cal S}} %S gothique
\newcommand {\Xg} {{\cal X}} %S gothique
\newcommand {\Yg} {{\cal Y}} %S gothique

\newcommand {\R} {{\mathbb R}}

\newcommand {\C} {{\mathbb C}} 

\newcommand {\Cb} {{\overline{\C}}}
\newcommand {\Bb} {{\bf B}}

\newcommand {\Ga} {\Gamma}
\newcommand {\p} {\psi}
\newcommand {\f} {\phi}
\newcommand {\ep} {\varepsilon}

\newcommand {\Th} {\Theta}

\newcommand {\ka} {\kappa}
\newcommand {\lam} {\lambda}

\newcommand {\n} {\nu}
\newcommand {\s} {\sigma}
\newcommand {\Si} {\Sigma}
\newcommand {\w} {\omega}

\def\lbeq(#1){\label{eqn:#1}}
\def\refeq(#1){{\rm (\ref{eqn:#1})}}
\def\refeqs(#1,#2){{\rm (\ref{eqn:#1}, \ref{eqn:#2})}}
\def\lbth(#1){\label{th:#1}}
\def\refth(#1){{\rm Theorem \ref{th:#1}}}
\def\refthb(#1){{\bf Theorem \ref{th:#1}}}
\def\refths(#1,#2){{\rm Theorems \ref{th:#1} and \ref{th:#2}}}
\def\refthsb(#1,#2){{\bf Theorems \ref{th:#1} and \ref{th:#2}}}
\def\lblm(#1){\label{lm:#1}}
\def\reflm(#1){{\rm Lemma \ref{lm:#1}}}
\def\reflmss(#1,#2,#3){{\rm Lemmas \ref{lm:#1}, \ref{lm:#2} and \ref{lm:#3}}}
\def\reflms(#1,#2){{\rm Lemmas \ref{lm:#1} and \ref{lm:#2}}}
\def\reflmb(#1){{\bf Lemma \ref{lm:#1}}}
\def\lbprp(#1){\label{prp:#1}}
\def\refprp(#1){{\rm Proposition \ref{prp:#1}}}
\def\lbass(#1){\label{ass:#1}}
\def\refass(#1){{\rm Assumption \ref{ass:#1}}}
\def\lbcor(#1){\label{cor:#1}}
\def\refcor(#1){{\rm Corollary \ref{cor:#1}}}
\def\lbsssec(#1){\label{sssec:#1}}
\def\refsssec(#1){\S \ref{sssec:#1}}
\def\lbssec(#1){\label{ssec:#1}}
\def\refssec(#1){\S \ref{ssec:#1}}

\newcommand {\bgth} {\begin{theorem}}
\newcommand {\edth} {\end{theorem}}
\newcommand {\bgprp} {\begin{proposition}}
\newcommand {\edprp} {\end{proposition}}
\newcommand {\bgdf} {\begin{definition}}
\newcommand {\eddf} {\end{definition}}
\newcommand {\bgass} {\begin{assumption}}
\newcommand {\edass} {\end{assumption}}
\newcommand {\bglm} {\begin{lemma}}
\newcommand {\edlm} {\end{lemma}}
\newcommand {\bgcor} {\begin{corollary}}
\newcommand {\edcor} {\end{corollary}}
\newcommand {\bgpf} {\begin{proof}}
\newcommand {\edpf} {\end{proof}}
\newcommand {\bgrm} {\begin{remark}}
\newcommand {\edrm} {\end{remark}}

\newcommand {\bqn} {\begin{equation}}
\newcommand {\eqn} {\end{equation}}

\newcommand {\ben} {\begin{enumerate}}
\newcommand {\een} {\end{enumerate}}
\newcommand {\ph} {\varphi}

\newcommand {\la} {\langle}
\newcommand {\ra} {\rangle}

\newcommand {\ax} {{\la x \ra}}
\newcommand {\ay} {{\la y \ra}}

\newcommand {\br} {\begin{array}}
\newcommand {\er} {\end{array}}
\newcommand {\lap} {\Delta}
\newcommand {\pa} {\partial}

\newtheorem{theorem}{Theorem}[section]
\newtheorem{lemma}[theorem]{Lemma}
\newtheorem{proposition}[theorem]{Proposition}
\newtheorem{definition}[theorem]{Definition}
\newtheorem{corollary}[theorem]{Corollary}

\newtheorem{assumption}[theorem]{Assumption}
\theoremstyle{definition}
\newtheorem{remark}[theorem]{Remark}
\newcommand {\absleq} {{\leq_{|\, \cdot\, |}\, }}

\title{Remarks on $L^p$-boundedness of wave 
operators for Schr\"odinger operators 
with threshold singularities}
\author{
K. Yajima\thanks{
Department of Mathematics, Gakushuin University, 
1-5-1 Mejiro, Toshima-ku, Tokyo 171-8588, Japan. 
Supported by JSPS grant in aid for scientific research 
No. 22340029 }}
\date{}
\begin{document}
\allowdisplaybreaks
\maketitle

\begin{abstract}
We consider the continuity property in Lebesgue spaces $L^p(\R^m)$ 
of wave operators $W_\pm$ of scattering theory for Schr\"odinger operator 
$H=-\lap + V$ on $\R^m$, $|V(x)|\leq C\ax^{-\delta}$ for some $\delta>2$ 
when $H$ is of exceptional type, i.e. $\Ng=\{u \in \ax^{-s} L^2(\R^m) \colon 
(1+ (-\lap)^{-1}V)u=0 \}\not=\{0\}$ for some 
$1/2<s<\delta-1/2$. It has recently 
been proved by Goldberg and Green for $m\geq 5$ that  
$W_\pm$ are bounded in $L^p(\R^m)$ for 
$1\leq p<m/2$, the same holds for $1\leq p<m$ if all $\f\in \Ng$ 
satisfy $\int_{\R^m} V\f dx=0$ and, 
for $1\leq p<\infty$ if in addition 
$\int_{\R^m} x_i V\f dx=0$, $i=1, \dots, m$. 
We make the results for $p>m/2$ more precise and prove 
in particular that these conditions are also necessary 
for the stated properties of $W_\pm$. We also prove that, 
for $m=3$, $W_\pm$ are bounded in $L^p(\R^3)$ for 
$1<p<3$ and that the same holds for $1<p<\infty$ if and only 
if all $\f\in \Ng$ satisfy $\int_{\R^3}V\f dx=0$ and  
$\int_{\R^3} x_i V\f dx=0$, $i=1, 2, 3$, simultaneously. 

\end{abstract}

\section{Introduction} 

Let $H_0=-\lap$ be the free Schr\"odinger operator on the 
Hilbert space $\Hg=L^2(\R^m)$ with domain domain 
$D(H_0)=\{u\in \Hg\colon -\lap u \in \Hg\}$ and 
$H= H_0 + V$, $V$ being the 
multiplication operator with the real measurable function 
$V(x)$ which satisfies   
\begin{equation}\lbeq(decay-con)
|V(x)|\leq C\ax^{-\delta} \ \ \text{for some }\ \delta>2, \ 
\ax=(1+|x|^2)^\frac12. 
\end{equation}
Then, $H$ is selfadjoint in $\Hg$ with a core 
$C_0^\infty(\R^m)$ and it satisfies the following 
properties (see e.g. \cite{KatoS1,Ku-0,RSi,RS3,RS4}):  
\ben
\item[{\rm (i)}]  The spectrum $\s(H)$ of $H$ consists 
of the absolutely continuous (AC for short) part 
$[0,\infty)$ and a finite number of non-positive 
eigenvalues of finite multiplicities. 
\een
We write $\Hg_{ac}(H)$ for the AC spectral subspace 
of $\Hg$ for $H$ and $P_{ac}(H)$ for the orthogonal 
projection onto $\Hg_{ac}(H)$.   
\ben 
\item[{\rm (ii)}] Wave operators 
$W_\pm=\lim_{t \to\pm\infty}e^{itH}e^{-itH_0}$ 
defined by strong limits exist and are complete, viz.  
${\rm Image}\, W_\pm = \Hg_{ac}(H)$. 
They are unitary 
from $\Hg$ onto $\Hg_{ac}(H)$ and intertwine $H_{ac}$ 
and $H_0$. Hence, for Borel functions $f$,    
\begin{equation}\lbeq(inter)
f(H)P_{ac}(H) = W_\pm f(H_0) W_\pm^\ast.
\end{equation}
\een
If follows that various mapping properties of $f(H)P_{ac}$ 
may be deduced from those of $f(H_0)$ if the corresponding 
ones of $W_\pm$ are known. In particular, if 
$W_\pm \in \Bb(L^p(\R^m))$ for 
$1\leq p_1 \leq p\leq p_2 < \infty$,  
then $W_\pm^\ast \in \Bb(L^q(\R^m))$ for 
$q_2 \leq q \leq q_1$, $1/p_j+1/q_j=1$, $j=1,2$, and   
\bqn \lbeq(pq)
\|f(H)P_{ac}(H)\|_{\Bb( L^q, L^p)} 
\leq C_{pq}\|f(H_0)\|_{\Bb(L^q, L^p)}, 
\eqn 
for these $p$ and $q$ with $C_{pq}$ which 
are {\it independent of $f$}. 
We define the Fourier and the conjugate Fourier 
transforms $\Fg{u}(\xi)$ and $\Fg^\ast{u}(\xi)$ 
respectively by    
\[
\Fg u(\xi) = \int_{\R^m} e^{-ix\xi} u(x) dx\ \mbox{and} \  
\Fg^\ast u(\xi) = \frac1{(2\pi)^m}
\int_{\R^m} e^{ix\xi} u(x) dx.
\] 
We also write $\hat{u}(\xi)$ for $\Fg u(\xi)$. 

The intertwining property \refeq(inter) 
may be made more precise. Wave operators $W_\pm$ are 
transplantations (\cite{Stempak}) of the complete set  
of (generalized) eigenfunctions 
$\{e^{ix\xi} \colon \xi \in \R^m\}$ of $-\lap$ by 
those of out-going and in-coming scattering eigenfunctions 
$\{\ph_{\pm}(x,\xi) \colon \xi\in \R^m\}$ of 
$H=-\lap + V$ (\cite{Ku-0}):  
\[
W_{\pm}u(x) = \Fg_{\pm}^\ast \Fg u(x)=
\frac1{(2\pi)^d}\int_{\R^d} \ph_{\pm}(x,\xi) 
\hat u(\xi) d\xi, 
\]
where $\Fg_{\pm}$ and $\Fg_\pm^\ast$ are the generalized 
Fourier transforms associated with 
$\{\ph_{\pm}(x,\xi) \colon \xi\in \R^m\}$ 
and the conjugate ones defined respectively by 
\[
\Fg_{\pm} u(\xi) = \int_{\R^d}\overline{\ph_{\pm}(x,\xi)} 
u(x) dx, \quad 
\Fg_{\pm}^\ast u(\xi) = 
\frac1{(2\pi)^m} 
\int_{\R^d}\ph_{\pm}(x,\xi)u(x) dx.
\]
They satisfy $\Fg_{\pm}^\ast \Fg_{\pm} u = u$ for 
$u \in \Hg_{ac}(H)$ and, $\Fg_\pm \Fg_{\pm}^\ast u = u$ 
for $u \in L^2(\R^m)$. 
We define $F(D) \equiv \Fg^\ast M_F \Fg$ and 
$F(D_\pm) \equiv \Fg_{\pm}^\ast M_F \Fg_{\pm}$ 
for Borel functions $F$ on $\R^m$ where $M_F$ is the multiplication 
with $F(\xi)$. Then,  
\[
F(D_\pm) = W_\pm F(D)W_{\pm}^\ast u , \quad u \in \Hg_{ac}(H) 
\]
and $W_\pm$ transplant estimates for $F(D)$ in $L^p$-spaces 
to $F(D_\pm)$.  

In this paper we are interested in the problem whether or not 
$W_\pm$ are bounded in $L^p(\R^m)$. This will almost 
automatically imply the same property in Sobolev spaces 
$W^{k,p}(\R^m)= \{ u \in L^p(\R^m) \colon \pa^\alpha 
u \in L^p(\R^m)\}$ 
for integers $0 \leq k\leq 2$ (see Section 7 of \cite{FY}). 

There is now a large literature on this problem 
(\cite{AY, Be, DF,FY, Weder, Y-d2, JY-1, JY-2, Y-d3, Y-odd}) and 
it is well known that the answer depends on the spectral 
properties of $H$ at $0$, the bottom of the AC spectrum 
of $H$. We define 
\bqn 
\Eg=\{u \in H^2(\R^m) \colon (-\lap + V)u = 0\}, 
\eqn 
the eigenspace of $H$ with eigenvalue $0$ 
and, for $1/2<s<\delta-1/2$,  
\begin{equation}\lbeq(thres)
\Ng= \{u \in \ax^{s}L^2(\R^m): 
(1 + (-\lap)^{-1}V) u=0\} = 0 . 
\end{equation} 
Functions $\f$ in $\Ng$ satisfy 
$-\lap \f + V \f=0$ for $x \in \R^m$. 
The space $\Ng$ is finite dimensional, 
independent of $1/2<s<\delta-1/2$,  
$\Eg \subset \Ng$ and, if $m\geq 5$, $\Eg=\Ng$ (\cite{JK}). 
The operator $H$ is 
said be of {\it generic type} if 
$\Ng=\{0\}$ and of {\it exceptional type} otherwise. 
When $H$ is of {\it generic type}, we have rather 
satisfactory results (though there is much space for 
improving conditions on $V$) and it 
has been proved that $W_\pm$ are bounded in $L^p(\R^m)$ 
for all $1\leq  p \leq \infty$ if $m \geq 3$ and, for all 
$1<p<\infty$ if $m=1$ and $m=2$ under various smoothness 
and decay at infinity assumptions on $V$ (see \cite{Be} 
for the best result when $m=3$); but they are 
in general {\it not} bounded in $L^1(\R^1)$ or 
$L^\infty(\R^1)$ when $m=1$ (\cite{Weder}). 

When $H$ is of exceptional type, it is long known that the 
same results hold when $m=1$ (see \cite{Weder, AY, DF}). 
For higher dimensions $m\geq 3$, it is first shown 
(\cite{Y-odd, FY}) 
that $W_\pm$ are bounded in $L^p(\R^m)$ for $3/2<p<3$ if $m=3$ 
and for $\frac{m}{m-2}<p<\frac{m}2$ if $m\geq 5$, which is subsequently 
extended to $1<p<3$ for $m=3$ and $1<p<m/2$ for $m\geq 5$ 
(\cite{Y-arxiv}). Then, recently, Goldberg and Green 
(\cite{GG}) have substantially improved these results 
by proving the following theorem for $m\geq 5$. In what 
follows 
in this paper, we assume $m \geq 3$ and $V$ satisfies the 
following assumption. The constant $m_\ast$ is defined by  
\[
m_\ast=(m-1)/(m-2).
\]
\bgass \lbass(V) 
$V$ is a real valued measurable function such that 
\ben 
\item[{\rm (1)}] $\Fg (\ax^{2\s} V) \in L^{m_\ast}$ 
for some $\s>1/m_\ast$. 
\item[{\rm (2)}] $|V(x)|\leq C \ax^{-\delta}$ 
for some $\delta> \left\{\br{ll} m+4, 
\  &  \mbox{if}\ \ 3\leq m\leq 7, \\ 
m+3, \  & \mbox{if}\ \ m\geq 8 \er \right.$ and $C>0$. 
\een
\edass  
\noindent The condition (1) requires certain 
smoothness on $V$. 

We write $\la u,v\ra= \int_{\R^m} \overline{u(x)}v(x) dx$ 
and define subspaces $\Eg_1 \subset \Eg_0 \subset \Ng$ 
respectively by 
\bqn \lbeq(eg0eg1)
\Eg_0=\{\f\in \Ng \colon \la V, \f \ra =0\}, \quad   
\Eg_1=\{\f \in \Eg_0 \colon \la x V, \f\ra=0\},  
\eqn 
where $\la xV,\f\ra=0$ 
means $\la x_i V, \f\ra=0$ for all $1\leq i\leq m$.
We have $\dim \Ng/\Eg_0 \leq 1$, 
$\Eg_0= \Eg$ if $m=3$ and $\Ng=\Eg$ if $m \geq 5$. 

\bgth[Goldberg-Green] \lbth(GG) Suppose that $V$ 
satisfies \refass(V) and that 
$H$ is of exceptional type. Then, if $m\geq 5$, $W_\pm$ 
are bounded in $L^p(\R^m)$ for $1\leq p<m/2$. 
They are bounded in $L^p(\R^m)$ also for $1\leq p<m$ 
if $\Ng=\Eg_0$ and for $1\leq p<\infty$ if $\Ng=\Eg_1$. 
\edth 

In this paper, we show following theorems which 
in particular prove the corresponding 
result for $m=3$ 
and that $\Ng=\Eg_0$ and $\Ng=\Eg_1$ of 
\refth(GG) are also necessary conditions for 
respective cases. We write $P$, $P_{0}$ 
and $P_1$ for the orthogonal 
projections onto $\Eg$, $\Eg_{0}$ and $\Eg_{1}$ respectively. 
Because $(-\lap)^{-1}V$ is a real operator, we may take 
the bases of $\Ng$, $\Eg_0$ and $\Eg_1$ which consist   
of real functions and $P$, $P_{0}$ and $P_1$ are real 
operators: For the conjugation 
$(\Cg u)(x) = \overline{u(x)}$, 
\bqn \lbeq(conj)
\Cg^{-1} P \Cg=P, \quad \Cg^{-1} P_0 \Cg=P, \quad  
\Cg^{-1} P_1 \Cg=P_1.
\eqn 
We state results for $m=3$, $m=5$ and $m\geq 6$ separately. 
If $m=3$, it is known that $W_\pm$ are not bounded in 
$L^1(\R^3)$ when $H$ 
is of exceptional type (\cite{Y-disp, ES}).

\begin{theorem}\lbth(theo) Let $m=3$. Suppose that $V$ 
satisfies \refass(V) and that $H$ is of exceptional type. 
Then: 
\ben 
\item[{\rm (1)}] $W_\pm$ are bounded in $L^p(\R^3)$ 
for $1<p<3$. 
\item[{\rm (2)}] For $3<p<\infty$, there 
exists a constant $C$ such that  
\bqn \lbeq(3d-th)
\|(W_\pm \pm a \ph \otimes |D|^{-1}V\ph + 
P)u\|_{L^p}
\leq C \|u\|_{L^ p}, 
\eqn  
where $\ph$ is the real function defined by 
\refeq(canonical) (called canonical resonace),  
$a= 4\pi i|\la V, \ph\ra|^{-2}$ and $P$ 
may be replaced by $P\ominus P_1$.
\item[{\rm (3)}] If $W_\pm$ are bounded in $L^p(\R^3)$ 
for some $3<p<\infty$, then $\Ng=\Eg_1$. In this case they  
are bounded in $L^p(\R^3)$ for all $1<p<\infty$.  
\een
\end{theorem}

\begin{theorem}\lbth(theo5) Let $m=5$. Suppose that $V$ 
satisfies \refass(V) and that $H$ is of exceptional type. 
Then: 
\ben 
\item[{\rm (1)}] $W_\pm$ are bounded in $L^p(\R^5)$ for 
$1<p<5/2$. 
\item[{\rm (2)}] For $5/2<p<5$, 
there exists a constant $C$ such that  
\bqn \lbeq(5d-th-1)
\left\|\left(W_\pm  \pm a_0 (|D|^{-1}V\ph) \otimes \ph + 
\frac{P}{2} \right)u \right\|_{L^p}
\leq C \|u\|_{L^p}, 
\eqn  
where $\ph= PV$, $V$ being considered as a function,  
$a_0= i/(24\pi^2)$ and $P$ may be replaced by $P\ominus P_0$. 
If $W_\pm$ are bounded in $L^p(\R^5)$ for some $\frac52<p<5$, 
then $\Ng=\Eg_0$. In this case  
they are  bounded in $L^p(\R^5)$ for all $1<p<5$. 
\item[{\rm (3)}] By virtue of {\rm (1)} and {\rm (2)}, 
the condition $\Eg=\Eg_0$ is 
necessary for $W_\pm$ to be bounded in $L^p(\R^5)$ for 
some $p>5$. Suppose $\Eg=\Eg_0$. Then, 
\bqn \lbeq(5d-th-2)
\left\|\left(W_\pm + P\right)u
\right\|_{L^p}
\leq C \|u\|_{L^p}  
\eqn 
for a constant $C$, where $P=P_0$ may be replaced by 
$P_0\ominus P_1$. 
If $W_\pm$ are bounded in $L^p(\R^5)$ for some $p>m$, 
then $\Ng=\Eg_1$. In this case they are 
bounded in $L^p(\R^5)$ for all $1<p<\infty$. 
\een
\end{theorem}

\bgth \lbth(theo6)
Let $m\geq 6$. Suppose that $V$ satisfies 
\refass(V) and that $H$ is of exceptional type. Then: 
\ben 
\item[{\rm (1)}] $W_\pm$ are bounded in $L^p(\R^m)$ 
for $1<p<m/2$. 
\item[{\rm (2)}] For $\frac{m}2<p<m$, 
there exists a constant $C>0$ such that 
\bqn 
\|(W_\pm + D_m P)u\|_{L^p} 
\leq C_p \|u\|_{L^{p}}, 
\eqn 
where $P$ may be replaced by $P\ominus P_{0}$ and 
\begin{numcases}
{D_m=} \frac{\Ga\left(\frac{m-2}2\right)}
{\sqrt{\pi} \Ga\left(\frac{m-1}2\right)}, & 
$\mbox{$m$ is odd}$, \\
\frac{2^m \Ga\left(\frac{m}{2}\right)}
{\sqrt{\pi}\Ga\left(\frac{m-1}2\right)} 
\int_1^\infty (x^2+1)^{-(m-1)}dx, & $\mbox{$m$ is even}$. 
\lbeq(dm-e1)
\end{numcases} 
If $W_\pm$ are bounded in $L^p(\R^m)$ for some $m/2<p<m$ 
then, $\Eg=\Eg_0$. In this case they are bounded in $L^p(\R^m)$ 
for all $1<p<m$ 
\item[{\rm (3)}] Suppose $\Eg=\Eg_0$. 
Let $m<p<\infty$. 
Then, for a constant $C_p$, 
\bqn \lbeq(5d-th2)
\|(W_\pm + P)u\|
\leq C_p \|u\|_{L^{p}},
\eqn 
where $P$ may be replaced by $P_0 \ominus P_1$. 
If $W_\pm$ are bounded in $L^p(\R^m)$ for some $p>m$, then 
$\Eg=\Eg_1$. In this case they are bounded in 
$L^p(\R^m)$ for all $1<p<\infty$.
\een
\end{theorem}

\bgrm 
{\rm (1)} The integral in \refeq(dm-e1) may be 
computed explicitly: 
\bqn \lbeq(shin)
\int_1^\infty (x^2+1)^{-(m-1)}dx
= \frac{\Ga\left(m-\frac32\right)}{4\Ga(m-1)}
\left(\sqrt{\pi}
-\sum_{j=1}^{m-2}\frac{\Ga(j)2^{-j+1}}
{\Ga\left(j+\frac12\right)}
\right). 
\eqn 
{\rm (2)} There are examples of $V$ such that 
$\Eg_1= \Eg_0 \subsetneq \Ng$, 
$\Eg_1 \subsetneq \Eg_0 =\Ng$ and 
$\Eg_1 \subsetneq \Eg_0 \subsetneq \Ng$ 
(see Example 8.4 of \cite{Jensen}). 

\noindent 
{\rm (3)} Murata's result (Theorem 1.2 of \cite{Mu}) 
also implies that, if $\Ng\not=0$, $W_\pm$ are not in general 
bounded in $L^p(\R^m)$ for $p>3$ if $m=3$ and for $p>\frac{m}2$ 
if $m\geq 5$.  
\edrm 

The rest of the paper is devoted to the proof of 
Theorems. In spite that substantial part of 
\refths(theo5,theo6) overlaps with \refth(GG) and that  
they miss critically important $L^1$-boundedness, we 
present the proof of Theorems which is very different 
from the one by Goldberg and Green (\cite{GG}). Our proof  
heavily uses harmonic analysis machinery, which produces 
sharper results for larger $p$'s, however, at the same time, 
prevents us from reaching end points $p=1$ and $p=\infty$. 
We prove the theorems only for $W_{-}$ since conjugation 
changes the direction of time, 
viz. $\Cg^{-1}e^{-itH}\Cg=e^{itH}$, and  
\bqn 
W_{+}= \Cg^{-1} W_{-}\Cg. 
\eqn 

We use the following notation and conventions: 
The $\ell$-th derivative of $f(x)$, $x\in \R$ 
is denoted by $f^{(\ell)}(x)$. 
$\Si={\mathbb S}^{m-1}=\{x\colon x_1^2+ \cdots + x_m^2=1\}$ 
is the unit sphere in $\R^m$ and 
$\w_{m-1}=2\pi^\frac{m}{2}/\Ga\left(\frac{m}2\right)$ 
is its area. 
The coupling and the inner product are anti-linear 
with respect to the first component,  
\[
(u,v)= \la u, v \ra = \int_{\R^n} \overline{u(x)}v(x) dx, 
\]
in accordance with the interchangeable notaion 
for the rank $1$ operator 
\[
|u\ra \la v| = u \otimes v \ \colon \ 
\f \mapsto u \la v, \f \ra.
\]
This notation is used also when $v$ is in a 
certain function space and $u$ in its dual space.
\[
f \, \absleq \, g \ \mbox{means}\ |f|\leq |g|.
\]
For Banach spaces $X$ and $Y$, $\Bb(X,Y)$ 
is the Banach space of bounded operators from $X$ to $Y$ 
and $\Bb(X)=\Bb(X,X)$; $\Bb_\infty(X,Y)$ and 
$\Bb_\infty(X)$ are spaces of compact operators; 
and the dual space $\Bb(X,\C)$ of $X$ is denoted by $X^\ast$. 
The identity operators in various 
Banach spaces are indistinguishably denoted by $1$. 
For $1\leq p \leq \infty$, 
$\|u\|_p=\|u\|_{L^p}$ is the norm of $L^p(\R^m)$ and $p'$ is 
its dual exponent, 
$1/p+1/p'=1$. When $p=2$, we often omit $p$ 
and write $\|u\|$ for $\|u\|_2$. We interchangeably write 
$L^p_{w}(\R^m)$ or $L^{p,\infty}(\R^m)$ for weak-$L^p$ spaces 
and $\|u\|_{p,w}$ or $\|u\|_{p,\infty}$ for their norms. 
For $s \in \R$, 
\[
L^2_{s}=\ax^{-s}L^2=
L^2(\R^m, \ax^{2s}dx), \quad H^s(\R^m)=\Fg L^2_s(\R^m)
\]
are the weighted $L^2$ spaces  and Sobolev spaces. 
The space of rapidly decreasing functions is 
denoted by $\Sg(\R^m)$. 

We denote the resolvents of $H$ and $H_0$ respectively by  
\[
R(z)= (H-z)^{-1}, \quad R_0(z) = (H_0 -z)^{-1}.
\]
We parameterize $z \in \C \setminus [0,\infty)$ 
as $z=\lam^2$ by $\lam \in \C^+$, the open upper half plane 
of $\C$, so that the positive and the negative parts of 
the boundary $\{\lam \colon \pm \lam \in (0,\infty)\}$ 
are mapped onto the upper and the lower edges of 
the positive half line $\{z\in \C \colon z>0\}$. We define 
\[
G(\lam) = R(\lam^2), \quad G_0(\lam) = R_0(\lam^2), 
\quad  \lam \in \C^+.
\]
These are $\Bb(\Hg)$-valued meromorphic functions of 
$\lam \in \C^+$ and the limiting absorption principle 
\cite{Ku-0} (LAP for short) asserts that, when considered as 
$\Bb(\ax^{-s}L^2, \ax^{t}L^2)$-valued 
functions for $s, t>\frac12$ and $s+t>2$, $G_0(\lam)$ 
has H\"older continuous extensions to its closure 
$\Cb^+=\{z: \Im z \geq 0\}$. The same is true also for  
$G(\lam)$, but, if $H$ is of exceptional type, 
it has singularities at $\lam=0$. In what follows $z^\frac12$ 
is the branch of square root of $z$ cut along the negative 
real axis such that $z^\frac12>0$ when $z>0$. 

The plan of the paper is as follows: In section 2, we record 
some results most of which are well known 
and which we use in the sequel. They include:  
\begin{itemize}
\item[$\bullet$] Formulas for the integral kernel of 
$G_0(\lam)$ as exponential-polynomials in odd 
dimensions or their superpositions in even dimensions.  

\item[$\bullet$] 
Representation of 
$\la \p|(G_0(\lam)-G_0(-\lam))u\ra$ 
as the linear combination of Fourier transforms of 
$r^{j+1}M(r, \overline{\p}\ast \check{u})$, 
$M(r,f)$ being the average of $f$ over the sphere of radius $r$ 
centered at the origin. 

\item[$\bullet$] The Muckenhaupt weighted inequality and 
examples of $A_p$-weights. 

\end{itemize}
In section 3, we recall 
and improve results of \cite{Y-odd} and \cite{FY} on the 
behavior as $\lam \to 0$ of $(1+ G_0(\lam)V)^{-1}$ and 
reduce the problem to the $L^p$-boundedness of  
\bqn \lbeq(zs-repeat)
Z_s u= -\frac{1}{\pi i}\int^\infty_{0}  
G_0(\lam)VS(\lam)(G_0(\lam)-G_0(-\lam))u\lam F(\lam) d\lam 
\eqn 	
where $S(\lam)$ is the singular part of the expansion of 
$(1+ G_0(\lam)V)^{-1}$ at $\lam=0$ and 
$F\in C_0^\infty(\R)$ is such that 
$F(\lam)=1$ near $\lam=0$. 

We prove \refth(theo) in Section 4, \refth(theo5) 
for odd dimensions $m\geq 5$ in Section 5 and for even 
dimensions in Section 6. We explain the basic strategy 
of the proof at the end of \refssec(4-1) after most of 
basic ideas appears in the simplest form. 

{\bf Acknowledgment} \ The part of this work was carried 
out while the author was visiting Aalborg 
and Aarhus universities in summers of 2014 and 2015 
respectively. He would like to express his 
sincere gratitude to both institutions for the hospitality, 
to Arne Jensen, Jacob Schach Moeller and Erik Skibsted 
in particular. He thanks Professors Goldberg and Green 
for sending us their manuscript \cite{GG} before posting on 
the arXiv-math, Shin Nakano for computing the recursion 
formula for \refeq(shin) and Fumihiko Nakano for bringing 
\cite{GKP} to his attention.  

%%%%%%%%%%%%%%%%%%%%%%%%%%%%%%%%%%%%%%%%%%%%%%

\section{Preliminaries} 

In this section we record some well known results 
which we use in what follows.

\subsection{Integral kernel of the free resolvent}
For $m\geq 2$, resolvent $G_0(\lam)$ for 
$\Im \lam \geq 0$ is the convolution with  
\begin{equation}\lbeq(co-ker) 
G_0(\lam,x) = \frac{e^{i\lam|x|}}{2(2\pi)^{\frac{m-1}{2}}
\Ga\left(\frac{m-1}{2}\right)|x|^{m-2}}
\int^\infty_0 e^{-t}t^{\frac{m-3}2}
\left(\frac{t}2-i\lam|x|\right)^{\frac{m-3}2}dt   
\end{equation}
(\cite{WW}). When $m\geq 3$ is odd, 
it is an exponential polynomial like function.  
\bglm \lblm(odd-green) Let $m \geq 3$ be odd. Then:
  
\begin{equation}\lbeq(coker)
G_0(\lam,x) = \sum_{j=0}^{(m-3)/2} 
C_j \frac{(\lam|x|)^j e^{i\lam |x|}}{|x|^{m-2}} \ \mbox{with}\ 
C_j=\frac{(-i)^j(m-3-j)!}
{2^{m-1-j}\pi^{\frac{m-1}2}j! (\frac{m-3}2 -j)!}. 
\end{equation} 
The constant $C_0$ may also be written as 
$C_0= (m-2)^{-1}\w_{m-1}^{-1}$ and   
\bqn \lbeq(c0c1)
iC_0 +C_1=0, \quad \mbox{when $m \geq 5$}.
\eqn 
\edlm 

If $m$ is even, the structure of $G_0(\lam, x)$ 
is more complex and this makes the analysis harder. 
For partly circumventing the difficulty 
we express $G_0(\lam,x)$ as a 
superposition of exponential-polynomial like functions 
of the form \refeq(coker). 
This will allow a part of the proof for even dimensions  
to go in parallel with the odd dimensional cases.  We set 
\[
\n= \frac{m-2}2 .
\]
Define operators $T_j^{(a)}$, $j=0, \dots, \n$ for 
superposing over parameter $a>0$ by 
\begin{gather} \lbeq(sum-j)
T_j^{(a)}[f(x,a)] = C_{m,j}\w_{m-1}
\int_0^\infty (1+a)^{-(2\n-j+\frac12)} 
f(x,a) \frac{da}{\sqrt{a}}, \\
C_{m,j}\w_{m-1}= 
(-2i)^j\frac{\Ga\left(2\n-j+\frac12\right)}{(m-2)!\sqrt{\pi}}
\begin{pmatrix} \n \\ j \end{pmatrix}.  \lbeq(const-cm0) 
\end{gather}
The factor $\w_{m-1}$ is added for shorting 
some formulas below (see \refeq(usformula)).  

\bglm If $m\geq 4$ is even, then we have 
\bqn \lbeq(co-kerb) 
G_0(\lam,x) =
\sum_{j=0}^\n \w^{-1}_{m-1}T_j^{(a)}
\left[e^{i\lam|x|(1+2a)}\frac{(\lam|x|)^{j}}{|x|^{m-2}} 
\right]. 
\eqn 
\edlm 
\bgpf Let $C_{m\ast} = 2^{m-1}\pi^{\frac{m-1}{2}}
\Ga\left(\frac{m-1}{2}\right)$. 
In the formula \refeq(co-ker): 
\bqn \lbeq(co-ker-t)
G_0(\lam,x) = \frac{e^{i\lam|x|}}
{C_{m\ast} |x|^{m-2}}
\int^\infty_0 e^{-t}t^{\frac{m-3}{2}}
\left(t-2i\lam|x| \right)^{\frac{m-3}{2}}dt,
\eqn 
write $(t-2i\lam|x|)^{\frac{m-3}{2}}
=(t-2i\lam|x|)^{\n}(t-2i\lam|x|)^{-\frac{1}{2}}$, 
expand  $(t-2i\lam|x|)^{\n}$ via the 
binomial formula and use the identity  
\bqn \lbeq(frac)
z^{-\frac12}= \frac1{\sqrt{\pi}}\int_0^\infty 
e^{-az}\, a^{-\frac12}\, da, \quad {\Re\, } z>0 
\eqn
for $(t-2i\lam|x|)^{-\frac12}$. The right hand side 
of \refeq(co-ker-t) becomes 
\[
\sum_{j=0}^\n \frac{(-2i)^j}{\sqrt{\pi}C_{m\ast}}
\begin{pmatrix} \n \\ j \end{pmatrix} 
\iint_{\R_{+}^2} e^{-(1+a)t} t^{2\n-j}  
\left(e^{i\lam|x|(1+2a)}\frac{(\lam|x|)^{j}}{|x|^{m-2}} \right) 
\frac{dt}{\sqrt{t}}\frac{da}{\sqrt{a}}. 
\]
The integral converges absolutely if $m\geq 4$ and 
we obtain \refeq(co-kerb) after 
performing the integral with respect to $t$. \edpf

\subsection{Spectral measure of $H_0$}. 
The spectral measure of $H_0=-\lap$ is AC and Stone's theorem 
implies that the spectral projection $E_0(d\mu)$ 
is given for $\mu=\lam^2$, $\lam>0$ by    
\[
E_0(d\mu)=\frac{1}{2\pi i}
(R_0(\mu+i0)-R_0(\mu-i0))d\mu  
= \frac{1}{i\pi} (G_0(\lam)-G_0(-\lam)) \lam d\lam .
\]

\bglm \lblm(mulfo)  Let $m \geq 3$ and 
$u, v \in (L^1\cap L^2)(\R^m)$. Then, 
both sides of the 
following equation 
can be continuously extended to $\lam=0$ and  
\bqn \lbeq(devi-by)
\lam^{-1} \la v, (G_0(\lam)-G_0(-\lam)) u \ra = 
\la |D|^{-1}v, (G_0(\lam)-G_0(-\lam)) u \ra , \quad 
\lam \geq 0.  
\eqn 
For bounded continuous functions $f$ on $\R$ 
we have for $\lam \geq 0$, 
\bqn 
f(\lam)\la v, (G_0(\lam)u -G_0(-\lam))u \ra 
= \la v, (G_0(\lam)u -G_0(-\lam))f(|D|)u \ra. \lbeq(fd-1) 
\eqn
\edlm 
\bgpf For $u, v \in (L^1\cap L^2)(\R^m)$ we have   
\bqn 
\la v, (G_0(\lam)-G_0(-\lam)) u \ra= 
 \frac{\lam^{m-2}i}{2(2\pi)^{m-1} } \int_{\Si} 
\overline{\hat v(\lam\w)}\hat{u}(\lam\w)d\w,  
\lbeq(spec-mea)
\eqn 
where $\Si={\mathbb S}^{m-1}$. It follows, since 
$\widehat{|D|^{-1}v}(\lam\w)= \lam^{-1}\hat{v}(\lam\w)$, 
$\lam>0$, that 
\bqn 
\la |D|^{-1} v, (G_0(\lam)-G_0(-\lam)) u \ra
= \frac{\lam^{m-3}i}{2(2\pi)^{m-1} } 
\int_{\Si} \overline{\hat v(\lam\w)} \hat{u}(\lam\w)d\w.  
\lbeq(spec-mea-1)
\eqn 
The right side extends to a continuous 
function of $\lam \geq 0$ when $m\geq 3$ 
and \refeq(devi-by) follows by comparing 
\refeq(spec-mea) 
and \refeq(spec-mea-1). Eqn. \refeq(fd-1) likewise follows. 
\edpf 

We define the spherical average of a function $f$ on 
$\R^m$ by 
\begin{equation}\lbeq(M)
M(r, f) = \frac1{\omega_{m-1}}\int_{\Si} f(r\w) d\w, 
\quad \mbox{for all}\ r \in \R. 
\end{equation}
We often write $M_f(r)=M(r,f)$. We have 
$M_f(-r)=M_f(r)$ and H\"older's inequality implies 
\bqn \lbeq(m-holder)
\left(\frac1{\w_{m-1}}
\int_0^\infty |M_f(r)|^p r^{m-1}dr\right)^{1/p} 
\leq \|f\|_p, \quad 1\leq p \leq \infty.
\eqn 
For an even function $M(r)$ of $r\in \R$, 
define $\tilde M(\rho)$ by  
\bqn \lbeq(tiM)
\tilde{M}(\rho )= 
\int_{\rho }^\infty rM(r) dr 
\left(= -\int_{-\infty}^\rho  rM(r)dr\right).
\eqn 
\bglm Suppose $M(r)=M(-r)$ and $\la r \ra^2 M(r)$ is 
integrable. Then, 
\bqn  
\lbeq(MtildeM-1)
\int_{\R} e^{-ir\lam } r M(r)dr 
= \frac{\lam}{i} \int_\R e^{-ir \lam} \tilde M(r) dr,\\ 
\int_{\R} \tilde{M}(r) dr= 
\int_{\R} r^2M (r) dr. 
\eqn 
\edlm 
\bgpf Since $rM(r)= -\tilde{M}(r)'$, integration by parts 
gives the first equation. We differentiate 
both sides of the first and set $\lam=0$. The second 
follows. \edpf 

We denote $\check u(x) =u(-x)$, $x\in \R^m$. 
(The sign $\check{u}$ will be reserved for this 
purpose and will not be used to denote the conjugate 
Fourier transform.)

\paragraph{Representation formula for odd dimensions.}
\begin{lemma}\lblm(p-stg) Let $m\geq 3$ be odd 
and $u,\p \in C_0^\infty(\R^m)$. Define 
$c_j = {\w}_{m-1}C_j$, $1\leq j\leq \frac{m-3}2$, 
where $C_j$ are the constants in {\rm \refeq(coker)}. 
Then, for $\lam>0$ we have 
\bqn \lbeq(spheric)
\la \p , (G_0(\lam)-G_0(-\lam))u \ra  
= \sum_{j=0}^{\frac{m-3}2} c_j (-1)^{j+1}\lam^{j} 
\int_{\R} e^{-i\lam r}r^{1+j} 
M_{\overline{\p}\ast \check u}(r)dr. 
\eqn 
\end{lemma} 
\begin{proof} We compute $\la \p , G_0 (\lam )u \ra$ 
by using the integral kernel \refeq(coker) of $G_0(\lam)$. 
Change the order 
of integration and use polar coordinates. Then, 
\begin{align*}
& \la \p , G_0 (\lam )u \ra 
= \sum_{j=0}^{\frac{m-3}2}C_j
\int_{\R^m}\overline{\p(x)}
\left(\int_{\R^m}
\frac{\lam^je^{i\lam |y|}u(x-y)}{|y|^{m-2-j}}dy
\right)dx \\
=& \sum_{j=0}^{\frac{m-3}2} C_j 
\int_{\R^m}
\frac{\lam^j e^{i\lam |y|}(\overline{\p} \ast \check u)(y)}
{|y|^{m-2-j}}dy = \sum_{j=0}^{\frac{m-3}2} c_j 
\int_0^{\infty} \lam^j 
e^{i\lam r}r^{1+j} M_{\overline{\p} \ast \check u}(r)dr. 
\end{align*}
Since $M_{\overline{\p} \ast \check u}(r)$ is even, 
change of variable $r$ by $-r$ yields  
\[
-\la \p, G_0 (-\lam)u \ra =
\sum_{j=0}^{\frac{m-3}2} c_j 
\int_{-\infty}^0 \lam^j e^{i\lam r}r^{1+j} 
M_{\overline{\p} \ast \check u}(r)dr. 
\]
Add both sides of last two equations and change  
$r$ by $-r$.  
\end{proof}

\paragraph{Representation formula for even dimensions.} 
If $m$ is even, we have the analogue of 
\refeq(spheric). For a function $M(r)$ on $\R$ and $a>0$, 
define   
\[
M^a(r)=M((1+2a)^{-1}r).
\] 

\begin{lemma} Let $m\geq 2$. Let 
$u,\p \in {C}_0^\infty(\R^m)$. Then
\begin{gather} 
\langle \p ,(G_0(\lam)- G_0(-\lam))u\rangle 
= \sum_{j=0}^\n (-1)^{j+1} T_j^{(a)} \left[
\frac{\lam^j 
\Fg (r^{j+1} M^{a}_{\overline{\p } \ast \check{u}})(\lam)}
{(1+2a)^{j+2}}
\right],   \lbeq(usformula) \\
\mbox{For $j=0$}, \ -T_0^{(a)} 
\left[
\frac{
\Fg(r M^a_{\overline{\p } \ast \check{u}})(\lam)
}
{(1+2a)^2} \right]
= iT_0^{(a)}\left[
\frac{
\lam (\Fg \widetilde{M^a_{\overline{\p } \ast \check{u}}})(\lam)
}
{(1+2a)^2} 
\right]. 
\lbeq(usformula-add)
\end{gather}
\end{lemma}  
\bgpf  Define 
$ B_j(\lam,r,a)= e^{i\lam r(1+2a)}(\lam r)^{j}r^{-(m-2)}$  
and  
\[ 
B_j(\lam,a)u(x)= \int_{\R^m}
B_j(\lam, |y|, a)u(x-y) dy, \quad j=0, \dots, \n.
\]
Then, \refeq(co-kerb) and change of the order of 
integrations 
imply  
\bqn \lbeq(gom) 
\la \p , (G_0(\lam)-G_0(-\lam)) u \ra  =\sum_{j=0}^\n 
\frac{T_j^{(a)}}{\w_{m-1}}
\left[ \la \p , (B_j(\lam,a)-B_j(-\lam,a))u \ra \right].
\eqn 
We have, as in odd dimensions, that 
for $u \in \Sg(\R^m)$ and $\p  \in L^1(\R^m)$ 
\begin{align*}
& \la \p , B_j(\lam,a) u\ra= \int_{\R^{m}}\left(\int_{\R^m}
\overline{\p (x)} B_j(\lam, |y|, a)u(x-y)dy \right)dx \notag \\
& = \int_{\R^m} B_j(\lam, |y|,a) 
(\overline{\p } \ast \check{u})(y)dy 
= {\w}_{m-1}\int_0^\infty e^{i(1+2a)\lam{r}}(\lam{r})^j 
r M_{\overline{\p } \ast \check{u}}(r) dr. 
\end{align*} 
Replacing $\lam$ to $-\lam$ and changing the variable 
$r$ to $-r$, we have 
\[
-\la \p , B_j(-\lam,a) u\ra
= {\w}_{m-1} \int_{-\infty}^0 e^{i(1+2a)\lam{r}}(\lam{r})^j r 
M_{\overline{\p } \ast \check{u}}(r) dr,
\] 
where we used that 
$M_{\overline{\p } \ast \check{u}}(-r)
=M_{\overline{\p } \ast \check{u}}(r)$.
Adding these two yields
\bqn \lbeq(k-3)
\la \p , (B_j(\lam,a)- B_j(-\lam,a)) u\ra 
= {\w}_{m-1} \int_{\R} e^{i(1+2a)\lam r}(\lam r)^j r 
M_{\overline{\p } \ast \check{u}}(r) dr .
\eqn 
Change $r$ to $-r$ in the right of \refeq(k-3), 
plug the result with \refeq(gom) 
and, at the end, change the variable 
$r$ to $-r/(1+2a)$. Then, \refeq(k-3) becomes   
\[ 
\frac{(-1)^{j+1}{\w}_{m-1}}{(1+2a)^{j+2}}
\int_{\R} e^{-i\lam r} \lam^j r^{j+1} 
M^{a}_{\overline{\p } \ast \check{u}}(r) dr 
= \frac{(-1)^{j+1}{\w}_{m-1} \lam^j}{(1+2a)^{j+2}}
\Fg (r^{j+1} M^{a}_{\overline{\p } \ast \check{u}})(\lam)
\]
and \refeq(usformula) follows. 
If we use the first of \refeq(MtildeM-1),   
the right of the last equation for $j=0$ becomes   
\[
\frac{i\lam(\Fg \widetilde{M_{\overline{\p } 
\ast \check{u}}^a})(\lam)}
{(1+2a)^2}{\w}_{m-1}
\]
and we obtain \refeq(usformula-add). 
\end{proof}

\subsection{Some results from harmonic analysis.}  

The following lemma on weighted inequality 
(cf. \cite{Gr}, Chapter 9) plays crucial role in 
this paper.  

\begin{lemma}\lblm(ap) The weight function $|r|^{a}$ is  
an $A_p$ weight on $\R$ if and only 
if $-1< a <p-1$. 
The Hilbert transform $\tilde \Hg$ and the 
Hardy-Littlewood maximal operator 
$\Mg$ are bounded in $L^p(\R, w(r)dr)$ for 
$A_p$ weights $w(r)$.  
\end{lemma}
Modifying the Hilbert transform $\tilde{\Hg}$, we define 
\bqn \lbeq(Hg) 
\Hg u(\rho ) = \frac{(1+\tilde \Hg)u(\rho )}{2} 
= \frac1{2\pi}\int_0^\infty e^{ir\rho } \hat{u}(r)dr. 
\eqn 
We shall repeatedly use following $A_p$ weights 
on $\R^1$ to the operator $\Mg\Hg$:  
\bqn \lbeq(ex-weight)
\mbox{$|r|^{m-1-p(m-1)}$,\ \ $|r|^{m-1-2p}$,\ \ 
$|r|^{m-1-p}$ \ \ and \ \ $|r|^{m-1}$}, 
\eqn 
respectively for $1<p<\frac{m}{m-1}$, 
$\frac{m}3<p<\frac{m}2$, $\frac{m}2<p<m$ and $m<p$.   

For a function $F(x)$ on $\R^m$, we say $G(|x|)\in L^1(\R^m)$ 
is a radial decreasing integrable majorant 
(RDIM for short) of $F$ if $G(r)>0$ is decreasing 
and $|F(x)| \leq G(|x|)$ for a.e. 
$x\in \R^m$. The following lemma is well known 
(see e.g. \cite{Stein}, p.57). 

\bglm \lblm(max) 
\ben
\item[{\rm (1)}]
A rapidly decreasing function $F \in \Sg(\R^m)$ has a RDIM.
\item[{\rm (2)}] If $F$ has a RDIM. then there is a constant 
$C>0$ such that 
\begin{equation}\lbeq(max)
|(F \ast u)(t)| \leq C(\Mg u)(t), 
\quad t \in \R. 
\end{equation}
\een
\edlm

\bglm \lblm(add-cut) 
For $u$ and $F \in L^1(\R)$ such that 
$\hat{u}, \, \hat{F} \in L^1(\R)$ we have 
\bqn \lbeq(add-cut)
\frac1{2\pi}\int_0^\infty e^{i\lam\rho }
F(\lam) \hat{u}(\lam ) d\lam =  (\Fg^\ast{F}\ast \Hg u)(\rho ). 
\eqn 
\edlm 
\bgpf Let $\Th(\lam)=\left\{\br{l} 1, \ \mbox{for}\ \lam>0  \\ 
0, \ \mbox{for}\ \lam\leq 0 \er \right. $.   
Then, the left side of \refeq(add-cut) equals  
\begin{multline*}
\frac1{2\pi}\int_{\R} e^{i\lam\rho }
F(\lam)\Th(\lam)\hat{u}(\lam) d\lam= 
\frac1{2\pi}\int_{\R}\left(
\int_{\R} e^{i\lam(\rho -\xi)}\Fg^\ast{F}(\xi)d\xi
\right) \Th(\lam)\hat{u}(\lam) d\lam\\
=\int_{\R} \Fg^\ast {F}(\xi) 
\Fg^\ast \{\Th(\lam)\hat{u}(\lam)\}(\rho -\xi) d\xi 
=(\Fg^\ast{F} \ast \Hg u)(\rho ) 
\end{multline*} 
as desired. 
\edpf

%%%%%%%%%%%%%%%%%%%%%%%%%%%%%%%%%%%%%%%%%%%%%%%%

\section{Reduction to the low energy analysis} 
We write $W_{-}=W$ in the sequel. 
When $u \in \ax^{-s}L^2$, $s>1/2$, $Wu$ may be expressed 
via the boundary values of resolvents (e.g. \cite{Ku-0}):
\begin{align} 
W u & =u-\lim_{\ep \downarrow 0, N\uparrow \infty}
\frac{1}{\pi i}\int^N_{\ep}  
G(\lam)V(G_0(\lam)-G_0(-\lam))u \lam d\lam
\lbeq(stationary-0) \\
& = u -\frac{1}{\pi i}\int^\infty_{0}  
G(\lam)V(G_0(\lam)-G_0(-\lam))u\lam d\lam 
\lbeq(stationary)
\end{align} 
Here the right of \refeq(stationary-0) is the Riemann integral 
of an $\ax^{t}L^2$-valued continuous function  
where $t>1/2$ is such that $s+t>2$, the result belongs 
to $L^2(\R^m)$ and the limit exists in $L^2(\R^m)$, 
which we symbolically write as \refeq(stationary). 

We decompose $W$ into the high and the low energy parts 
\bqn \lbeq(lh-decom)
W=W_{>} + W_{<} \equiv W \Psi(H_0) + 
W \Phi(H_0),
\eqn 
by using cut off functions $\Phi\in C_0^\infty(\R)$ 
and $\Psi\in C^\infty(\R)$ such that 
\[
\Phi(\lam^2) + \Psi(\lam^2) \equiv 1, 
\quad \Phi(\lam^2)=1 \ \mbox{near $\lam=0$  and 
$\Phi(\lam^2)=0$ for $|\lam|>\lam_0$}
\] 
for a small constant $\lam_0>0$. 
We have proven in previous papers \cite{Y-odd, FY} 
that, under \refass(V),  
$W_{>}$ is bounded in $L^p(\R^m)$ for all 
$1\leq p \leq \infty$ if $m \geq 3$ and we only 
need to study  $W_{<} = \Phi(H_0)+ Z$ where 
\bqn 
Z=  - \frac{1}{\pi i}\int^\infty_{0}  
G(\lam)V(G_0(\lam)-G_0(-\lam))\lam \Phi(H_0) d\lam. 
\lbeq(stationary-l)
\eqn 
Evidently $\Phi(H_0)\in \Bb(L^p(\R^m))$ for all 
$1\leq p \leq \infty$ and we have only to study 
the operator $Z$ defined by \refeq(stationary-l). 
Since $\delta>2$, the LAP 
(cf. Lemma 2.2 of \cite{Y-odd}) implies that $G_0(\lam)V$ 
is a H\"older continuous function of 
$\lam \in \R$ with values in $\Bb_\infty (L^{-s})$ 
for any $\frac12<s<\delta-\frac12$ and, the absence 
of positive eigenvalues (\cite{Kato-e}) implies 
that $1+ G_0(\lam)V$ is invertible for $\lam>0$ 
(cf. \cite{Ag}). It follows from the resolvent equation 
$G(\lam)= G_0(\lam)-G_0(\lam)V G(\lam)$ that $G(\lam)V$ 
may be expressed in terms of $G_0(\lam)V$:  
\bqn \lbeq(res-eq0)
G(\lam)V= 
G_0(\lam)V(1+ G_0(\lam)V)^{-1}\ \mbox{for} \ \lam \not=0  
\eqn 
and it is {\it locally} H\"older continuous for  
$\lam\in \R \setminus\{0\}$ with values in 
$\Bb_\infty (L^2_{-s})$. Thus, we have the expression 
of $Z$ in terms of the free resolvent $G_0(\lam)$:  
\bqn \lbeq(z)
Z u = -\frac{1}{\pi i}\int^\infty_{0}  
G_0(\lam)V(1+ G_0(\lam)V)^{-1}
(G_0(\lam)-G_0(-\lam))\lam F(\lam)u d\lam, 
\eqn 
where $F(\lam)= \Phi(\lam^2)$. If $H$ is of 
generic type, ${\rm Ker}_{L^2_{-s}} (1+ G_0(0)V)=\Ng=\{0\}$ 
for any $\frac12<s<\delta-\frac12$ and 
$1+G_0(\lam)V$ is invertible for $\lam$ in a  
neighborhood $\lam=0$ and both sides of \refeq(res-eq0) 
become H\"older continuous. We then have 
shown in \cite{Y-odd, FY} that $Z$ is bounded in $L^p(\R^m)$ 
for all $1\leq p \leq \infty$ under \refass(V). 

\subsection{Low energy behavior of $(1+ G_0(\lam)V)^{-1}$. }

If $H$ is of exceptional type,   
$(1+ G_0(\lam)V)^{-1}$ becomes singular at $\lam=0$ 
and we describe its singularities here.  Before doing 
so we recall some properties of functions in $\Ng$. 
Recall (\cite{Stein-O}) that for $0<s<m$ : 
\bqn \lbeq(fractional)
|D|^{-s}u(x) = \Fg^\ast (|\xi|^{-s} \hat{u})(x) 
= \frac{\Ga\left(\frac{m-s}2\right)}
{2^s \pi^{\frac{m}2}\Ga\left(\frac{s}2\right)} 
\int_{\R^m}\frac{u(y)}{|x-y|^{m-s}} dy .
\eqn 
When $s=1$ and $s=2$, the constants in front of the integral 
respectively equal to $\pi^{-1}\omega_{m-2}^{-1}$ and 
$C_0=(m-2)^{-1}\w_{m-1}^{-1}$ of \refeq(coker).  

\bglm \lblm(asymp-reei) 
\ben 
\item[{\rm (1)}] 
Functions $\f$ in $\Ng$ satisfy 
$\f \in \ax^{-s}H^2(\R^m)\cap C^1(\R^m)$ 
for any $s>1/2$ and $\nabla\f$ is H\"older continuous. 
They satisfy the following asymptotic expansion as $|x|\to \infty$:  
\begin{multline}
\f(x) = -\frac{C_0}{|x|^{m-2}}\int_{\R^m}(V\f)(y) dy  \\
- \frac{1}{\w_{m-1}}\sum_{j=1}^m 
\frac{x_j}{|x|^{m}}
\int_{\R^m} y_j (V\f)(y) dy  + O(|x|^{-m}). \lbeq(asymp-d) 
\end{multline} 
\item[{\rm (2)}] 
For $\f\in \Ng \setminus\{0\}$,  
$\f \otimes \f \in \Bb(L^p(\R^m))$ for $\frac{m}{m-2}<p<\frac{m}2$ if 
$\f \in \Ng \setminus \Eg_0$ ($m\geq 5$), for $\frac{m}{m-1}<p< m$ 
if $\f \in \Eg_0 \setminus \Eg_1$ and for $1<p< \infty$ if $\f \in \Eg_1$.  
\item[{\rm (3)}] If $\ax^2 u \in L^1(\R^m)$, $|D|^{-1}u(x)$ has the 
following expansion as $|x|\to \infty$: 
\bqn \lbeq(asymp-d1) 
\frac{\int_{\R^m} u dx }{\pi \omega_{m-2}|x|^{m-1}}+ 
\sum_{j=1}^m 
\frac{(m-1)x_j }{\pi \omega_{m-2}|x|^{m+1}}
\int_{\R^m}x_j u dx  + O(|x|^{-m-1}) .
\eqn 
\een
\edlm 
\bgpf (1) The smoothness property of $\f$ is well known 
(see e.g. Corollary 2.6 of \cite{AS}). 
We have from \refeq(fractional) that  
\bqn \lbeq(redfe)
\f(x) = -C_0 \int_{\R^m} \frac{V(y)\f(y)}{|x-y|^{m-2}} dy. 
\eqn
Taylor's formula implies that 
\[
\left|\frac{1}{|x-y|^{m-2}}-\frac1{|x|^{m-2}}
- \frac{(m-2)x\cdot y}{|x|^{m-1}} \right|
\leq C \frac{\ay^2} {\ax^{m}}, \quad |x-y|\geq 1 
\]
and \refeq(asymp-d) follows. Statement 
(2) follows from \refeq(asymp-d). 
We omit the proof of (3) which is similar to that 
of \refeq(asymp-d1).
\edpf 

\subsubsection{Odd dimensional cases}

The structure of singularities depends on $m$. For odd 
dimensions $m \geq 3$ we have the following 
results (see, e.g. Theorem 2.12 of \cite{Y-odd}). 
We state it separately for $m=3$ and $m \geq 5$.  
In the following \refths(III,m=5) for odd $m \geq 3$ and 
\refth(even6) for even $m\geq 6$, we will indiscriminately 
write $E(\lam)$ for the operator valued function of 
$\lam$ defined near $\lam=0$ which, when inserted in 
\refeq(z) for $(1+ G_0(\lam)V)^{-1}$, produces the operator 
which is bounded in $L^p(\R^m)$ for all $1\leq p \leq \infty$.

\paragraph{The case $m =3$.} By virtue of \refeq(asymp-d), 
we have for $\f \in \Ng$ that 
\bqn \lbeq(asym-3)
\f(x) = \frac{L(\f)}{|x|} + O(|x|^{-2}) \ 
\mbox{as $|x|\to \infty$}, \quad 
L(\f) = \frac{-1}{4\pi}\int_{\R^3}V(x)\f(x) dx.
\eqn 
Thus, 
$\Eg= \{\f \in \Ng\setminus\{0\} \colon L(\f)=0\}(=\Eg_0)$ 
and,  as $\Ng \ni \f \mapsto L(\f)\in \C$ 
is continuous, $\dim \Ng/\Eg \leq 1$. Any 
$\ph \in \Ng \setminus \Eg$ 
is called {\it threshold resonance} of $H$. 
We say that $H$ is of exceptional type 
of {\it the first kind} if $\Eg=\{0\}$, 
{\it the second} if $\Eg=\Ng$ and {\it the third kind} if 
$\{0\}\subsetneq \Eg \subsetneq \Ng$. 
We let $D_0, D_1, \dots$ be integral operators defined by 
\[
D_j u(x) = \frac1{4{\pi}j!} \int_{\R^3}|x-y|^{j-1}u(y) dy, 
\quad j=0,1, \dots.
\]
so that we have the formal Taylor expansion  
\[
G_0(\lam)u(x)= \frac1{4\pi}\int_{\R^3}
\frac{e^{i\lam|x-y|}}{|x-y|}u(y)dy 
=\sum_{j=1}^\infty (i\lam)^j D_j u .
\]
If $H$ is of exceptional type of the third kind, 
$-(V\f, \f)$ defines inner product on $\Ng$ and there 
is a unique real $\p\in \Ng$ such that 
\bqn \lbeq(res-eigen-orth)
-(V\p, \f)=0, \quad \forall \f \in \Eg, \quad 
-(V\p,\p)=1 \ \mbox{and}\ L(\p)>0. 
\eqn 
We define {\it the canonical resonance} by 
\bqn \lbeq(canonical)
\ph=\p + P VD_2 V \p\in \Ng.
\eqn 
If $H$ is of exceptional type of the first kind, 
then $\dim\Ng=1$ and there is a 
unique $\ph\in \Ng$ such that $-(V\ph, \ph)=1$ 
and $L(\ph)>0$ and we call this the canonical resonance. 
We have the following result for $m=3$ 
(see e.g. \cite{Y-odd}). 

\begin{theorem}\lbth(III) Let $m=3$ 
and let $V$ satisfy $|V(x)|\leq C\ax^{-\delta}$ 
for some $\delta>3$. Suppose that $H$ is of exceptional 
type of the third kind and let $\ph$ be the canonical 
resonance and $a = 4\pi i |\la V,\ph \ra|^{-2}$. Then:  
\bqn \lbeq(res)
(I+G_0(\lam)V)^{-1}
=\frac{PV}{\lam^2} + i\frac{PV D_3VPV}{\lam}  
- \frac{a}{\lam}|\varphi\ra \la \varphi| V  + E(\lam). 
\eqn 
If $H$ is of exceptional type of 
the first or the second kind, \refeq(res) 
holds with $P=0$ or $\ph=0$ respectively. 
\end{theorem}

\paragraph{The case $m \geq 5$.} 
If $m\geq 5$, \refeq(asymp-d) implies $\Ng=\Eg$. 

\begin{theorem}\lbth(m=5) Let $m\geq 5$ be odd and  
$|V(x)|\leq C\ax^{-\delta}$ for some $\delta>m+3$.  
Suppose $H$ is of exceptional type. Then:
\ben 
\item[{\rm (1)}] If $m=5$ then,  
with $\ph =P V$, $V$ being considered as a function, 
\bqn 
(I+G_0(\lam)V)^{-1}
=\frac{PV}{\lam^2} - 
\frac{a_0}{\lam}|\varphi\ra \la \varphi| V  + E(\lam) , \quad 
a_0=\frac{i}{24\pi^2}.
\lbeq(res-5)
\eqn 
\item[{\rm (2)}] If $m\geq 7$ then 
\bqn 
(I+G_0(\lam)V)^{-1} = \frac{PV}{\lam^2} + E(\lam). 
\lbeq(res-7)
\eqn  
\een
\end{theorem}
Define $S(\lam)=(I+G_0(\lam)V)^{-1}-E(\lam)$ and 
\bqn 
Z_{s}  = \frac{i}{\pi}\int^\infty_0 
G_0( \lam)VS(\lam)(G_0(\lam)-G_0(-\lam))F(\lam)\lam 
d\lam. 
\lbeq(zs) 
\eqn 
Then, it follows from \refths(III,m=5) that 
$Z-Z_s\in \Bb(L^p(\R^m))$ for all $1\leq p \leq \infty$ 
and we have only to study $Z_s$ in what follows.  

\subsection{Even dimensional case} 
When $m$ is even, singular terms of 
$(1+ G_0(\lam)V)^{-1}$ may contain 
logarithmic factors. The following is the improvement of 
Proposition 3.6 of \cite{FY}. We let $\dim \Eg=d$ and 
$\{\f_1, \dots, \f_d\}$ be the real orthonormal basis of $\Eg$. 
For making the expression simpler, we state the theorem 
for $V(1+ G_0(\lam)V)^{-1}$. 

\bgth \lbth(even6)
Let $m\geq 6$ be even. Suppose 
$|V(x)|\leq C \ax^{-\delta}$ for $\delta>m+4$ if $m=6$ 
and for $\delta>m+3$ if $m\geq 8$. Let $\ph=P V$ with $V$ 
being considered as a function. Then, we have the 
following statements for $\Im \lam \geq 0$ and 
$\log \lam$ such that $\log \lam \in \R$ for 
$\lam>0$: 
\ben 
\item[{\rm (1)}] If $m=6$ then, we have that   
\begin{multline}\lbeq(sing-1) 
V(1+ G_0(\lam)V)^{-1}= \frac{VP V}{\lam^2} + 
\frac{\w_{5}}{(2\pi)^6} \log \lam  (V\ph\otimes V\ph) \\
+ \left(\frac{\w_{5}\|\ph\|}{(2\pi)^6}\right)^2 
\lam^2 \log^2 \lam (V\ph \otimes V\ph) 
+ \lam^2 \log \lam F_2 + VE(\lam), 
\end{multline}
where $F_2$ is an operator of rank at most $8$ such that 
\bqn \lbeq(sing-ef) 
F_2 = \sum_{a,b=1}^{8} \ph_a \otimes \p_b, 
\quad \ph_a,\ \p_b \in (L^1 \cap L^\infty)(\R^6).  
\eqn 
\item[{\rm (2)}] If $m\geq 8$, then we have with a constant 
$c_m$ that  
\begin{equation}\lbeq(sing-2) 
V(1+ G_0(\lam)V)^{-1}= 
\frac{VP V}{\lam^2} + 
c_m (V\ph \otimes V\ph) \lam^{m-6} \log \lam + VE(\lam). 
\end{equation}
\item[{\rm (3)}] If $m \geq 12$, then 
$c_m (V\ph \otimes V\ph) \lam^{m-6} \log \lam$ 
of \refeq(sing-2) may 
be included in $VE(\lam)$. 
\een
\edth 
\bgpf We prove (1) only, using the notation of the proof 
of subsection 3.2.1 of \cite{FY}. A slightly more careful 
look at the argument there shows that, in spite of 
Eqn.(3.5) of \cite{FY}, $V(1+ G_0(\lam)V)^{-1}$ is actually 
given by 
\bqn 
\frac{VPV}{\lam^2} + VD_{01} \log \lam + 
VD_{21} \lam^2 \log \lam + VD_{22} \lam^2\log^2\lam + VE(\lam) .
\eqn 
Here, with $F_{jk}= F_{jk}(0)$, $F_{jk}(\lam)$ being 
defined by (3.16) of \cite{FY}, and 
$A(0)=(2\pi)^{-6}\w_{m-1}(1\otimes 1)$, $VD_{01}$ 
and $VD_{22}$ are rank $1$ operators given by  
\begin{gather} 
VD_{01}= VP VF_{01} P V 
= VP VA(0)VP V= \frac{\w_{m-1}}{(2\pi)^6}(V\ph \otimes V\ph), 
\lbeq(ta-1g)
\\
VD_{22} = V(P VF_{01})^2 P V = V(P V A(0)VP)^2 V 
= \frac{\w_{m-1}^2}{(2\pi)^{12}} \|\ph\|^2(V\ph \otimes V\ph),  
\notag 
\end{gather}
where we used $PVQ=PV$ and$VQP=VP$ and,  
\begin{align}
VD_{21} & = VPV(F_{21} + F_{00}PVF_{01}
+ F_{01}PVF_{00})PV  \lbeq(dst) \\
& - VX(0)\overline{Q}D_2 V P V F_{01}P V 
- VX(0)\overline{Q}A(0)VPV  \lbeq(dst0)\\
& - VP V F_{01}P V Q D_2 V \overline{Q}X(0) 
- P VQ A(0)V\overline{Q}X(0). \lbeq(dst1)
\end{align}
It is obvious that the first line \refeq(dst) 
is of rank at most $\min(4, d)$ and of the form 
$\sum \alpha_{jk} (V\f_j \otimes V\f_k)$; four other 
operators are of rank one and of the form $f \otimes g$ 
with $f\, \in (L^1\cap L^\infty)(\R^6)$. We check this 
for $VX(0)\overline{Q}D_2 V P V F_{01}P V$ as a prototype.
We have $D_2=D_0^2$ and $D_0V\ph = -\ph $. Thus, 
\refeq(ta-1g) implies  
\[
VX(0)\overline{Q}D_2 V P V F_{01}P V
= - (2\pi)^{-6}\w_{m-1}
(VX(0)\overline{Q}D_0 \ph)\otimes (V\ph). 
\]
Here $D_0 \ph \in C^2(\R^6)$ 
and satisfies $D_0 \ph \absleq C\ax^{-2}$  
by virtue of \reflm(asymp-reei). Hence, a fortiori 
$D_0\ph \in C_0(\R^6)$, the Banach space of continuous 
functions which converge to $0$ as $|x|\to \infty$.  
It is obvious that $\Xg\equiv \overline{Q}C_0(\R^6) \subset C_0(\R^6)$ and 
$X(0)=N^{-1}(0)=[\overline{Q}(1+D_0V)\overline{Q}]^{-1}$ is an isomorphism 
of $\Xg$. This is because $T=\overline{Q}D_0 V\overline{Q}$ 
is compact both in $\Xg=\overline{Q}C_0(\R^6)$ and 
$\Yg=\overline{Q}L^2_{-\delta+2}(\R^6)$, $\Xg \cap\Yg$ is dense 
in $\Yg$ and ${\rm Ker}_{\Yg}(1+ T)=\{0\}$ (see e.g. Lemma 2. 11 of \cite{GY}). 
Thus, $VX(0)\overline{Q}D_0 \ph(x) \absleq C \ax^{-\delta}$. 
\end{proof}

It follows from \refth(even6) that $Z u= Z_{s}u + Z_{\log}u$ 
modulo the operator which is bounded in $L^p$ for all 
$1\leq p \leq \infty$ 
and we need study  
\begin{align}
Z_{se}& = \frac{i}{\pi}\int^\infty_0 G_0( \lam)
VPV(G_0(\lam)-G_0(-\lam))F(\lam)\lam^{-1} d\lam, \lbeq(e-zs) 
\\
Z_{\log}& = \sum_{j,k}\frac{i}{\pi}\int^\infty_0 
G_0(\lam)\lam^{2j} (\log \lam)^k D_{jk}
(G_0(\lam)-G_0(-\lam))F(\lam)\lam d\lam,  \lbeq(loge) 
\end{align} 
for even $m \geq 6$, where the sum and $D_{jk}$ are  
as in \refth(even6). 

\section{Proof of \refthb(theo)} 

The proof of \refth(theo) for $m=3$ is the simplest 
and is the prototype for other dimensions and, most of 
the basic ideas already appear here.  

\subsection{The case of exceptional type of the first kind} \lbssec(4-1)
We begin with the case that $H$ is of exceptional type of 
the first kind and, we let $\ph$ be the canonical resonance, 
$a= 4\pi i|\la V, \ph\ra|^{-2}\not=0$ and    
\bqn \lbeq(dpi)
\p(x)=|D|^{-1}(V\ph)(x)
= \frac1{2\pi^2} \int \frac{V(y)\ph(y)}{|x-y|^2} dy.
\eqn 
The following lemma proves \refth(theo) when 
$H$ is of exceptional type of the first kind.   

\bglm \lblm(3d-first) 
{\rm (1)} For $1<p<3$, there exists a constant 
$C_p$ such that  
\bqn 
\|Z_{s}u\|_p \leq C_p  \|u\|_p, \quad u\in C_0^\infty(\R^3) . 
\lbeq(3d-first) 
\eqn 
{\rm (2)} For $3<p<\infty$, there exists a constant 
$C_p$ such that  
\bqn \lbeq(3d-1-2)
\|(Z_{s}+ a \ph \otimes \p) u \|_p \leq C_p \|u\|_p, 
\quad u \in C_0^\infty(\R^3). 
\eqn 
{\rm (3)} For $p$ outside $1<p<3$, 
$Z_s$ is unbounded in $L^p(\R^3)$.  
\edlm 

\bgpf Recall $c_0= C_0 \w_2=1$. In this case 
$S(\lam)= - \frac{a}{\lam}|\ph\ra \la \ph| V$ and  
\bqn \lbeq(f-k1)
Z_{s}u= -\frac{ia}{\pi}\int^\infty_{0}  
G_0(\lam)V\ph\ra 
\la V\ph |(G_0(\lam)-G_0(-\lam))u \ra F(\lam) d\lam. 
\eqn 
Defining $M(r)= M(r,(V\ph)\ast \check u)$, we substitute 
\refeq(coker) and \refeq(spheric) respectively for $G_0(\lam)$ 
and $\la V\ph |(G_0(\lam)-G_0(-\lam))u \ra$. Then,  
\[ 
Z_{s} u = \frac{ai}{\pi}\int^\infty_{0}
\left(\int_{\R^3}
\
\frac{e^{i\lam |x-y|}V(y)\ph(y)}{4\pi|x-y|}dy
\right)    
\left(\int_{\R} e^{-i\lam r}r  
M(r)dr\right) 
F(\lam) d\lam. 
\]
If we change the order of integrations,  
\begin{gather}
Z_{s}u= \frac{ ai}{2\pi}\int_{\R^3}
\frac{K_{0}( |x-y|)V(y)\ph(y)}{|x-y|}dy, \lbeq(3d-1) \\
K_0(\rho)= \frac1{2\pi}\int_0^\infty 
e^{i\lam\rho} F(\lam) 
\left(\int_{\R} e^{-ir\lam}rM(r)dr \right)d\lam.
\lbeq(k0-def)
\end{gather}
Since $\Fg^\ast F \in \Sg(\R)$, it follows 
by virtue of \reflms(max,add-cut) that  
\bqn \lbeq(1s-1)
K_0(\rho)=\{(\Fg^\ast F)\ast \Hg(rM(r))\}(\rho)
\absleq C \Mg \Hg (rM)(\rho).
\eqn 
Function $K_0(\rho)$ may also be expressed as 
\bqn \lbeq(parts-10)
K_0(\rho)=\frac{i}{2\pi\rho }\int_0^\infty e^{i\lam\rho } 
\left(F(\lam)\int_\R e^{-ir\lam} rM (r) dr \right)' d\lam .
\eqn 
and, after integration by parts, we see that $K_0(\rho)$ 
satisfies also 
\bqn 
K_0(\rho )\absleq C \rho^{-1}(\Mg\Hg(r^2 M)(\rho ) 
+\Mg\Hg(r M)(\rho )). 
\lbeq(3d-3)
\eqn 
The boundary term does not appear in \refeq(parts-10) 
since $\int_{\R} rM(r) dr =0$. 

\noindent 
(1a) Let $3/2<p<3$. By virtue of Young's inequality   
\bqn \lbeq(Yo-1)
\|Z_{s}u\|_p \leq 
\frac{|a|(4\pi)^{1/p}}{2\pi}\|V\ph\|_1 
\left(\int_0^\infty \left|
\frac{K_0(\rho)}{\rho}\right|^p \rho^{2} d\rho\right)^{1/p}. 
\eqn 
We estimate $K_0(\rho)$ by 
\refeq(1s-1) and use that 
$\rho^{2-p}$ is an $A_p$ weight on $\R$. \reflm(ap) 
and Young's inequality imply     
\begin{align}
& \left(\int_0^\infty 
\left|
\frac{K_0(\rho)}{\rho}\right|^p \rho^{2} d\rho\right)^{1/p} 
\leq 
C \left(\int_0^\infty |\Mg \Hg (rM)(\rho)|^p 
\rho^{2-p} d\rho\right)^{1/p}  \notag \\
& \leq C_p \left(\int_0^\infty M(r)^p r^2 dr\right)^{1/p}
\leq C_p \|V\ph \ast u\|_p \leq C_p \|V\ph\|_1 \|u\|_p .
\lbeq(3d-b1)
\end{align}
and $\|Z_s u\|_p \leq C_p \|V\ph\|_1^2 \|u\|_p$. 

\noindent 
(1b) For $1<p<\frac32$, we use estimate 
\refeq(3d-3) and that $\rho^{2-2p}$ is an $A_p$ 
weight on $\R$ and obtain that    
\begin{align}
& \left(\int_0^\infty 
\left|
\frac{K_0(\rho)}{\rho}\right|^p \rho^{2} 
d\rho\right)^{\frac1{p}}\leq 
\left(\int_0^\infty 
|(\Mg\Hg(r^2 M)+ \Mg\Hg(r M))(\rho)|^p
\rho^{2-2p}d\rho \right)^{\frac1{p}} \notag \\
& \leq C \left(\int_0^\infty 
|M(r)|^p \max(r^{2},r^{2-p})dr\right)^{\frac1{p}} 
\leq C(\|V\ph\|_1+ \|V\ph\|_{p'})\|u\|_p, 
\lbeq(3dim-fir)
\end{align}
where we estimated the integral over 
$0\leq r \leq 1$ by using that 
\bqn  
\sup |M(r)| \leq \|V\ph \ast u\|_\infty\leq \|V\ph\|_{p'}
\|u\|_p.  \lbeq(supM)
\eqn 
Thus, we have 
$\|Z_su\|_p \leq C(\|V\ph\|_1+ \|V\ph\|_{p'})\|V\ph\|_1\|u\|_p$ 
for $1<p<3/2$. Combining (1a) and (1b), we obtain 
\refeq(3d-first) for $1<p<3$ by interpolation(\cite{BL}).  

\noindent 
(2) Let $p>3$. Writing $\int_{\R} r e^{-ir\lam}M(r)dr 
=i \left(\int_{\R} e^{-ir\lam}M(r)dr \right)'$ in 
\refeq(k0-def), we apply integration by parts and 
obtain yet another expression of $K_0(\rho)$:     
\bqn 
K_0(\rho) = \frac{-i}{2\pi}
\int_{\R}M(r)dr 
-\frac{i}{2\pi}\int_0^\infty 
\left(e^{ i\lam\rho} F(\lam) \right)' 
\left(\int_{\R} e^{-ir\lam}M(r)dr \right) d\lam.  \lbeq(k0-def-a)
\eqn  
Denote the second term by $\tilde{K}_{0}(\rho)$. 
By virtue of \reflms(max,add-cut),  
\bqn \lbeq(est-tK0)
\tilde{K}_{0}(\rho)\absleq C (\rho +1) \Mg \Hg(M)(\rho).
\eqn 
Substituting 
\refeq(k0-def-a) for $K_0(\rho)$ in \refeq(3d-1), we obtain  
$Z_{s}u = Z_{b} u + Z_{i} u$, 
where $Z_b$ and $Z_i$ are operators produced by  
$\frac{-i}{2\pi}
\int_{\R}M(r)dr$ and $\tilde{K}_{0}(\rho)$, respectively. 
Because  
\bqn 
\frac{1}{\pi}\int_{\R}M(r)dr
=\frac1{2\pi^2} \int_{\R^3}\left(\int_{\R^3}
\frac{(V\ph)(x+y)}{|x|^2} dx \right) u(y)dy=\la \psi, u\ra 
\lbeq(D2phi)
\eqn 
by the definition \refeq(dpi), we have by 
using \refeq(redfe) for $m=3$ that 
\bqn \lbeq(b-3d2)
 Z_{b}u (x) = \frac{a}{4\pi^2}\int_{\R}M(r)dr \cdot 
\int_{\R^3}\frac{V(y)\ph(y)}{|x-y|}dy 
= -a |\ph\ra \la \p | u\ra .
\eqn 
We splite the integral as 
\bqn \lbeq(zi-3df)
Z_i u(x)=\frac{ai}{2\pi}\left(\int_{|y|\leq 1}
+ \int_{|y|>1} \right) 
\frac{\tilde{K}_0(|y|)(V\ph)(x-y)}{|y|}dy 
= I_1(x) + I_2(x).
\eqn 
For estimating $I_2$ we use \refeq(est-tK0) for $\rho\geq 1$: 
$|\tilde{K}_0(\rho)|\leq C\rho \Mg\Hg(M)(\rho)$. 
Since $\rho^2$ is an $A_p$-weight on $\R$ 
for $p>3$, we have by using Young's and H\"older's 
inequalities and \reflm(ap) that
\begin{align}
& \|I_2\|_p  \leq C \|V\ph\|_1 
\left(\int_0^\infty |\Mg\Hg(M)(\rho)|^p 
\rho^2 d\rho\right)^{\frac1{p}} \notag \\
& \leq C \|V\ph\|_1 
\left(\int_0^\infty|M(r)|^p r^2dr\right)^{\frac1{p}}
\leq C \|V\ph\|_1^2 \|u\|_p. \lbeq(3d-se-3d)
\end{align}
H\"older's inequality implies, with $p'=p/p-1$, that 
\[
|I_1 (x)|
\leq C 
\left(\int_{|y|\leq 1}
\left|\frac{(V\ph)(x-y)}{|y|}\right|^{p'} dy\right)^{1/p'}
\left(\int_0^1 |\tilde{K_0}(\rho)|^p \rho^2 d\rho \right)^{1/p} .
\]
Since $\tilde{K_0}(\rho)\absleq C \Mg\Hg(M)(\rho)$ for 
$0<\rho <1$ by virtue of \refeq(est-tK0)  and since $\rho^2$ is an 
$A_p$-weight, we obtain as in \refeq(3d-se-3d) that 
\bqn 
\left(\int_0^1 |\tilde{K_0}(\rho)|^p 
\rho^2 d\rho \right)^{1/p} 
\leq C 
\left(\int_0^\infty |\Mg\Hg(M)(\rho)|^p 
\rho^2 d\rho \right)^{1/p} \leq C \|u\|_p .
\lbeq(3d-se-3d+)
\eqn 
It follows by virtue of Minkowski's inequality that 
\bqn %[begin{align}
\|I_1\|_p  %& 
\leq C \|u\|_p \left\|\left(\int_{|y|\leq 1}
\left|\frac{(V\ph)(x-y)}{|y|}\right|^{p'} 
dy\right)^{1/p'}\right\|_p  %\notag \\& 
\leq C \|u\|_p \|V\ph\|_p % \int_{|y|\leq 1}|y|^{-p'} dy 
\lbeq(3d-flas)
\eqn %nd{align}
because $1<p'<3/2<3<p<\infty$. Thus,  
\[
\left\|\int_{\R^3}
\frac{\tilde{K}_0(|x-y|)V(y)\ph(y)}{|x-y|}dy \right\|_p 
\leq 
C (\|V\ph\|_p + \|V\ph\|_1) \|u\|_p. 
\]
With \refeq(b-3d2) this proves \refeq(3d-1-2). 

\noindent 
(3) It is well-known that $W$ is 
unbounded in $L^1(\R^3)$ in this case 
(\cite{Y-disp,ES}), hence so is $Z_s$. 
Since $\int_{\R^3}V\ph dx \not=0$, \reflm(asymp-reei) 
implies that $\ph \not\in L^p(\R^3)$ for 
$1\leq p\leq 3$ and that $\p\in L^p(\R^3)^\ast$ if and 
only if $p>3$. Hence, $\ph \otimes \p$ is 
unbounded in $L^p(\R^3)$ for any $1\leq p\leq\infty$. 
Thus, statement (2) implies that $Z_s$ is unbounded in 
$L^p(\R^3)$ for $p\geq 3$. This completes the proof of 
the lemma. 
\edpf 

We review here the basic strategy of this 
subsection as it will be repeatedly employed 
in the following (sub)sections. We express 
$Z_s u $ as the convolution 
\refeq(3d-1) of $V\ph$ and $K_0(\rho)$ of 
\refeq(k0-def). 
By applying integration by parts if necessary 
we represent and estimate $K_0(\rho)$ as in 
\refeq(1s-1), \refeq(3d-3) or \refeq(est-tK0) 
by using $\Mg\Hg$. These estimates are used 
for proving  
\bqn \lbeq(lpnorm)
\left(\int_0^\infty \left|K_0(\rho)\right|^p 
\rho^{2-p} d\rho \right)^{\frac1{p}}
\left(= \w_2^{-\frac1{p}}
\left\|\frac{K_0(|x|)}{|x|}\right\|_p \right)
\leq C \|u\|_p 
\eqn 
via the weighted inequality for $\frac32<p<3$, 
$1<p<\frac32$ and $p>3$ respectively. 
Desired estimates are then obtained by combining 
\refeq(lpnorm) and Young's inequality. However, 
the boundary term appears in the integration by 
parts for large values of $p>3$ which obstructs 
the $L^p$-boundedness. We represent the obstruction  
explicitly in terms of functions of $\Ng$ and 
show that $L^p$-boundedness depends on the properties 
of functions in $\Ng$. Suitable modifications, 
improvements and additional arguments will of course 
become necessary in what follows, which have already 
been appeared the simplest case of this subsection.

\subsection{The cases of the second and third kinds}

Let $H$ be of exceptional type of the second kind. Then,    
\bqn \lbeq(s2)
S(\lam)= \frac{P V}{\lam^2} + i \frac{P V D_3 V P V}{\lam},
\eqn 
where $D_3$ is the integral operator with 
kernel $|x-y|^2/4\pi$. We take the real orthonormal 
basis $\{\f_1, \dots, \f_n\}$ of $\Eg$  and define 
$a_{jk}=\pi^{-1}\la \f_j |VD_3V|\f_k\ra \in \R$. 
We have $\la V, \f_j \ra=0$, $1\leq j \leq n$. 
Substituting \refeq(s2) for $S(\lam)$ in \refeq(zs), we have  
\begin{gather}
Z_{s}u= Z_{s0}u + Z_{s1}u = \sum_{j,k=1}^n Z_{s0,jk}u
+ \sum_{j=1}^n Z_{s1,j}, \\
\lbeq(d3-zsem-0)
Z_{s0,jk} u = i a_{jk} \int^\infty_{0}  
G_0(\lam)V\f_j \ra 
\la V\f_k |(G_0(\lam)-G_0(-\lam))u \ra F(\lam) d\lam, \\
\lbeq(d3-zsem-1)
Z_{s1,j}u = \frac{i}{\pi}\int^\infty_{0}  
G_0(\lam)V\f_j \ra 
\la V\f_j |(G_0(\lam)-G_0(-\lam))u \ra F(\lam) 
\frac{d\lam}{\lam}.
\end{gather} 

\bglm \lblm(3d-second-0) For any $1<p<\infty$, there 
exists a constant $C_p$ such that  
\bqn 
\|Z_{s0}u\|_p \leq C_p \|u\|_p, \quad u\in C_0^\infty(\R^3).  
\lbeq(3d-second-0) 
\eqn 
\edlm 
\bgpf The operator $Z_{s0,jk}$ is equal to $Z_s$ of 
\refeq(f-k1) with two $\ph \in \Ng$'s being replaced 
by $\f_j$ and $\f_k \in \Eg$ and $a$ by $-{\pi}a_{jk}$. 
Thus, the proof of \reflm(3d-first) implies that 
$Z_{s0,jk}\in \Bb(L^p(\R^3))$ for $1<p<3$ and that  
\bqn \lbeq(3d-sec-d)
Z_{s0,jk} -{\pi} a_{jk} \f_j \otimes  |D|^{-1}(V\f_k) 
\in \Bb(L^p(\R^3)), \quad p>3. 
\eqn 
Here $\f_j \otimes  |D|^{-1}(V\f_k)$ is 
bounded in $L^p(\R^3)$ for $p>3$ because $\f_j\in L^p(\R^3)$ 
and $|D|^{-1}(V\f_k)\in (L^p(\R^3))^\ast$ by virtue of 
\refeq(asymp-d) and \refeq(asymp-d1).  
Thus $Z_{s0,jk} \in \Bb(L^p(\R^3))$ for $3<p$ and, hence, for 
$1<p<\infty$ by interpolation. This proves the lemma.   
\edpf 

\bglm \lblm(3d-second-1)  
\ben
\item[{\rm (1)}] Let $1<p<3$. Then, for a constant 
$C_p$, we have 
\bqn \lbeq(z1s-pbdd)
\|Z_{s1}u \|_p \leq C_p \|u\|_p,
\quad u \in C_0^\infty(\R^3). 
\eqn 
\item[{\rm (2)}] Let $3<p<\infty$. Then, for a constant 
$C_p$, we have 
\bqn \lbeq(3d-2-est)
\|(Z_{s1}+P)u\|_p \leq C \|u\|_p, 
\quad u \in C_0^\infty(\R^3). 
\eqn 
In \refeq(3d-2-est) $P$ may be replaced by $P\ominus P_1$ 
by virtue of \reflm(asymp-reei).  
\item[{\rm (3)}] The operator 
$Z_{s1}$ is bounded in $L^p(\R^3)$ 
for some $p>3$ if and only if $\Eg=\Eg_1$. In this case 
$Z_{s1}$ is bounded in $L^p(\R^3)$ for all $1<p<\infty$. 
\een
\edlm 
\bgpf Define $\p_j(x) =|D|^{-1} (V\f_j)(x)$,  
$j=1, \dots, n$. Then \reflm(mulfo) implies   
\bqn \lbeq(sec-1)
Z_{s1,j}u = \frac{i}{\pi}\int_0^\infty 
G_0(\lam)|V\f_j\ra 
\la \p_j |(G_0(\lam) -G_0(-\lam))u\ra F(\lam)d\lam 
\eqn 
which can be obtained from $Z_{s}u$ 
of \refeq(f-k1) by replacing $a$ by $-1$, the first 
$V\ph$ by $V\f_j$ and the second by $\p_j$. Thus, 
it may be expressed by using 
$K_{0,j}(\rho)$ of \refeq(k0-def) with $M(r)$ being replaced by 
$M_j(r)=M(r, \p_j \ast \check{u})$:
\bqn 
Z_{s1,j}u = \frac{1}{2\pi{i}}\int_{\R^3}
\frac{K_{0,j}(|x-y|)V(y)\f_j(y)}{|x-y|}dy.  \lbeq(3d-21)
\eqn

\noindent 
(1) The argument of (1a) in the proof of 
\reflm(3d-first) implies     
\bqn  \lbeq(3d-h)
\|Z_{s1,j}u\|_p \leq C \|V\f_j\|_1 \|\p_j \ast u\|_p, \quad 3/2<p<3
\eqn 
(see \refeq(3d-b1)) and the one of (1b) does 
\bqn  
\|{Z}_{s1,j}u\|_p \leq C \|V\f_j\|_1 
(\|\p_j \ast u\|_p + \|\p_j \ast u\|_\infty), \quad 1<p<3/2 \lbeq(3d-i)
\eqn  
(see \refeq(3dim-fir)).  
Since $\int V\f_j dx=0$, \refeq(asymp-d1) implies 
that $\p_j=|D|^{-1}\f_j \in L^q(\R^3)$ 
for all $1<q\leq \infty$ and that the convolution 
operator with $\p_j(x)$ is bounded in $L^p$ for any 
$1<p<\infty$ via Calder\'on-Zygmund theory (see e.g. 
\cite{Stein}, pp. 30-36).  
Thus, $\|\p_j \ast u\|_p \leq C \|u\|_p$,   
$\|\p_j \ast u\|_\infty \leq \|\p_j\|_{p'}\|u\|_p$ 
and ${Z}_{s1,j}$ is bounded in $L^p(\R^3)$ for all 
$1<p<3$, $j=1, \dots, n$. Statement (1) follows. 

\noindent 
(2) Integration by parts as in \refeq(k0-def-a) by 
using the identity 
$\int_{\R} e^{-ir\lam}rM_j(r)dr= 
i \left(\int_{\R} e^{-ir\lam}M_j(r)dr \right)'$
implies that $K_{0,j}(\rho)$ may be written as 
\bqn 
 -\frac{i}{2\pi}
\int_{\R}M_j(r)dr 
-\frac{i}{2\pi}\int_0^\infty 
\left(e^{i\lam\rho} F(\lam) \right)' 
\left(\int_{\R} e^{-ir\lam}M_j(r)dr \right) d\lam, 
\lbeq(k0-def-jb)
\eqn 
which we insert into \refeq(3d-21). 
Since $|D|^{-1}\p_j= (-\lap)^{-1}(V\f_j)= - \f_j$, 
\refeq(D2phi) with  $\p_j \in \Eg$ in place of 
$V\ph$ produces $-\la \f_j |u \ra $. It follows 
that the boundary term of \refeq(k0-def-jb) produces   
\bqn \lbeq(pr-bdry)
\frac{-1}{4\pi}\int_{\R^3}\frac{V(y)\f_j(y)}
{|x-y|}dy\cdot \frac1{\pi}\int_{\R}M_j(r)dr 
=- |\f_j \ra \la \f_j |u \ra 
\eqn 
as in \refeq(b-3d2). 
Denote by $\tilde{K}_{0j}(\rho)$ and $\tilde{Z}_{s1,j}$ 
the second term of \refeq(k0-def-jb) and the operator it 
produces via \refeq(3d-21). They can respectively be 
obtained from $\tilde K_0(\rho)$ of \refeq(k0-def-a) 
and $Z_i$ of \refeq(zi-3df) by replacing $M(r)$ 
and $\tilde K_0(\rho)$ by $M_j(r)$ and $\tilde{K}_{0,j}(\rho)$. 
Thus, the argument of step (2) of the proof of \reflm(3d-first), 
\refeq(3d-se-3d) and \refeq(3d-flas) in particular, 
implies that    
\bqn \lbeq(proid)
\|\tilde{Z}_{s1,j}u \|_p \leq 
C (\|V\f_j \|_p + \|V\f_j\|_1)\|\p_j \ast u\|_p, \quad 3<p<\infty.
\eqn 
The Calder\'on-Zygmund theory with \refeq(asymp-d1) 
once more implies  $\left\|\tilde{Z}_{s1,j}u \right\|_p 
\leq C \|u\|_p$.  
Since $\f \otimes \f \in \Bb(L^p)$ 
for all $1<p<\infty$ if $\f \in \Eg_1$ 
by virtue of \refeq(asymp-d), this together with 
\refeq(pr-bdry) proves statement (2). 

\noindent 
(3) It is obvious from (1) and (2) that $Z_{s1} \in \Bb(L^p(\R^3))$ 
for all $1<p<\infty$ if $\Eg=\Eg_1$. Suppose then that 
$Z_{s1} \in \Bb(L^p(\R^3))$ for some $p>3$ then $P\ominus P_1$ 
must be bounded in $L^p(\R^3)$ by virtue of (2). Take the 
orthonormal basis $\{\f_1, \dots, \f_d\}$ of $\Eg\ominus \Eg_1$  
and $\{\rho_1, \dots, \rho_d\} \subset C_0^\infty(\R^3)$ 
such that $\{(\rho_j, \f_k)\}$ becomes the unit matrix. 
Then, $(P\ominus P_1)\rho_j= \f_j$, $j=1, \dots, n$  and, if 
$P\ominus P_1$ is bounded in $L^p(\R^3)$ for some $p\geq 3$, there 
must exist a constant $C>0$ such that 
\[
|(u,\f_j)|= |((P\ominus P_1) u, \rho_j)| \leq C_j \|u\|_p, 
\quad \mbox{for all} \ u \in C_0^\infty(\R^3).
\]
Then, $\f_j$ has to be in $L^{p'}(\R^3)$ for $p'\leq 3/2$ 
for all $j=1, \dots, n$. This implies  $\f_j=0$ by virtue 
of \refeq(asymp-d). Thus, $\Eg=\Eg_1$ must hold. This completes 
the proof. 
\edpf 

\reflm(3d-second-0) and \reflm(3d-second-1) prove \refth(theo) when 
$H$ is of exceptional type of the second kind. The following lemma 
completes the proof of \refth(theo).

\bglm \lblm(third) 
Suppose that $H$ is of exceptional type of the third kind. 
Then:
\ben
\item[{\rm (1)}] 
$W $ is bounded in $L^p(\R^3)$ for all $1<p<3$. 
\item[{\rm (2)}]
$W + a\ph\otimes (|D|^{-1}V\ph) + P$ 
is bounded in $L^p(\R^3)$ for all $p>3$. 
\item[{\rm (3)}] 
$W$ is unbounded in $L^p(\R^3)$ for any $p>3$ and $p=1$.  
\een
\edlm 
\bgpf 
The combination of \reflmss(3d-first,3d-second-0,3d-second-1) 
proves statements (1) and (2). Suppose that $W$ 
is bounded in $L^p(\R^3)$ 
for some $3<p<\infty$. Then, so is 
$a(\ph\otimes (|D|^{-1}V\ph))+ P$. 
Let $\p\in \Ng$ be the function which defines the 
canonical resonance $\ph$ by \refeq(canonical) 
and which satisfies \refeq(res-eigen-orth). Then,  
\[
(V\p, a(\ph\otimes (|D|^{-1}V\ph))u + Pu)
= - a(|D|^{-1}V\ph, u) , \quad u \in C_0^\infty(\R^3) 
\]
and this must be extended to a bounded functional of 
$u \in L^p(\R^3)$. Hence, $|D|^{-1}V\ph \in L^q(\R^3)$ 
for $q=(p-1)/p <3/2$. This contradicts \refeq(asymp-d) 
because $\int_{\R^3}V(x)\ph(x)dx \not=0$ and (3) is proved. 
\edpf

\section{Proof of \refthsb(theo5,theo6) for odd $m$}

If $m\geq 5$, then $\Ng= \Eg$ and we let 
$\{\f_1, \dots, \f_d\}$ be the real orthonormal basis of 
$\Eg$. \refth(m=5) implies that, with $a_0=i/(24\pi^2)$,     
\bqn  \lbeq(si-57)
S(\lam)= \left\{
\br{ll} 
\lam^{-2}P V -a_0 \lam^{-1}(\ph \otimes V\ph), \quad 
& \mbox{if $m=5$},  \\
\lam^{-2} P V,  \ & \mbox{if $m \geq 7$}. \er \right.
\eqn 
Note that $\ph\not=0$ if and only if $\Eg_1 \not=\Eg$. 
We substitute \refeq(si-57) for $S(\lam)$ in \refeq(zs) and 
apply \refeq(coker) and \refeq(spheric) as previously. 
Let $C_j, c_k$, $1\leq j,k\leq \frac{m-3}2$ respectively 
be constants of \refeq(coker) and \refeq(spheric). 
Then, we have 
\bqn \lbeq(wa) 
Z_s u= Z_{s0}u + Z_{s1}u, 
\eqn 
where $Z_{s0}=0$ for $m\geq 7$ and, for $m=5$, with 
$M(r)=M(r, V\ph \ast \check{u})$  
\begin{gather}
Z_{s0}u
= -2ia_0 \sum_{j,k=0,1}(-1)^{j+1}C_k c_j {Z}_{s0}^{jk}u, 
\lbeq(5dd) 
\\
Z_{s0}^{jk}u(x) = \int_{\R^5} \frac{V\ph (y)}{|x-y|^{3-k}}
K^{(j,k)}_0(|x-y|)dy, \lbeq(5d-s0) \\
K^{(j,k)}_0(\rho )= \frac1{2\pi}
\int_0^\infty e^{i\lam\rho }\lam^{j+k}
\left(\int_{\R} e^{-i\lam r} r^{j+1} M(r)dr
\right)F(\lam)d\lam,
\lbeq(Kjk0-def)
\end{gather}
and $Z_{s1}u$ is defined for all $m\geq 5$ by 
\bqn \lbeq(z5s1)
Z_{s1}u =\sum_{l=1}^d Z_{s1}(\f_l)u
\eqn 
where, for $\f\in \Eg$, with $M(r)=M(r, V\f \ast \check{u})$, 
\begin{gather} 
{Z}_{s1}(\f)u= 2i \sum_{j,k=0}^{\frac{m-3}2} 
(-1)^{j+1}C_k c_j {Z}_{s1}^{jk}(\f) ,\lbeq(z5s1-a) \\ 
{Z}_{s1}^{jk}(\f)u(x)=\int_{\R^m} 
\frac{V\f (y)}{|x-y|^{m-2-k}}K^{(j,k)}(|x-y|)dy, 
\lbeq(ws1jk) \\
K^{(j,k)}(\rho )= \frac1{2\pi}
\int_0^\infty e^{i\lam\rho }\lam^{j+k-1}
\left(
\int_{\R} e^{-i\lam r} r^{j+1} M(r)dr
\right)F(\lam)d\lam.  \lbeq(Kjk-def)
\end{gather}
Note that $Z_{s0}^{jk}u$ and $K^{(j,k)}_0(\rho)$ 
are obtained from $Z_{s1}^{jk}u$ and $K^{(j,k)}(\rho)$ 
by changing $\f$ by $\ph$ and $\lam^{j+k-1}$ by 
$\lam^{j+k}$ in \refeq(Kjk-def).  
 
We shall prove the last 
statements of (2) and (3) of \refths(theo5,theo6)  
only for $Z_{s1}(\f)$  since the argument of the 
proof of (3) of \reflm(3d-second-1) 
can easily be adapted for proving the same statements 
for $Z_{s1}$.

\subsection{Estimate of $Z_{s0}$ for $m=5$}

We begin by proving the following lemma for $Z_{s0}$, 
assuming $\ph\not=0$. 
\bglm 
\lblm(5-7) \ben 
\item[{\rm (1)}] $Z_{s0}$ is bounded in $L^p(\R^5)$ for 
$1<p<5$.
\item[{\rm (2)}] $Z_{s0}+a_0 | \ph \ra \la |D|^{-1}(V\ph)|$ 
is bounded in $L^p(\R^5)$ for $5/2<p<\infty$. 
\item[{\rm (3)}] $Z_{s0}$ is not bounded in $L^p(\R^5)$ 
if $p\geq 5$. 
\een
\edlm 
\bgpf For $\ph=PV$, we have 
$\int_{\R^5} V\ph dx = \|\ph\|^2 >0$ 
and, by virtue of \refeq(asymp-d) and 
\refeq(asymp-d1), 
$\ph \otimes |D|^{-1}(V\ph)\in \Bb(L^p(\R^5))$ 
if and only if $5/3<p<5$. Hence, statement (3) 
follows (2). 
Using that  
$e^{i\rho\lam}= (i\rho)^{-(k+1)}\pa_\lam^{k+1}
e^{i\rho\lam}$ and 
$\int_{\R} \lam^{j+1} e^{-i\lam r} M(r)dr
=i^{j}\left(\int_{\R} e^{-i\lam r} M(r)dr\right)^{(j)}$, 
we apply integration by parts to \refeq(Kjk0-def) 
and write $K^{(j,k)}_0(\rho )$ in two ways  
\begin{align}
K^{(j,k)}_0(\rho )& = \frac{i^{k+1}}{2\pi\rho^{k+1}}
\int_0^\infty e^{i\rho\lam}
\left(\lam^{j+k}F(\lam) \int_{\R} 
e^{-i\lam r} r^{j+1} M(r)dr \right)^{(k+1)}d\lam  
\lbeq(parts-1) 
\\
& = \frac{(-i)^{j}}{2\pi}
\int_0^\infty \left(e^{i\lam\rho } \lam^{j+k}F(\lam)
\right)^{(j)} 
\left(\int_{\R} e^{-i\lam r}r M(r)dr\right)
d\lam.  \lbeq(parts-2a) 
\end{align}
Note that boundary terms do not appear 
in \refeq(parts-1) since 
$\int_{\R} r M(r) dr=0$ and, if $k=1$, we may apply 
further integration by parts to \refeq(parts-2a) 
without having boundary term and 
\bqn 
K^{(j,k)}_0(\rho )=\frac{(-i)^{j+1}}{2\pi}
\int_0^\infty \left(e^{i\lam\rho } \lam^{j+k}F(\lam)
\right)^{(j+1)} 
\left(\int_{\R} e^{-i\lam r} M(r)dr\right)
d\lam.  \lbeq(parts-2) 
\eqn 
We then apply \reflms(max,add-cut) to the right 
sides and obtain the following estimates for $j,k=0,1$: 
\begin{numcases} 
{K^{(j,k)}_0 (\rho) \absleq }
C\rho^{-(k+1)} \sum_{l=0}^{k+1}\Mg\Hg(r^{j+l+1}M)(\rho), 
\lbeq(lkjk-1) \\
C(1+\rho^{j+k})\Mg\Hg(r^{1-k}M)(\rho). 
\lbeq(lkjk-2) 
\end{numcases} 

\noindent 
(a) Let $1<p<5/4$. 
Since $|r|^{-4(p-1)}$ is an $A_p$ weight on $\R$ and $3p-4>-1$, 
we have by using \refeq(lkjk-1) and \refeq(supM) that, 
for any $j,k=0,1$,    
\begin{align}
& \left\|\frac{K^{(j,k)}_0 (|y|)}{|y|^{3-k}}\right\|_p 
\leq C  
\sum_{l=0}^{k+1}
\left(\int_0^\infty \frac{|\Mg\Hg(r^{j+l+1}M)(\rho)|^p}
{\rho^{4(p-1)}}  d\rho\right)^{1/{p}} 
\notag \\ 
& \leq C \left(
\int_{0}^1 
\frac{|M(r)|^pdr}{r^{3p-4}} + \int^\infty_1 |M(r)|^p r^4dr 
\right)^{\frac1{p}} \leq C(\|V\ph\|_{p'}+ \|V\ph\|_{1})\|u\|_p .
\lbeq(4-49a)
\end{align}
Young's inequality then implies  
$\|Z_{s0}^{jk}u\|_p \leq C\|V\ph\|_1 
(\|V\ph\|_{p'}+ \|V\ph\|_{1})\|u\|_p$. 

\noindent 
(b) We next show that $\|Z_{s0}^{j1}u\|_p \leq C \|u\|_p$ 
for $p>5$ and $j=0,1$. 
Interpolating this with the result of (a), we then have  
the same for all $1<p<\infty$.  
We split the integral as in \refeq(zi-3df) 
and repeat the argument after it: 
\[
|Z_{0s}^{j1}u(x)|\leq 
C\left(\int_{|y|\leq 1} + \int_{|y|>1}\right) 
\frac{|V\ph(x-y)|}{|y|^{2}}
|K^{(j,1)}_0 (|y|)|dy = I_1(x) + I_2(x). 
\] 
For $\rho\geq 1$, we have 
$K^{(j,1)}_0(\rho ) \absleq C \rho^2 \Mg \Hg (M(r))(\rho)$ 
by virtue of \refeq(lkjk-2) and  
since $r^4$ is $A_p$ weight on $\R$ if $p>5$. It follows 
that   
\begin{align}
\|I_2\|_p & \leq C\|V\ph\|_1 
\left\| \frac{K^{(j,1)}_0}{|x|^2}\right\|_{L^p(|x|\geq 1)} 
\leq C \|V\ph\|_1\left(\int_0^\infty 
|\Mg \Hg (M)(\rho)|^p \rho^{4}d\rho\right)^{\frac1{p}}  \notag \\
& \leq C \|V\ph\|_1 \left(\int_0^\infty 
|M(r)|^p r^{4}dr\right)^{1/p} 
\leq C \|V\ph\|_1^2 \|u\|_p. \lbeq(efg) 
\end{align}
H\"older's inequality and 
\refeq(lkjk-2) for $0\leq \rho \leq 1$, 
$K^{(j,1)}_0(\rho ) \absleq C \Mg \Hg (M)(\rho)$ imply  
\[ 
|I_1(x)|  
\leq C\left(\int_{|y|\leq 1} 
\left|\frac{|V\ph(x-y)|}{|y|^2}\right|^{p'}dy 
\right)^{\frac1{p'}} 
\left(\int_0^1  
|\Mg \Hg (M(r))(\rho)|^p \rho^{4}d\rho\right)^{\frac1{p}}.
\]  
Since $p'\leq \frac{5}{4}$ if $p>5$, 
Minkowski's inequality and \refeq(efg) imply 
\bqn \lbeq(efg-2)
\|I_1 \|_p \leq C \|V\ph\|_1 \|V\ph\|_p \|u\|_p .
\eqn 

\noindent 
(c) We finally prove 
$-2ia_0 C_0(c_1 {Z}_{s0}^{10}-c_0 Z_{s0}^{00}) 
+ a_0 |\ph \ra \la |D|^{-1}(V\ph)| \in \Bb(L^p(\R^5))$ 
for $p>5/2$. This will complete the proof of the lemma. 
Indeed, by virtue of \refeq(5dd), this and (b) clearly 
imply statement (2); 
since $| \ph \ra \la |D|^{-1}(V\ph)|$ is bounded 
in $L^p(\R^5)$ for $5/3<p<5$ as remarked previously, 
this also implies 
$-2ia_0 C_0(c_1Z_{s0}^{10}-c_0 Z_{s0}^{00})\in \Bb(L^p(\R^5))$ 
for $5/3<p<5$ and, hence, for $1<p<5$ by virtue 
of result (a) and interpolation. Then, (b) 
yields statement (1). If $k=0$, further integration by parts 
to \refeq(parts-2a) produces 
boundary term:
\begin{multline} \lbeq(4.58)
K^{(j,0)}_0(\rho ) = 
\frac{(-i)^{j+1}}{2\pi}
j! \int_{\R} M(r)dr \\ 
+ \frac{(-i)^{j+1}}{2\pi}
\int_0^\infty \left(e^{i\lam\rho } \lam^{j}F(\lam)
\right)^{(j+1)} 
\left(\int_{\R} e^{-i\lam r} M(r)dr\right) d\lam.
\end{multline}
The second integral, which we denote by 
$\tilde{K}^{(j,0)}_0(\rho)$, satisfies  
\bqn \lbeq(p5p52)
\tilde{K}^{(j,0)}_0(\rho) \absleq 
C (1+ \rho^{j+1}) \Mg\Hg(M)(\rho)
\leq C (1+ \rho^{j+2}) \Mg\Hg(M)(\rho) 
\eqn 
and we estimate 
the operator $\tilde{Z}^{j0}$ 
obtained by replacing $K^{(j,0)}_0(\rho)$ 
by $\tilde{K}^{(j,0)}_0(\rho)$ in \refeq(5d-s0) 
by repeating the argument of step (b):  
Split $\tilde{Z}^{j0}u(x)$  as in there 
and obtain  
$\|I_2\|_p\leq C \|u\|_p$ for $5/2<p<5$ (resp. $p>5$) 
by using the first (resp. second) estimate 
of \refeq(p5p52) and that $r^{4-p}$ (resp. $r^4$) is 
an $A_p$-weight on $\R$. Likewise we obtain 
$\|I_1\|_p \leq C \|u\|_p$ for $5/2<p<5$ 
(resp. $p>5$) by first applying H\"older's inequality 
by considering the integrand as 
$(|V\ph(x-y)|/|y|^2)\cdot (|\tilde{K}^{(j,0)}_0(|y|)|/|y|)$ 
(resp. $|V\ph(x-y)|/|y|^3  \cdot |\tilde{K}^{(j,0)}_0(|y|)|$) 
and then using Minkowski's inequality. Thus, we have 
for $j=0,1$ that 
\bqn 
\|\tilde{Z}^{j0} u\|_p \leq C \|u\|_p , \quad 5/2<p<\infty.
\eqn 
The contribution of boundary terms of \refeq(4.58) to  
$c_0 K_0^{(00)}- c_1 K_0^{(10)}$ is given 
by virtue of \refeq(c0c1) and \refeq(asymp-d1) by 
\begin{multline*}
(c_1-ic_0) \times \frac1 {2\pi}\int_{\R} M(r)dr=
\frac{c_0}{\pi{i}}\int_{\R} M(r)dr 
=-4\pi^2 C_0 i \la |D|^{-1}(V\ph), u\ra 
\end{multline*} 
and this contributes to  
$2a_0 i C_0(c_0Z_{s0}^{00}-c_1 Z_{s0}^{10})u(x)$ by   
\[
8\pi^2 a_0 C_0^2 \int_{\R^5}\frac{V\ph(y)}{|x-y|^3}dy 
\cdot 
(\la |D|^{-1}(V\ph), u\ra ) 
= -a_0 \ph(x) \la |D|^{-1}(V\ph), u \ra,  
\]
where we used $8\pi^2 C_0=1$ when $m=5$. 
This proves the lemma. \edpf 

\subsection{Estimates of $Z_{s1}$ for $m \geq 5$.}

We next study $Z_{s1}u$ for all $m\geq 7$.  
By virtue of \refeq(z5s1) and \refeq(z5s1-a) 
and the remark at the beginning of section 5, it suffices 
to study $Z_{1s}^{jk}(\f)u$ defined by \refeq(ws1jk) 
for $\f\in \Eg$. For simplifying notation, we often omit 
$\f$ from $Z_{1s}^{jk}(\f)$. Define  
\bqn \lbeq(m-ast)
M_\ast(r)=M(r, |D|^{-1}(V\f) \ast \check{u}).
\eqn 
Then, by virtue of \refeq(devi-by), 
$K^{(j,k)}(\rho)$ may also be expressed as 
\bqn  \lbeq(kjk-mast)
K^{(j,k)}(\rho)= \frac{1}{2\pi}
\int_0^\infty e^{i\lam\rho} \lam^{j+k} F(\lam) 
\left(\int_{\R} e^{-i\lam r} r^{j+1} 
M_\ast (r)dr\right)d\lam  
\eqn 
which has the larger factor $\lam^{k+j}$ than $\lam^{k+j-1}$ 
of \refeq(Kjk-def). 
We omit the proof of the following lemma which is 
essentially the same as that of \refeqs(lkjk-1,lkjk-2)   
\bglm 
$K^{(j,k)}(\rho )$ satisfies the following estimates: 
\begin{numcases}  
{K^{(j,k)}(\rho) \absleq } 
C \rho^{-k-1}\sum_{l=0}^{k+1} \Mg\Hg(r^{j+1+l}M)(\rho), 
& $j \geq 2$.    \lbeq(jk-1) \\
C (1+ \rho^{j-1})\Mg\Hg(r^2 M)(\rho), & $j \geq 1$. 
\lbeq(jk-0) \\[7pt] 
C(1+ \rho^{j}) \Mg\Hg(rM)(\rho),   
& $k+j \geq 1 $.  \lbeq(jk-3) \\[7pt]
C(1+ \rho^{j+1}) \Mg\Hg(M)(\rho),  
& $k \geq 2$. \lbeq(jk-2) \\[7pt]
C(1+ \rho^{j}) \Mg\Hg(rM_\ast)(\rho),   
& $k \geq 0$. \lbeq(jk-4) 
\end{numcases} 
\edlm

\bglm \lblm(lemma)
Suppose $m \geq 5$ and $\f \in \Eg$. Then:
\ben 
\item[{\rm (1)}] If $j \geq 2$, $Z_{1s}^{jk}(\f)$, 
$k=0,\dots, \frac{m-3}{2}$, are bounded in 
$L^p(\R^m)$ for $1<p<\frac{m}{2}$.  
\item[{\rm (2)}] For $k \geq 2$, 
$Z_{1s}^{jk}(\f)$, $j=0,\dots, \frac{m-3}{2}$, 
are bounded in $L^p(\R^m)$ for $\frac{m}3<p$. 
\item[{\rm (3)}] For all $j,k$, $Z_{1s}^{jk}(\f)$ is bounded 
in $L^p(\R^m)$ for $\frac{m}3<p<\frac{m}2$. 
\een 
If both $j,k \geq 2$, $Z_{1s}^{jk}(\f)$ is bounded 
in $L^p(\R^m)$ for all $1<p<\infty$. 
\edlm 
\bgpf (a) We first prove (1) for $1<p<\frac{m}{m-1}$. 
General case follows from this and (3) by interpolation.  
We use \refeq(jk-1) and that 
$r^{-(m-1)(p-1)}$ is an $A_p$ weight on $\R$ for 
$1<p<\frac{m}{m-1}$. 
Then, estimating as in \refeq(4-49a), we obtain   
\begin{align} 
& \|Z_{s1}^{jk} u\|_p 
\leq  C \|V\f\|_1 \left(\int_0^\infty |M(r)|^p r^{m-1}dr 
+ \int_0^1 \frac{|M(r)|^p}{r^{(m-4)p}} r^{m-1}dr \right)^{1/p}
\notag   \\
& \leq C\|V\f\|_1 (\|V\f\|_1 + \|V\f\|_{p'})\|u\|_p. \lbeq(const-ab)
\end{align}

\noindent 
(b) We next prove (2) for $p>m$. General case then follows 
from this and (3) by interpolation. We split the 
integral as in \refeq(zi-3df):  
\[
Z_{s1}^{jk}u(x) 
\absleq 
\left(\int_{|y|\leq 1} +\int_{|y|\geq 1}\right)
\frac{|V\f(x-y)|}{|y|^{m-2-k}}|K^{(j,k)} (|y|)|dy 
= I_1(x) + I_2(x). 
\]
Using \refeq(jk-2) for $\rho\geq 1$ and that  
$r^{m-1}$ is $A_p$ weight on $\R$ if $p>m$, we obtain  
\bqn 
\|I_2\|_p \leq C \|V\f\|_1 \left(\int_1^\infty 
|\Mg \Hg (M)(\rho)|^p \rho^{m-1}d\rho\right)^{\frac1{p}} 
\leq C \|V\f\|_1^2 \|u\|_p. 
\lbeq(efg-0) 
\eqn 
H\"older's inequality and 
\refeq(jk-2) for $0\leq \rho \leq 1$ imply that 
\bqn 
|I_1(x)| \leq \left(\int_{|y|\leq 1} 
\left|\frac{V\f(x-y)|}{|y|^{m-2-k}}\right|^{p'}dy \right)^{1/{p'}} 
\left(\int_0^1 
|\Mg \Hg (M)(\rho)|^p \rho^{m-1}d\rho\right)^{1/p}. 
\lbeq(efg-1)
\eqn  
Then, Minkowski's inequality and the estimate as in 
\refeq(efg-0) yield  
\[
\|I_1 \|_p 
\leq C \|V\f\|_1 \|u\|_p 
\left(\int_{|x|<1} 
\frac{\|V\f\|_p^{p'} dx}{|x|^{(m-2-k)p'}}\right)^{1/p'}
\leq  C \|V\f\|_1\|V\f\|_p \|u\|_p 
\]
because $p'\leq \frac{m}{m-1}$ if $p>m$ and 
$|y|^{-(m-2-k)p'}$ is integrable over $|y|\leq 1$. 
Thus, statement (2) for $p>m$ follows. 

\noindent  
(c) We prove statement (3) by modifying the argument 
in step (b). Let $\frac{m}3<p<\frac{m}2$. Then,  
$r^{m-1-2p}$ is an $A_p$ weight on $\R$.  
We split  the integral of $Z^{jk}_{s1}u(x)$ as in step (b). \\
(i) Let $j \geq 1$. Estimate \refeq(jk-0) for $\rho\geq 1$ 
and \reflm(ap) yield 
\bqn  
\|I_2\|_p \leq C \|V\f\|_1 \left(\int_0^\infty 
|\Mg \Hg (r^2 M)(\rho)|^p \rho^{m-1-2p}d\rho\right)^{\frac1{p}} 
\leq C \|V\f\|_1^2 \|u\|_p. 
\lbeq(efg-01) 
\eqn 
Estimate \refeq(jk-0) for $\rho\leq 1$ and H\"older's 
inequality imply    
\[ 
|I_1(x)| \leq \left(\int_{|y|\leq 1} 
\left|\frac{|V\f(x-y)|}{|y|^{m-4-k}}\right|^{p'}dy 
\right)^{\frac1{p'}} 
\left(\int_0^1 |\Mg \Hg (r^2 M)(\rho)|^p 
\rho^{m-1-2p}d\rho\right)^{\frac1{p}}. 
\]  
Minkowski's inequality and the second estimate of 
\refeq(efg-01) 
imply $\|I_1\|_p \leq C \|V\f\|_p \|V\f\|_1 \|u\|_p$ 
as previously and, hence, 
$\|Z_{s1}^{jk}u\|_p \leq C\|u\|_p$. 

\noindent
(b) Let $j=0$. Express $K^{(0,k)}(\rho )$ by using 
$\tilde{M}(r)$ of \refeq(tiM) and estimate as    
\bqn \lbeq(mtm-12)
K^{(0,k)}(\rho )= \frac1{2i\pi}
\int_0^\infty e^{i\lam\rho }\lam^{k}
\left(\int_{\R} e^{-i\lam r} \tilde{M}(r)dr
\right)F(\lam)d\lam
\absleq  C\Mg \Hg(\tilde{M})(\rho).
\eqn 
Since $\rho^{-(m-2-k)}\leq \rho^{-2}$ for $\rho\geq 1$,  
Young's inequality, \reflm(ap) and Hardy's inequality  
yield   
\begin{align} 
& \|I_2\|_p \leq C \|V\f\|_1 
\left(\int_0^\infty 
|\tilde{M}(r)|^p r^{m-1-2p}dr\right)^{1/p} \notag \\
& \leq C \|V\f\|_1 
\left(\int_0^\infty 
|M(r)|^p r^{m-1}dr\right)^{1/p} 
\leq C \|V\f\|_1 \|V\f\|_p \|u\|_p. \lbeq(4-100)
\end{align}
H\"older's inequality and \refeq(mtm-12) imply  
\[
|I_1(x)| \leq \left(\int_{|y|\leq 1} 
\left|\frac{|V\f(x-y)|}{|y|^{m-4-k}}\right|^{p'}dy 
\right)^{1/{p'}} 
\left(\int_0^1 
|\Mg \Hg (\tilde{M})(\rho)|^p \rho^{m-1-2p}d\rho\right)^{1/p} 
\]
Estimate the second factor by \refeq(4-100) 
and use Minkowski's equality. This yields  
$\|I_1\|_p \leq C \|V\f\|_p \|V\f\|_1 \|u\|_p$.
The last statement follows from (1) and (2) by 
interpolation. 
\edpf

\bglm \lblm(joint) 
Let $m \geq 5$ and  $\f\in \Eg$. Then: 
\ben 
\item[{\rm (1)}]  For $1<p<\frac{m}{2}$,  
$\|(c_0 Z_{s1}^{(0,k)}-c_1 Z_{s1}^{(1,k)})u\|_p \leq C \|u\|_p$ 
for all $0 \leq k \leq \tfrac{m-3}2$.
\item[{\rm (2)}] 
The operator $Z_{s1}(\f)$ is bounded in $L^p(\R^m)$ 
for $1<p<\frac{m}2$. 
\een
\edlm 
\bgpf It suffices to prove the estimate of (1) for 
$1<p<\frac{m}{m-1}$ since that for $1<p<\frac{m}2$ 
follows from this and \reflm(lemma) (3) by interpolation 
and since statement (2) follows from this and 
statement (1) of \reflm(lemma). Using the identity  
$e^{i\lam\rho}=(i\rho)^{-k-1}\pa_{\lam}^{k+1}e^{i\lam\rho}$, 
we apply integration by parts $k+1$ times to the integral 
of \refeq(mtm-12) and use the identity \refeq(MtildeM-1). 
We obtain 
\begin{align}
& K^{(0,k)}(\rho)=\frac{i^k}{2\pi\rho^{k+1}}
\left(k! \int_{\R}r^2 M(r)dr \right. \notag \\
& \quad \left. + 
\sum_{l=0}^{k+1}
\begin{pmatrix} k+1 \\ l \end{pmatrix}
\int_0^\infty e^{i\lam\rho}
(\lam^{k}F)^{(k+1-l)} 
\int_{\R} e^{-i\lam r}(-ir)^{l} \tilde{M}dr
d\lam \right).   \lbeq(K0k) 
\end{align}
Integration by parts $k+1$ times  to 
$K^{(1,k)}(\rho)$ of \refeq(Kjk-def) likewise yields  
\begin{align}
& K^{(1,k)}(\rho)= \frac{i^k}{2\pi{i}\rho^{k+1}}
\left(-k! \int_{\R}r^2 M(r)dr \right. 
\notag 
\\
& \left. - 
\sum_{l=0}^{k+1}
\begin{pmatrix} k+1 \\ l \end{pmatrix}
\int_0^\infty e^{i\lam\rho}
(\lam^{k}F)^{(k+1-l)} 
\int_{\R} e^{-i\lam r}(-ir)^{l}r^2 M dr
d\lam \right) . \lbeq(K1k)
\end{align}
Since $c_0-ic_1=0$, the boundary terms of \refeq(K0k) 
and \refeq(K1k) cancel and  
\[ 
\frac{c_0 K^{(0,k)}(\rho) - c_1 K^{(1,k)}(\rho)}
{\rho^{m-2-k}} 
\absleq \frac{C}{\rho^{m-1}}
\sum_{l=0}^{k+1}(\Mg \Hg(r^{l}\tilde{M})(\rho)+ 
\Mg \Hg(r^{l+2} M)(\rho)). \lbeq(coK)  
\]
For $1<p<\frac{m}{m-1}$, $\rho^{-(m-1)(p-1)}$ is an $A_p$-weight 
on $\R$. It follows by Young's inequality, 
\reflm(ap) and Hardy's inequality that 
$\|(c_0 Z^{(0,k)}-c_1 Z^{(1,k)})u\|_p$ is bounded by 
$C\|V\f\|_1$ times       
\begin{align} 
& \sum_{l=0}^{k+1}  
\left(\int_0^\infty 
(|\tilde M(r)|^p r^{pl} + 
|M(r)|^p r^{p(l+2)}) r^{m-1-p(m-1)}dr\right)^{1/p} 
\lbeq(9th)
\\
& \quad \leq C  
\left(\int_0^1 \frac{|M(r)|^p}{r^{p(m-3)}} r^{m-1}dr 
+ \int_0^\infty |M(r)|^p r^{m-1} dr \right)^{1/p} 
\lbeq(91th) 
\\
& \quad\leq  C(\|V\f\|_{p'}+ \|V\f\|_p)\|u\|_p.
\lbeq(W0ka)
\end{align}
Here we used $k+3\leq m-1$ for $m\geq 5$ in the first step 
and $p(m-1)<m$ in the last. 
This proves the estimate of (1) for $1<p<\frac{m}{m-1}$. 
\edpf 

\reflm(5-7) and the second statement of 
\reflm(joint) prove statement (1) 
of \refths(theo5,theo6) for odd $m$. The following 
lemma (and \reflm(5-7) for the case $m=5$) proves 
statement (2) of these theorems for odd $m$. 

\bglm \lblm(m/2m) 
Let $m\geq 5$, $\f \in \Eg$ and $\frac{m}2<p<m$.  
Then, for a constant $C>0$, 
\bqn \lbeq(m/2m)
\left\|Z_{s1}(\f)u+\frac{\Ga\left(\frac{m-2}{2}\right)}
{\sqrt{\pi}\Ga\left(\frac{m-1}2\right)}  
\la u, \f\ra \f \right\|_p \leq C \|u\|_p. 
\eqn 
If $Z_{s1}(\f) \in \Bb(L^p)$ for some $\frac{m}2<p<m$, 
then $\f\in \Eg_0$ and $Z_{s1}(\f) \in \Bb(L^p)$ for 
all $1<p<m$. 
\edlm 
\bgpf Let $j+k\geq 1$. Since $m-2-(k+j)\geq 1$, we have from 
\refeq(jk-3) that 
\[
\frac{K^{(j,k)}(\rho)}{\rho^{m-2-k}} 
\absleq C \left(\frac1{\rho^{m-2-k}}
+\frac1{\rho}\right)\Mg \Hg (r M)(\rho)
\]
Using that $r^{m-1-p}$ is $A_p$ weight and 
$(m-2)p'<m$ for $m/2<p<m$, 
we repeat the argument of the step (b) or (c) of the 
proof of \reflm(lemma) and obtain 
\bqn 
\|Z^{jk}_{s1}u\|_p \leq C \|u\|_p, \quad j+k\geq 1. 
\eqn 
It remains to consider $-2i C_0 c_0 Z^{00}_{s1}$, 
see \refeq(z5s1-a). We apply integration by parts 
to \refeq(kjk-mast) with $j=k=0$ 
\begin{align} 
& K^{(0,0)}(\rho)= \frac{i}{2\pi}
\int_0^\infty e^{i\lam\rho} F(\lam) 
\pa_{\lam}\left(\int_{\R} e^{-i\lam r}  M_\ast (r)dr\right)d\lam \notag 
\\
& = \frac{-i}{2\pi}\int_{\R} M_\ast (r)dr
- \frac{i}{2\pi}
\int_0^\infty (e^{i\lam\rho} F(\lam))' 
\left(\int_{\R} e^{-i\lam r} M_\ast (r)dr\right)d\lam.
\lbeq(lef)
\end{align} 
We denote the second integral of \refeq(lef) and the 
operator produced by inserting it into \refeq(ws1jk) 
with $j=k=0$ in place of $K^{(0,0)}(\rho)$ by 
$K_\ast^{(0,0)}(\rho)$ and $Z_\ast^{00}$ respectively. 
We have $|K_\ast^{(0,0)}(\rho)|
\leq C(1+\rho)\Mg\Hg(M_\ast)(\rho)$. 
Decompose   
\[
Z_\ast^{00}u(x)\absleq  
\left(\int_{|y|\leq 1}+ \int_{|y|\geq 1}\right) 
|(V\f)(x-y)|\frac{|K_\ast^{(0,0)}(|y|)|}{|y|^{m-2}}dy 
= I_1(x)+ I_2(x)
\] 
as previously. For estimating $\|I_2\|_p$, define 
$1/q=1/p-1/m$ and apply Young's inequality, 
H\"older's inequality, \reflm(ap) noticing that $q>m$ and 
$r^{m-1}$ is $A_q$ weight and, Hardy-Littlewood-Soblev 
inequality recalling that 
$|D|^{-1}(V\f)\ast(x)\absleq C \ax^{1-m}$. 
We obtain   
\begin{align}
& \|I_2\|_p 
\leq C \|V\f\|_1 \left(\int_0^\infty
|\Mg\Hg(M_\ast)(\rho)|^q \rho^{m-1} d\rho\right)^{1/q}
\left\|\frac{1}{|y|^{m-3}}\right\|_{L^{m}(|y|>1)} \notag \\
& \leq C \|V\f\|_1 \||D|^{-1}(V\f)\ast \check{u}\|_q
\leq C\|V\f\|_1 \||D|^{-1}(V\f)\|_{\frac{m}{m-1},w} \|u\|_p.
\lbeq(deco-8) 
\end{align} 
For $I_1(x)$, H\"older's inequality implies  
\[
|I_1(x)|\leq C \left(\int_{|y|\leq 1}
\left|\frac{(V\f)(x-y)}{|y|^{m-2}}\right|^{q'}dy\right)^{1/q'} 
\left(\int_{|y|\leq 1}
|\Mg\Hg(M_\ast)(|y|)|^{q}dy\right)^{1/q}.
\]
The second factor on the right is bounded by 
$C\||D|^{-1}(V\f)\|_{\frac{m}{m-1},w}\|u\|_p$ 
as in \refeq(deco-8) and $q'<\frac{m}{m-1}<\frac{m}2<p$. 
It follows by Minkowski's inequality that 
\[
\|I_1\|_p \leq C \|V\f\|_p \|u\|_p \left(\int_{|y|\leq 1}
\frac{dy}{|y|^{(m-2)q'}}\right)^{1/q'}
\leq C \|V\f\|_p \|u\|_p. 
\]
Thus, we have $\|Z^{00}_\ast u\|_p \leq C \|u\|_p$ 
for $\frac{m}2<p<m$. The boundary term of 
\refeq(lef) is, by virtue of \refeq(fractional) and 
that $c_0=(m-2)^{-1}$, equal to  
\begin{align}
& \frac{-i}{2\pi}\int_{\R} M_\ast (r)dr
= \frac{-i}{\pi\omega_{m-1}} 
\int_{\R^m}\left(\int_{\R^m}
\frac{|D|^{-1}(V\f)(y)}{|x-y|^{m-1}} dy 
\right)u(x) dx \notag 
\\
& = \frac{-i\Ga\left(\frac{m}{2}\right)}
{\sqrt{\pi}\Ga\left(\frac{m-1}2\right)}  
\int_{\R^m} |D|^{-2}(V\f)(x)u(x) dx
= \frac{i\Ga\left(\frac{m}{2}\right)}
{\sqrt{\pi}\Ga\left(\frac{m-1}2\right)}\la \f, u \ra. 
\lbeq(bd-m)
\end{align}
Inserting this into the right of \refeq(ws1jk) for 
$j=k=0$, we see the contribution of the boundary 
term to $Z_{s1}(\f)u$ is given by 
\[
\frac{2c_0 C_0 \Ga\left(\frac{m}{2}\right)}
{\sqrt{\pi}\Ga\left(\frac{m-1}2\right)}
\int_{\R^m}\frac{V\f(y)}{|x-y|^{m-2}}dy 
\la \f, u \ra
= - \frac{\Ga\left(\frac{m-2}{2}\right)}
{\sqrt{\pi}\Ga\left(\frac{m-1}2\right)} |\f \ra \la \f, u\ra.
\] 
This proves the first statement. If 
$Z_{s1}\in \Bb(L^p)$ for some $\frac{m}2<p<m$, \refeq(m/2m) implies 
$\f \otimes \f \in \Bb(L^p)$ for this $p$. Then, \refeq(asymp-d) 
implies that $\f$ must satisfy $\la \f, V\ra=0$ and 
and $\f \otimes \f \in \Bb(L^p)$ for all $\frac{m}{m-1}<p <m$. 
Then, $Z_{s1}\in \Bb(L^p)$ must be satisfied for all 
$\frac{m}2<p<m$ and, hence, for all $1<p<m$ by \reflm(joint) 
and interpolation. 
\edpf 

We finally study $Z_{1s}(\f)$ in $L^p(\R^m)$ for $p>m$. 
If $Z_{1s}(\f)\in \Bb(L^p(\R^m))$ for some $p>m$, then 
\reflm(m/2m) implies $\f \in \Eg_0$. Thus, assume 
$\f\in \Eg_0$ in the following lemma. 
The following lemma proves statements (3) of \refth(theo5) 
and \refth(theo6) for odd $m\geq 7$. 

\bglm \lblm(m<p)
Let $m \geq 5$ be odd, $p>m$ and $\f \in \Eg_0$. Then: 
\ben 
\item[{\rm (1)}] For a constant $C_p>0$,  
$\|Z_{s1}(\f)u + |\f \ra \la \f|) u \|_p \leq C \|u\|_p$.  
\item[{\rm (2)}]  
If $Z_{1s}(\f)$ is bounded in $L^p(\R^m)$ for some $p>m$, 
then $\f \in \Eg_1$. In this case $Z_{1s}$ is bounded in $L^p(\R^m)$ 
for all $1<p<\infty$. 
\een
\edlm 
\bgpf Considering that 
$\int_{\R} r^{j+1} e^{-i\lam r} M_\ast (r)dr
= i^{j+1}\left(\int_{\R} e^{-i\lam r} 
M_\ast (r)dr\right)^{(j+1)}$, we apply integration by parts 
to \refeq(kjk-mast). Then, for $k\geq 1$, we have  
\bqn 
K^{(j,k)}(\rho)
=\frac{(-i)^{j+1}}{2\pi}\int_0^\infty 
\left(e^{i\lam\rho}\lam^{j+k}F(\lam)\right)^{(j+1)} 
\left(\int_{\R} e^{-i\lam r} M_\ast (r)dr\right)
d\lam \lbeq(inte-pm)
\eqn
and, if $k=0$, additional boundary term which is 
given by virtue of \refeq(bd-m) by 
\bqn \lbeq(brytn)
\frac{(-i)^{j+1}j!}{2\pi}\int_{\R}M_\ast (r)dr
= \frac{i(-i)^{j}j! \Ga\left(\frac{m}2\right)}{\sqrt{\pi}\Ga\left(\frac{m-1}2\right)}
\la \f, u\ra, \quad j=0, \dots, \frac{m-3}2.  
\eqn 
Denote the right of \refeq(inte-pm) 
by $\tilde{K}^{(j,0)}(\rho)$ when $k=0$. Then,  
\bqn \lbeq(est-88)
\frac{K^{(j,k)}(\rho)}{\rho^{m-2-k}} 
\absleq C\left(1+ \frac1{\rho^{m-2}}\right) 
\Mg \Hg(M_\ast)(\rho), \quad 0\leq j,k\leq \frac{m-3}2  
\eqn 
and the same for $\tilde{K}^{(j,0)}(\rho)$. We split 
$Z_{s1}^{jk}u$ as previously:   
\[
Z_{s1}^{jk}u(x)=\left(\int_{|x-y|\leq 1} + \int_{|x-y|>1}\right)
\frac{V\f(y)K^{(j,k)}(|x-y|)}{|x-y|^{m-2-k}}dy 
= I_1(x)+ I_2(x).
\]
We estimate $I_2(x)$ by using 
\refeq(est-88) for $\rho\geq 1$, that 
$\rho^{m-1}$ is $A_p$ weight for $p>m$, \refeq(asymp-d) 
for $\f\in \Eg_0$ and the Calder\'on-Zygmund theory. This yields    
\begin{align}
\|I_2\|_p \leq \|V\f\|_1
\left(\int_1^\infty 
|\Mg \Hg(M_\ast)(\rho)|^p \rho^{m-1}d\rho\right)^{1/p} \notag \\
\leq \|V\f\|_1 \||D|^{-1}(V\f)\ast u\|_p 
\leq C \|V\f\|_1 \|u\|_p.   \lbeq(est-89)
\end{align}
H\"older's inequality and \refeq(est-88) for 
$\rho\leq 1$ imply 
\[
|I_1(x)|\leq C \left(\int_{|y|\leq 1}
\left|\frac{(V\f)(x-y)}{|y|^{m-2}}\right|^{p'}dy\right)^{1/p'} 
\left(\int_{|y|\leq 1}
|\Mg\Hg(M_\ast)(|y|)|^{p}dy\right)^{1/p}.
\]
The second factor on the right is bounded by $C\|u\|_p$ 
as in \refeq(est-89). Since $p'<\frac{m}{m-1}<m<p$, 
it follows by Minkowski's inequality that 
\[
\|I_1\|_p \leq C \|V\f\|_p \|u\|_p \left(\int_{|y|\leq 1}
\frac{dy}{|y|^{(m-2)p'}}\right)^{1/p'}
\leq C \|V\f\|_p \|u\|_p. 
\]
Thus, $Z_{s1}^{jk}\in \Bb(L^p(\R^m))$ for $p>m$ if 
$k\geq 1$ and the same for the operator 
$\tilde{Z}_{s1}^{j0}$ produced by $\tilde{K}^{(j,0)}(\rho)$. 
The contribution of boundary terms \refeq(brytn) to 
$Z_{s1}(\f)$ is given by using the constants $C_j$ 
of \refeq(coker) by 
\begin{align}
& 2i \sum_{j=0}^{\frac{m-3}2}
C_0 C_j(-1)^{j+1}\w_{m-1} 
\left(\int_{\R^d}\frac{(V\f)(y)}{|x-y|^{m-2}}dy \right) 
\frac{i(-i)^{j}j! \Ga\left(\frac{m}{2}\right)}
{\sqrt{\pi}\Ga\left(\frac{m-1}2\right)}
\la \f, u\ra  \notag \\
& \hspace{1cm} =-\tilde{D}_m |\f\ra \la \f, u\ra, \quad 
\tilde{D}_m =  
\sum_{j=0}^{\frac{m-3}{2}}\frac{(m-3-j)!}
{2^{m-3-j}\left(\frac{m-3}2\right)!\left(\frac{m-3}2-j\right)!} .
\end{align}
The constant $\tilde{D}_m$ can be elementarily 
computed and with  $n=\frac{m-3}2$ 
\[
\tilde{D}_m= 
\sum_{k=0}^{n} 
\frac1{2^{2n-k}}\begin{pmatrix} 2n-k \\ n-k  \end{pmatrix} 
=
\sum_{k=0}^{n} 
\frac1{2^{n+k}}\begin{pmatrix} n+k \\ k  \end{pmatrix} =1. 
\]
(see also page 167 of \cite{GKP}.)
This proves statement (1). We omit the 
proof of (2) which is similar to the corresponding 
statement of \reflm(m/2m).  
\edpf 

Since $Z_{s1} u = \sum_{i=1}^n Z_{s1}(\f_j)$ for the orthonormal 
basis of $\Eg$, the combination of lemmas in this section 
proves \refths(theo5,theo6) for odd $m$. 

\section{Proof of \refthb(theo6) for even $m \geq 6$} 

For proving \refth(theo6) for even dimensions $m\geq 6$ 
we need study $Z_s$ and $Z_{\log}$ of \refeq(e-zs) 
and \refeq(loge). Since $Z_{\log}$ may be studied 
in a way similar to but simpler than that for 
$Z_s$, we shall be mostly concentrated on $Z_s$ and 
only briefly comment on $Z_{\log}$ at the end of the 
section. As in odd dimensions we take the real orthonormal 
basis $\{\f_1, \dots, \f_d\}$ of $\Eg$ and define, for 
$\f \in \Eg$, 
\bqn 
Z_{s}(\f)u= \frac{i}{\pi}\int^\infty_0 G_0( \lam)
|V\f\ra \la \f V|(G_0(\lam)-G_0(-\lam))F(\lam)\lam^{-1} 
d\lam.  
\lbeq(zsf) 
\eqn 
Then, we have 
\[
Z_s u= \sum_{j=1}^d Z_{s}(\f_j)u 
\]
and we study $Z_s(\f)$ for $\f \in \Eg$. 
In this section we choose and fix a $\f\in \Eg$ 
arbitrarily and write $M(r)=M(r,V\f \ast \check{u})$.
  
We wish to apply the argument for odd dimensions 
also to even dimensions as much as possible and, 
we express ${Z}_{s}(\f)$ as a superposition 
of operators which are of the same form as those studied 
in odd dimensions except scaling. We set $\n=(m-2)/2$. 
Define for $a>0$ 
\bqn 
M^a(r)=M(r/(1+2a))
\eqn 
and, for $j, k=0, \dots, \n$ and $a,b>0$, 
\bqn 
Q_{jk}^{a,b}(\rho) = \frac{(-1)^{j+1}}{2\pi(1+2a)^{j+2}} 
\int_{0}^\infty \lam^{j+k-1} e^{i\lam(1+2b)\rho} 
\Fg (r^{j+1} M^a)(\lam) F(\lam) d\lam. \lbeq(kjkab) 
\eqn 
As in \refeq(kjk-mast), we may express $Q_{jk}^{a,b}(\rho)$ 
by using $M_\ast(r)$ and increase the factor $\lam^{j+k-1}$ 
of \refeq(kjkab) to $\lam^{j+k}$: 
\bqn 
Q_{jk}^{a,b}(\rho) = \frac{(-1)^{j+1}}{2\pi(1+2a)^{j+2}} 
\int_{0}^\infty \lam^{j+k} e^{i\lam(1+2b)\rho} 
\Fg (r^{j+1} M^a_\ast)(\lam) F(\lam) d\lam. \lbeq(k0kab+) 
\eqn 
When $j=0$, we also use $\tilde{M}(r)$ of \refeq(tiM) to express 
$Q_{0k}^{a,b}(\rho)$ as follows:
\bqn  
Q_{0k}^{a,b}(\rho)= \frac{i}{2\pi (1+2a)^{2}} 
\int_{0}^\infty \lam^{k} e^{i\lam(1+2b)\rho} 
\Fg (\widetilde{M}^a)(\lam) F(\lam) d\lam.  \lbeq(k0kab) \\
\eqn  
\bglm \lblm(Wsm)
Let ${Q}_{jk}^{a,b}(\rho)$ be defined 
by \refeq(kjkab), \refeq(k0kab+) or \refeq(k0kab). Then,    
\bqn 
Z_s (\f)u(x) = \frac{2i}{\w_{m-1}}\sum_{j,k=0}^\n 
T^{(a)}_j  T^{(b)}_k \left[ \int_{\R^m} 
\frac{(V\f)(x-y) Q^{a,b}_{jk}(|y|)}{|y|^{m-2-k}}dy 
\right]. \lbeq(Wsm-rep) 
\eqn 
\edlm 
\bgpf We apply \refeq(usformula) for 
$\la V\f ,(G_0(\lam)-G_0(-\lam))u \ra$ and 
\refeq(co-kerb) for $G_0(\lam)$ in \refeq(zsf). 
We see that $Z_s (\f)u(x)$ is the integral 
with respect to $\lam\in (0,\infty)$ of 
\[
\frac{i}{\pi}\sum_{j,k=0}^\n 
T^{(a)}_j  T^{(b)}_k 
\left[
\frac{(-1)^{j+1}\lam^{j+k-1}}{(1+2a)^{j+2}\w_{m-1}}
\left(\frac{e^{i\lam(1+2b)|y|}}
{|y|^{m-2-k}}\ast V\f \right)
\Fg (r^{j+1} M^a)(\lam)
\right]
F(\lam).
\] 
Integrating with respect to $\lam$ first yields 
\refeq(Wsm-rep). 
\edpf 

We define, for $0\leq j,k \leq \nu$ and $a,b>0$, that  
\begin{gather} \lbeq(zdef)
{Z}^{jk}(\f)u(x)= \frac{2i}{\w_{m-1}} 
T^{(a)}_j  T^{(b)}_k \left[{Z}^{jk}_{a,b}(\f)u(x)\right], \\
{Z}^{jk}_{a,b}(\f)u(x)
=  \int_{\R^m} 
\frac{(V\f)(x-y) Q^{a,b}_{jk}(|y|)}{|y|^{m-2-k}}dy .
\lbeq(zjkab-def)
\end{gather} 
\reflm(Wsm) implies $Z_s(\f)u= \sum {Z}^{jk}(\f)u$. 
In what follows we often write $Z^{jk}u$ and 
$Z^{jk}_{a,b}$ respectively for $Z^{jk}(\f)u$ 
and $Z^{jk}_{a,b}(\f)$.

\subsection{Estimate of $\|{Z}^{jk}u\|_p$ for 
$(j,k)\not=(\n,\n)$.} 

We estimate ${Z}^{jk}$ for 
the case $(j,k)\not=(\nu,\nu)$ first, 
postponing the case $(j,k)=(\nu,\nu)$ to 
the next subsection. As we shall see, the argument 
used for odd dimensions applies to ${Z}^{jk}$ if  
$(j,k)\not=(\nu,\nu)$ modulo superpositions and 
scalings. 

\bglm \lblm(q-est) With suitable 
constants $C>0$, followings are majorants of 
$Q_{jk}^{a,b}(\rho)$ for $0\leq k,j\leq \nu$ 
which satisfy the attached conditions respectively:  
\begin{align} 
&  
{\rm (1)} 
\hspace{0.5cm} 
C\frac{\{\Mg \Hg(r^{j+1} M^a)\}((1+2b)\rho)}{(1+2a)^{j+2}}, 
\quad \mbox{if}\ \ j+k\geq 1.  \lbeq(njk) \\
&  {\rm (2)} 
\hspace{0.5cm} 
C\frac{\Mg \Hg (\widetilde{M}^a)((1+2b)\rho)}{(1+2a)^2}, 
\quad \mbox{if}\ \ j=0.  \lbeq(n0) \\
& {\rm (3)} 
\hspace{0.5cm} 
C\sum_{l=0}^{k+1}\frac{\Mg\Hg(r^{j+l+1} M^{a})((1+2b)\rho )}
{(1+2a)^{j+2}(1+2b)^{k+1}\rho ^{k+1}}, 
 \quad \mbox{if}\ \ 2 \leq j \leq \n. \lbeq(devel-a) \\
& 
{\rm (4)} 
\hspace{0.5 cm} 
C\frac{\Mg \Hg(r^{2} M^a)((1+2b)\rho)}{(1+2a)^{j+2}} 
\{(1+2b)^{j-1}\rho^{j-1} + 1\}, 
\quad \mbox{if}\ \ 1 \leq j. 
\lbeq(kjkab-sub)  \\
& 
{\rm (5)} 
\hspace{0.5cm}
C \frac{\Mg \Hg(r M^a)((1+2b)\rho)}
{(1+2a)^{j+2}}\{(2b+1)^j \rho^{j}+1\}, 
\quad \mbox{for all }\  j,k . 
\quad \lbeq(kjkab-sub-mre) 
\end{align}
\edlm 
\bgpf 
Define $\Phi_{jk}(\lam)=\lam^{j+k-1} F(\lam)$. 
If $j+k\geq 1$, $\Phi_{jk}\in C_0^\infty(\R)$ and 
\reflm(add-cut) implies 
$Q_{jk}^{a,b}(\rho)=(-1)^{j+1}(1+2a)^{-(j+2)} 
\{(\Fg \Phi_{jk}) \ast \Hg(r^{j+1} M^a)\}((1+2b)\rho)$.
Then, \refeq(njk) follows by applying \refeq(max). 
Likewise  we have \refeq(n0) from \refeq(k0kab) 
If $j\geq 2$, we apply integration by parts $k+1$ 
times  to \refeq(kjkab) using that 
$e^{i\lam(1+2b)\rho }= (i(1+2b)\rho )^{-(k+1)}
\pa_{\lam}^{k+1} e^{i\lam(1+2b)\rho }$ then, without boundary 
terms, 
\begin{align}\lbeq(devel)
& Q^{a,b}_{jk}(\rho )= \sum_{l=0}^{k+1} 
\frac{(-1)^{j+1}}{2\pi(1+2a)^{j+2}} 
\left(\frac{1}{-i(1+2b)\rho }\right)^{k+1} 
\begin{pmatrix} k+1 \\ l \end{pmatrix} \notag \\
& \times \int_{0}^\infty 
e^{i\lam(1+2b)\rho } \Phi_{jk}(\lam)^{(k+1-l)} 
\Fg ((-i)^l r^{j+l+1} M^{a})(\lam) d\lam 
\end{align}
and \refeq(devel-a) follows as previously.
If $j\geq 1$, we may apply integration 
by parts to \refeq(kjkab) by using that 
$\Fg (r^{j+1} M^a)(\lam) = i^{j-1}
\{\Fg (r^{2} M^a)(\lam)\}^{(j-1)}$. Then  
\bqn 
Q_{jk}^{a,b}(\rho) = i^{j-1}\int_{0}^\infty 
\frac{\left(\lam^{j+k-1}F(\lam) 
e^{i\lam(1+2b)\rho}\right)^{(j-1)} 
\Fg (r^{2} M^a)(\lam)}
{2\pi(1+2a)^{j+2}} 
d\lam  \lbeq(kjkab-sub-1) 
\eqn 
and \refeq(kjkab-sub) follows. 
Apply another integration by parts in \refeq(kjkab-sub-1). 
No boundary term appears as $\Fg(rM^a)(0)=0$, and we obtain 
\refeq(kjkab-sub-mre).  
\edpf 

\subsubsection{Estimate for  $1<p<\frac{m}{m-1}$} \lbsssec(6-1-1)

Define for $0\leq  \s \leq m-1$ and 
$1<p<\frac{m}{m-1}$: 
\bqn \lbeq(n)
N_{\s}^{a,b}(u) = \left( \int^\infty_0 
|\Mg\Hg(r^{\s} M^{a})((1+2b)\rho )|^p 
\rho ^{m-1-p(m-1)} d\rho  \right)^{1/p}.
\eqn 
\bglm \lblm(njl) For any 
$\frac{m}{1+\s}\leq q\leq \infty$,  
we have 
\bqn \lbeq(es-1)
N_{\s}^{a,b}\leq C \frac{(1+2b)^{m-1-\frac{m}p}}
{(1+2a)^{m-1-\frac{m}{p}-\s}}
(\|V\f\|_1+ \|V\f\|_{q})\|u\|_p.  
\eqn 
\edlm 
\bgpf Change variable first. 
Since $\rho^{m-1-p(m-1)}$ is an $A_p$-weight,  
\begin{multline}
N_{\s}^{a,b} =(1+2b)^{m-1-\frac{m}p} 
\left( \int^\infty_0 
|\Mg\Hg(r^{\s} M^{a})(\rho )|^p 
\rho ^{m-1-p(m-1)} d\rho  \right)^{1/p}  \\
\leq  C \frac{(1+2b)^{m-1-\frac{m}p}}
{(1+2a)^{m-1-\frac{m}{p}-\s}}
\left( \int^\infty_0 
|M(r)|^p r^{m-1-p(m-1-\s)} dr \right)^{1/p}. 
\lbeq(mn-p1)
\end{multline}
Denote by $I$ the integral on \refeq(mn-p1).  
Let  $\kappa=m-1-\s$. If $\kappa=0$, then 
$I \leq C\|V\f \ast u\|_p \leq C\|V\f\|_1\|u\|_p$ 
and \refeq(es-1) follows. 
Let $0<\kappa\leq m-1$. Split $I$ into integral 
over $0<r<1$ and $r>1$ and 
use $r^{m-1-p\kappa}\leq r^{m-1}$ for $r\geq 1$. Then, 
we have 
$I \leq C(\||x|^{-\kappa} (V\f \ast u)(x)\|_{L^p(|x|<1)} 
+ \|V\f\|_1 \|u\|_p)$. Take $\kappa'$ such that  
$\kappa<\kappa'<m$ and 
apply H\"older's and Young's inequalities for the 
integral over $|x|\leq 1$. We obtain with 
$q=\frac{m}{m-\kappa'}\in 
\left[\frac{m}{1+\s}, \infty\right]$ that 
\bqn 
I \leq 
C(\||x|^{-\kappa}\|_{L^{\frac{m}{\kappa'}}(|x|\leq 1)}
\|V\f\|_{q}+\|V\f\|_1) \|u\|_p . \lbeq(es-3)
\eqn  
This completes the proof. 
\edpf 

\bglm \lblm(zsq2) Suppose $1<p<\frac{m}{m-1}$. Then, 
for $2\leq j \leq \n$ and $0\leq k\leq \n$ such that 
$(j,k)\not=(\n, \n)$,  
\bqn \lbeq(zsq2)
\|{Z}^{jk}u\|_p \leq C \|u\|_p, \quad u \in 
C_0^\infty(\R^m).
\eqn
\edlm 
\bgpf Minkowski's and Young's inequality imply   
\bqn 
\|{Z}^{jk}u\|_p \leq 
2\omega_{m-1}^{-1}\|V\f\|_1 
\cdot T^{(a)}_j  T^{(b)}_k 
\left[\left\||x|^{2+k-m}Q^{a,b}_{jk}\right\|_p  
\right]. \lbeq(k-4) 
\eqn 
We apply \refeq(devel-a) to estimate 
$Q^{a,b}_{jk}(|x|)$. Then, since 
$\s \equiv j+l+1\leq m-1$ for $(j,k)\not=(\n,\n)$, 
\reflm(njl) implies 
\bqn 
\left\||x|^{2+k-m}Q_{jk}^{a,b}\right\|_p 
\leq C (1+2a)^{\frac{m}{p}-(m-k-1)}
(1+2b)^{m-2-\frac{m}p-k}\|u\|_p.  \lbeq(d-1)
\eqn   
We plug this to \refeq(k-4) and use $m-k-1\geq j+2$. 
Then, 
\begin{align*}
& \|{Z}^{jk} u\|_p 
\leq C_{mjk} T_j^{(a)} T_k^{(b)}[
(1+2a)^{\frac{m}{p}-(j+2)}
(1+2b)^{m-2-\frac{m}p-k}]\|u\|_p \\
& \leq C \|u\|_p 
\left(\int_0^\infty 
\frac{(1+2a)^{\frac{m}{p}-(j+2)}}{(1+a)^{(2\n-j+\frac12)}}
\frac{da}{\sqrt{a}}\right) 
\left(\int_0^\infty 
\frac{(1+2b)^{m-2-\frac{m}p-k}}
{(1+b)^{(2\n-k+\frac12)}}
\frac{db}{\sqrt{b}}\right) .
\end{align*} 
Counting powers show that the integrals 
are finite and the lemma follows. 
\edpf  

As in odd dimensions we use the cancellation in 
\begin{multline} \lbeq(sum-jk)
Z^{0k}u+ Z^{1k}u \\
=\frac{2i}{\omega_{m-1}}
\int_{\R^m} \frac{(V\f)(x-y)}{|y|^{m-2-k}}
T_k^{(b)}(T_0^{(a)}Q_{0k}^{a,b}(|y|)+ 
T_1^{(a)}Q_{1k}^{a,b}(|y|))dy 
\end{multline}
and obtain the following lemma. 

\bglm \lblm(zsq3) For $1<p<\frac{m}{m-1}$, 
there exists a constant $C>0$ such that 
\bqn 
\|({Z}^{(0,k)}+ {Z}^{(1,k)})u\|_p 
\leq C \|u\|_p, \quad k=0,\dots, \n.
\eqn
\edlm 
\bgpf We apply integration by parts $k+1$ times 
to \refeq(k0kab) and \refeq(kjkab) as in the proof 
of \refeq(devel-a). This produces  
\begin{align}
Q_{0k}^{a,b}(\rho )
& = \frac{-i^{k}k!(\Fg \widetilde{M^a})(0){\w}_{m-1}}
{2\pi(1+2a)^2(1+2b)^{k+1}\rho ^{k+1}}
- \frac{i^k {\w}_{m-1}}{2\pi}
\sum_{l=0}^{k+1}C_{k+1,l} Q_{0k,l}^{a,b}(\rho ),  \lbeq(b+1) \\
Q_{1k}^{a,b}(\rho )& = 
\frac{i^{k+1}k!\Fg(r^2 M^a)(0){\w}_{m-1}}
{2\pi(1+2a)^3(1+2b)^{k+1}\rho ^{k+1}}
+\frac{i^{k+1}{\w}_{m-1}}{2\pi}\sum_{l=0}^{k+1} 
C_{k+1,l} Q_{1k,l}^{a,b}(\rho ), \lbeq(b+2) 
\end{align}
where $Q_{0k,l}^{a,b}(\rho )$ and 
$Q_{1k,l}^{a,b}(\rho )$ are given and estimated as follows:  
\begin{align}
Q_{0k,l}^{a,b}(\rho ) & = \int_0^\infty 
\frac{e^{i\lam(1+2b)\rho }(\lam^k F(\lam))^{(k+1-l)}
(\Fg ((-ir)^l \widetilde{M^a})(\lam))}
{(1+2a)^2(1+2b)^{k+1}\rho ^{k+1}}d\lam \notag \\ 
& \absleq C 
\frac{\Mg\Hg (r^l \widetilde{M^a})((1+2b)\rho))}
{(1+2a)^2(1+2b)^{k+1}\rho ^{k+1}},  \lbeq(fg-1) \\
Q_{1k,l}^{a,b}(\rho )& = (-i)^{l} 
\int_0^\infty \frac{e^{i\lam(1+2b)\rho }(\lam^k F(\lam))^{(k+1-l)}
\Fg(r^{2+l} M^a)(\lam))}
{(1+2a)^3(1+2b)^{k+1}\rho ^{k+1}}
d\lam \notag \\
& \absleq C \frac{\Mg\Hg(r^{2+l} M^a)((1+2b)\rho)}
{(1+2a)^3(1+2b)^{k+1}\rho ^{k+1}}.
\lbeq(fg-2) 
\end{align} 
Eqn.\refeq(MtildeM-1) shows 
$\Fg (\widetilde{M^a})(0)
= \Fg (r^2 M^a)(0)
= (1+2a)^3 \int_0^\infty r^2 M(r) dr$ and  
\[
T_1^{(a)} [i]= T_0^{(a)}[(1+2a)]= (m-3)^{-1} 
\]
It follows that the sum of the superposition via $T_0^{(a)}$
of the boundary term of \refeq(b+1)  
and that via $T_1^{(a)}$  of \refeq(b+2)  vanishes:  
\bqn 
\frac{i^k k!}
{(1+2b)^{k+1}\rho^{k+1}}
\left(\int_0^\infty r^2 M(r) dr \right)
(T_1^{(a)}[i] -T_0^{(a)}[(1+2a)])=0.  
\lbeq(vabo)
\eqn 
For $1<p<\frac{m}{m-1}$, $\rho^{m-1-p(m-1)}$ is an 
$A_p$ weight on $\R$ and we have the identity: 
\bqn \lbeq(Ha-def)
\widetilde{M^a}(r) = \int_r^\infty s M^a (s) ds 
= (1+2a)^2 \tilde{M}((1+2a)^{-1}r). 
\eqn 
Then, \reflm(ap), \refeq(Ha-def), 
change of variable and Hardy's inequality imply   
\begin{align}
& \left\|\frac{Q_{0k,l}^{a,b}(|x|)}{|x|^{m-k-2}}\right\|_p 
\leq 
\frac{C(1+2b)^{m-1-\frac{m}{p}}}{(1+2a)^2(1+2b)^{k+1}} 
\left( \int^\infty_0 
|r^{l} \widetilde{M^{a}}(r)|^p 
r^{m-1-p(m-1)} dr \right)^{1/p} \notag \\
& \leq \frac{C(1+2a)^{\frac{m}{p}-(m-1-l)}}
{(1+2b)^{\frac{m}{p}-(m-k-2)}}
\left( \int^\infty_0 
|M(r)|^p r^{m-1-p(m-3-l)} dr \right)^{1/p}.   
\lbeq(last-1)
\end{align} 
The integral is similar to the integral which appeared 
in \refeq(mn-p1) and we remark $m-3-l \geq 0$ for $m\geq 6$. 
Thus, applying \refeq(es-3) with $\sigma=l+2$, we obtain    
\bqn 
\refeq(last-1) 
\leq \frac{C(1+2a)^{\frac{m}{p}-(m-k-2)}}
{(1+2b)^{\frac{m}{p}-(m-k-2)}}(\|V\f\|_1+ \|V\f\|_{\frac{m}{3}})
\|u\|_p, \ 0\leq l\leq k+1. \lbeq(d6-0)
\eqn  
Counting the powers of $a$ and $b$, we 
thus have from \refeq(d6-0) that  
\bqn 
T_0^{(a)}T_k^{(b)} \left[ 
\left\|
|x|^{2+k-m}Q_{0k,l}^{a,b}\right\|_p 
\right] \leq C \|u\|_p.  \quad \ 0\leq l\leq k+1 \lbeq(g-1).
\eqn 
Entirely similarly, starting from \refeq(fg-2), we obtain  
\begin{align}
& \left\|
\frac{Q_{1k,l}^{a,b}(|x|)}
{|x|^{m-k-2}}
\right\|_p 
\leq \frac{
C(1+2b)^{m-1-\frac{m}{p}}
}
{(1+2a)^3(1+2b)^{k+1}} 
\left( 
\int^\infty_0 
|r^{2+l}M^a(r)|^p 
r^{m-1-p(m-1)} dr 
\right)^{1/p}  \notag \\
& \leq 
\frac{
C(1+2a)^{\frac{m}{p}-(m-k-1)}
}
{
(1+2b)^{\frac{m}p-(m-k-2)}
} (\|V\f\|_1+ \|V\f\|_{\frac{m}{3}})
\|u\|_p, \ 0\leq l\leq k+1. \lbeq(d7-m)
\end{align}
The extra decaying factor $(1+2a)^{-1}$ of \refeq(d7-m) 
compared to \refeq(d6-0) cancels the extra increasing 
factor $(1+a)$ of $T_1^{(a)}$ compared to $T_0^{(a)}$ and 
we have    
\bqn 
T_1^{(a)}T_k^{(b)} 
\left[ 
\left\|{|x|^{k+2-m}}
{Q_{0k,l}^{a,b}(|x|)}
\right\|_p 
\right] \leq C \|u\|_p, \quad 0\leq l \leq k+1.
\lbeq(g-2)
\eqn 
In view of \refeq(sum-jk), \refeq(b+1), \refeq(b+2) and 
\refeq(vabo), \refeq(g-1) and \refeq(g-2) with the 
help of Young's and Minkowski's inequalities imply the lemma.  
\edpf 

\subsubsection{Estimate for $\frac{m}3<p<\frac{m}2$} 

The following lemma together with \reflm(zsq2) and 
\reflm(zsq3) will prove that 
$\sum_{(j,k)\not=(\n,\n)}Z^{jk}$ 
is bounded in $L^p(\R^m)$ for $1<p<\frac{m}2$. 

\bglm \lblm(e-m3) 
Let $\frac{m}3<p<\frac{m}2$. Then, 
for $(j,k)\not=(\n,\n)$,   
\bqn \lbeq(e-m3)
\|Z^{jk} u\|_p\leq C_p \|u\|_p 
\eqn 
for a constant $C_p>0$ independent of 
$u \in C_0^\infty(\R^m)$.
\edlm 
\bgpf Except the superposition the proof is virtually the 
repetition of that of statement (2) of \reflm(lemma). 

\noindent 
(1) Let $j \geq 1$ first. 
Since $\rho^{m-1-2p}$ is $A_p$ weight for 
$\frac{m}3<p<\frac{m}2$, we have 
\bqn 
\left(\int_0^\infty 
|\{\Mg \Hg(r^{2} M^a)\}(\rho)|^p 
\rho^{m-1-2p}d\rho \right)^{\frac1{p}} 
\leq C(1+2a)^{\frac{m}p}\|V\f\|_1 \|u\|_p. \lbeq(8-2)
\eqn
Splitting the integral of \refeq(zjkab-def) we define  
\bqn 
Z_{jk}^{a,b}u(x) = 
\left(
\int_{|y|<\frac1{1+2b}}+ \int_{|y|>\frac1{1+2b}}\right)   
\frac{(V\f)(x-y) Q_{jk}^{a,b}(|y|)}{|y|^{m-2-k}}dy
= I_1(x)+ I_2(x) .  \lbeq(spt)
\eqn 
For $I_1(x)$, we estimate $|y|^{-(m-k-2)}\leq |y|^{-(m-2)}$ 
for $|y|\leq 1$ and apply H\"older's inequality.  Then 
\[
\|I_1\|_p \leq \left\|
\int_{|y|\leq \frac1{2b+1}}
\frac{|(V\f)(x-y)|^{p'} dy}{|y|^{p'(m-4)}}
\right\|_{p/p'}^{1/p'} \left(\int_{|y|\leq \frac1{2b+1}} 
\left|
\frac{Q_{jk}^{a,b}(|y|)}{|y|^{2}}
\right|^p dy 
\right)^{1/p} 
\]
Minkowski's inequality implies that the first factor 
on the right 
is bounded by $C\|V\f\|_p (1+2b)^{m-4-\frac{m}{p'}}$ 
and $\frac{m}{p'}-(m-4)>1$.  For the second factor, 
we apply \refeq(kjkab-sub) for $(1+2b)\rho<1$ and 
then \refeq(8-2). We obtain 
\bqn 
\|I_1\|_p \leq 
C(1+2a)^{\frac{m}{p}-j-2}(1+2b)^{1-\frac{m}{p}} 
\|V\f\|_1 \|V\f\|_p \|u\|_p. \lbeq(8-9)
\eqn 
By Young's inequality 
$\|I_2\|_p \leq C \|V\f\|_1 
\||x|^{2+k-m}Q_{jk}^{a,b}(|x|)\|_{L^p((1+2b)|x|>1)}$. 
For the second factor, we use \refeq(kjkab-sub) 
for $(1+2b)\rho\geq 1$ and, after changing 
the variables $\rho \to (1+2b)^{-1}\rho$, we estimate 
$\rho^{-(m-2-k-(j-1))}\leq \rho^{-2}$ for $\rho\geq 1$ 
(here we used $(j,k)\not=(\nu,\nu)$) and apply \refeq(8-2) 
once more. Then,   
\begin{align} 
& \|I_2\|_p \leq C \|V\f\|_1 
\frac{(1+2b)^{m-2-k-\frac{m}{p}}}{(1+2a)^{j+2}}
\left(
\int_1^\infty |\{\Mg \Hg(r^{2} M^a)\}(\rho)|^p 
\rho^{m-1-2p}d\rho 
\right)^{\frac1{p}} \notag \\
& \leq C  
(1+2a)^{\frac{m}{p}-j-2}(1+2b)^{m-2-k-\frac{m}{p}}
\|V\f\|_1^2 \|u\|_p.  \lbeq(ev-8)
\end{align}
Since 
$m-2-k\geq 1$ and  
$(1+2a)^{\frac{m}p-j-2}(1+2b)^{m-2-k-\frac{m}{p}}$ 
is summable by $T^{(a)}_j T^{(b)}_k$, 
\refeq(8-9) and \refeq(ev-8) imply   
\bqn 
\|Z^{jk} u\|_p 
\leq C \|V\f\|_1 (\|V\f\|_1+ \|V\f\|_{p}) \|u\|_p. 
\lbeq(8-10) 
\eqn 
(2) When $j=0$, using \refeq(n0) in stead of 
\refeq(kjkab-sub) for estimating $Q_{0k}^{a,b}(\rho)$, 
we apply the argument in the proof in (1). 
Then, the help of 
\refeq(Ha-def) and Hardy's inequality it leads to the 
same estimates \refeq(8-9) and \refeq(ev-8) and, hence, 
to the desired \refeq(e-m3) for $j=0$. 
This completes the proof of the lemma.
\edpf 

\subsubsection{Estimate for $m/2<p<m$ and for $p>m$. } \lbsssec(6-1-3)
We now estimate  $Z^{jk}$, $(j,k)\not=(\nu,\nu)$, 
in $L^p(\R^m)$ for $\frac{m}2<p<m$ and for $p>m$. 
As in odd dimensions, 
$Z^{00}$ will not in general be bounded in $L^p(\R^m)$ 
when $\frac{m}2<p<m$ and likewise for all $Z^{0k}$, 
$k=0, \dots, \frac{m-2}{2}$ when $p>m$.  
Elementary computations using 
\[
z^{-k}=\frac1{\Ga(k)}
\int_{0}^\infty e^{-zt}t^{k-1}dt, \quad \Re z>0, \  \ 
k>0
\]
and the formula \refeq(const-cm0) for $C_{m,j}\w_{m-1}$ 
we obtain the following lemma.

\bglm \lblm(comput)
%\ben
\item[{\rm (1)}] We have $T_j^{(a)}[1]=(m-3-j)!/(m-2)!$. 
\item[{\rm (2)}] For $k\geq 1$ and $j=0,\cdots, \n$,
$T_j^{(a)}[(1+2a)^{-k}]$ is given by 
\bqn \lbeq(tja1)
\frac{(-i)^j2^{m-1}\Ga(2\n-j+k)}{(m-2)!\Ga(k)}
\begin{pmatrix} \n \\ j \end{pmatrix}.
\int_1^\infty \frac{(x^2-1)^{k-1}}{(x^2+1)^{2\n-j+k}}dx 
\eqn 
%\een
\edlm 

\bglm \lblm(even-m/2m) Let $\frac{m}2<p<m$ and 
$\f \in \Eg$. Then:
\ben 
\item[{\rm (1)}] If $(j,k) \not=(0,0)$ or 
$(j,k)\not=(\n,\n)$, $Z^{jk}$ is bounded in $L^p(\R^m)$:  
\bqn \lbeq(m/2m-a)
\|Z^{jk}u \|_p \leq C \|u\|_p, \quad 
u \in C_0^\infty(\R^m) 
\eqn 
\item[{\rm (2)}] There exists a constant $C>0$ such that  
$u\in C_0^\infty(\R^m)$, we have 
\begin{gather}\lbeq(meven-m)
\left\|Z^{00}u + D_m |\f\ra \la \f, u \ra 
\right\|_p \leq C \|u\|_p, \\  
D_m= \frac{2^m \Ga\left(\frac{m}{2}\right)}
{\sqrt{\pi}\Ga\left(\frac{m-1}2\right)}\, 
\int_1^\infty (1+x^2)^{m-1}dx. \lbeq(const-dm)
\end{gather}
If $Z^{00}(\f)$ is bounded in $L^p(\R^m)$ 
for some $\frac{m}2<p<m$ then $\f\in \Eg_0$. 
In this case $Z^{00}(\f)$ is bounded in $L^p(\R^m)$ 
for all $\frac{m}2<p<m$. 
\een
\edlm  
\bgpf (1) Split $Z_{a,b}^{jk}u(x)$ as in \refeq(spt) 
and apply the argument thereafter to $I_1(x)$ 
and $I_2(x)$ by using the estimate \refeq(kjkab-sub-mre).
Since $m-2-(k+j)\geq 1$ and $\rho^{m-1-p}$ is 
an $A_p$ weight for $\frac{m}2<p<m$, we have, as in \refeq(ev-8),  
\bqn \lbeq(8-11)
\|I_2\|_p \leq C  
\frac{(1+2b)^{m-2-k-\frac{m}{p}}}
{(1+2a)^{j+2-\frac{m}{p}}}
\|V\f\|_1^2 \|u\|_p.
\eqn 
For dealing with $I_1(x)$, we estimate 
$|y|^{-(m-2-k)}\leq |y|^{m-2}$ for $|y|\leq 1$ 
as previously but now decompose 
$|y|^{-(m-2)}=|y|^{-(m-3)} \cdot |y|^{-1}$,  
remarking that $(m-3)p'<m$ and $p/p'>1$. Then,  
we obtain as in \refeq(8-9) that 
\bqn \lbeq(8-12)
\|I_1 \|_p 
\leq \frac{(1+2a)^{\frac{m}{p}-j-2}}
{(1+2b)^{\frac{m}{p}}} \|V\f\|_1\|V\f\|_p \|u\|_p. 
\eqn 
Summing up \refeq(8-11) and \refeq(8-12)  
by $T^{(a)}_j T^{(b)}_k $, we obtain \refeq(m/2m-a). 

\noindent 
(2) Let $j=k=0$. 
We apply integration by parts to \refeq(k0kab+). 
\begin{multline}
Q_{00}^{a,b}(\rho) = \frac{-i}{2\pi(1+2a)^{2}} 
\int_{0}^\infty e^{i\lam(1+2b)\rho} 
\Fg (M^a_\ast)'(\lam) F(\lam) d\lam  \\
= \frac{i}{2\pi} \int_{\R} 
\frac{M_\ast^a(r)}{(1+2a)^2} dr + 
\frac{i}{(1+2a)^{2}} \int_{0}^\infty 
(F(\lam) e^{i\lam(1+2b)\rho})' \Fg (M_\ast^a)(\lam) d\lam .
\lbeq(q00-2) 
\end{multline}
Denote the second term on \refeq(q00-2) by $\tilde{Q}_{00}^{ab}(\rho)$ 
and by $\tilde{Z}^{00}$ the operator produced by inserting 
$\tilde{Q}_{00}^{ab}(\rho)$ for ${Q}_{00}^{ab}(\rho)$ 
in \refeq(zdef). We have  
\bqn \lbeq(ex-ra)
\tilde{Q}_{00}^{a,b}(\rho) 
\absleq C \frac{\Mg \Hg(M_\ast^a)((1+2b)\rho)}{(1+2a)^{2}} 
(1+ (1+2b)\rho). 
\eqn 
Let $\frac{m}2<p<m$. We split as in \refeq(spt) and estimate 
$I_2$ first: 
\[
\tilde{Z}^{00}u(x)= \left(
\int_{|y|< \frac{1}{1+2b}}+ 
\int_{|y|\geq \frac{1}{1+2b}}\right) 
\frac{(V\f)(x-y)\tilde{Q}_{00}^{a,b}(|y|)}{|y|^{m-2}}dy
= I_1(x) + I_2(x). 
\]
We obtain  
\begin{align} 
\|I_2\|_p & \leq 
C \|V\f\|_1\frac{(1+2b)^{m-2-\frac{m}{p}}}{(1+2a)^2} 
\left\|
\frac{\Mg\Hg(M_\ast^a)(|y|)}{|y|^{m-3}}
\right\|_{L^p(|y|>1)} \notag \\
& \leq C 
\|V\f\|_1\frac{(1+2b)^{m-2-\frac{m}{p}}}{(1+2a)^2} 
\left\|\frac{1}{|y|^{m-3}}
\right\|_{L^m(|y|>1)} 
\left(\int_0^\infty |M_\ast^a(r)|^q 
r^{m-1}dr \right)^{\frac1{q}} \notag \\
& \leq C \frac{\|V\f\|_1(1+2b)^{m-2-\frac{m}{p}}}
{(1+2a)^{2-\frac{m}{q}}}
\||D|^{-1}(V\f)\|_{\frac{m}{m-1},\infty}\|u\|_p,   
\lbeq(mane-11)
\end{align} 
where we used Young's inequality, \refeq(ex-ra) for 
$(1+2b)\rho\geq 1$ and 
change of variable $(1+2b)\rho$ to $\rho$ in the first stage, 
H\"older's inequality considering $p^{-1}= m^{-1}+ q^{-1}$ 
and that $1$ is an $A_q$ weight $q=mp/(m-p)>m$ in the second 
and finally weak-Young's inequality. 
For $I_1$, we apply H\"older's and Minkowski's 
inequalities and \refeq(mane-11) and obtain     
\begin{align}
& \|I_1\|_p 
\leq C \left\|\left(\int_{|y|\leq \frac1{1+2b}}
\left|\frac{(V\f)(x-y)}{|y|^{m-2}}\right|^{q'}dy
\right)^{\frac1{q'}}\right\|_p  
\notag \\
& \qquad \times (1+2b)^{-\frac{m}{p}} (1+2a)^{-2} 
\left(\int_{|y|\leq 1}
|\Mg\Hg(M_\ast^a)(|y|)|^{q}dy\right)^{1/q} \notag \\
&\qquad  \leq C (1+2b)^{-\frac{m}{p}} (1+2a)^{\frac{m}{p}-2} 
\|V\f\|_p \||D|^{-1}(V\f)\|_{\frac{m}{m-1},\infty}\|u\|_p. 
\lbeq(mane-12)
\end{align}
Summing \refeq(mane-11) and \refeq(mane-12) 
by $T^{(a)}_0 T^{(b)}_0$, we obtain 
$\|\tilde{Z}^{(0,0)}u\|_p \leq C \|u\|_p$. 

\noindent 
By virtue of \refeq(redfe) and \refeq(bd-m), 
the contribution to ${Z}^{00}u$ of the boundary term of \refeq(q00-2) is 
given by 
\begin{align}
& \frac{2i}{\w_{m-1}}
T_0^{(a)}T_0^{(b)}\left[\int_{\R^m}\frac{(V\f)(y)dy}{|x-y|^{m-2}} \cdot
\frac{i}{2\pi} \int_{\R} 
\frac{M_\ast(r)}{(1+2a)} dr \right] \notag 
\\
& \hspace{1cm} = 
-\frac{2}{C_0\w_{m-1}}
T_0^{(a)}[(1+2a)^{-1}]T_0^{(b)}[1]
\frac{\Ga\left(\frac{m}{2}\right)}
{\sqrt{\pi}\Ga\left(\frac{m-1}2\right)}\la \f, u \ra \f. \lbeq(rimy)
\end{align}
By using \reflm(comput) and $C_0\w_{m-1}=(m-2)^{-1}$. 
we can simplify  
\refeq(rimy) to $-D_m \la \f, u \ra \f$ with $D_m$ given by 
\refeq(const-dm) and \refeq(meven-m) follows. 
The last statement follows as in the odd dimensional 
case, see the remark after \reflm(m/2m). \edpf

Finally in this section we study $Z^{jk}(\f)u$ for 
$(j,k)\not=(\nu,\nu)$ in $L^p(\R^m)$ when $p>m$, 
assuming $\f \in \Eg_0$ by the same 
reason as in odd dimensions. We define  
\bqn \lbeq(djm)
D_{m,j}= 2^{m} \begin{pmatrix} \n \\ j \end{pmatrix}
\frac{\Ga\left(\frac{m}2\right)}{\sqrt{\pi}\Ga\left(\frac{m-1}2\right)}
\int_1^\infty \frac{(x^2-1)^{j}}{(x^2+1)^{m-1}}dx, \quad 
j=0,\dots, \n.  
\eqn 

\bglm \lblm(m<p-even)
Let $m \geq 6$ be even and $p>m$. Suppose that 
$\f \in \Eg_0$. Then: 
\ben
\item[{\rm (1)}] 
For $(j,k)$ such that $k\not=0$ and $(j,k)\not=(\nu,\nu)$, $Z^{jk}$ 
is bounded in $L^p(\R^m)$. 
\item[{\rm (2)}] There exists a constant $C>0$ such that 
\bqn 
\|Z^{j0}u + D_{j,m} \la \f, u\ra \f \|_p \leq C \|u\|_p , \quad 
j=0,\dots, \n.
\eqn 
\item[{\rm (3)}] 
If $Z^{j0}(\f)$ is 
bounded in $L^p(\R^m)$ for some $0\leq j \leq \n$ and some $m<p<\infty$, 
then $\f \in \Eg_1$. In this case, $Z^{j0}(\f)$ is bounded in $L^p(\R^m)$ 
for all $1<p<\infty$ and $0\leq j \leq \n$. 
\een
\edlm 
\bgpf We apply integration by parts $j+1$ times to \refeq(k0kab+):  
\bqn \lbeq(esub)
Q_{jk}^{a,b}(\rho) = \int_{0}^\infty 
\frac{(-i)^{j+1}\lam^{j+k}F(\lam) e^{i\lam(1+2b)\rho} 
\pa_\lam^{j+1} \{\Fg (M^a_\ast)(\lam)\}}
{2\pi(1+2a)^{j+2}} 
d\lam  .
\eqn 
(1) If $k\geq 1$, then no boundary terms appear and we have 
\bqn \lbeq(esub-1)
Q_{jk}^{a,b}(\rho)
 \absleq 
\frac{C\Mg \Hg(M_\ast^a)((1+2b)\rho)}{(1+2a)^{j+2}} 
\{(1+2b)^{j+1}\rho^{j+1} + 1\}.
\eqn 
Observing that $m-2-(k+j+1)\geq 0$ for $(j,k)\not=(\nu,\nu)$, 
that $r^{m-1}$ is $A_p$ weight on $\R$ for $p>m$ and that 
$(m-2-k)p'<m$,  we apply the argument used for proving 
\refeq(mane-11) and \refeq(mane-12) in 
the proof of the previous lemma and obtain   
\begin{align}
&\|Z^{jk}_{a,b}u\|_p  \leq 
\frac{
C\|V\f\|_1(1+2b)^{m-2-k-\frac{m}{p}}
}
{(1+2a)^{j+2-\frac{m}{p}}} \||D|^{-1}(V\f)\ast u\|_p  
 \notag  \\
& + 
\frac{C(1+2b)^{-\frac{m}{p}}\|V\f\|_p}{(1+2a)^{j+2-\frac{m}{p}}}
\left( 
\int_{|y|<\frac1{1+2b}}\frac{dy}{|y|^{(m-2-k)p'}}
\right)^{\frac1{p'}}
\||D|^{-1}(V\f)\ast u\|_p
\lbeq(103)
\end{align}
Since $\int (V\f)(x)dx=0$, $\||D|^{-1}(V\f)\ast u\|_p 
\leq C \|u\|_p$ for any $1<p<\infty$ by virtue of 
\refeq(asymp-d1) and the Calder\'on-Zygmund theory. It follows 
that 
\[
\|Z^{jk}u\|_p \leq T_j^{(a)} T_k^{(b)}\|Z^{jk}_{a,b}u\|_p \leq C \|u\|_p
\]
for $k\geq 1$ and $(j,k)\not=(\n,\n)$ and statement (1) is proved.  

\noindent 
(2) If $k=0$, then, $j+1$ times integration by parts in \refeq(esub) 
produces, in addition to the integral which is 
bounded by \refeq(esub-1) and whose contribution to 
$Z^{j0}$ produces a bounded operator in $L^p(\R^m)$ for 
$p>m$, which may be proved by repeating  the argument of step (1), 
we have the boundary term which may be expressed as follows by using 
\refeq(bd-m) once more: 
\bqn \lbeq(b-tem)
\frac{i^{j+1}j!}{2\pi(1+2a)^{j+1}}\int_{\R}M_\ast(r) dr 
= \frac{-i^{j+1}j!}{(1+2a)^{j+1}}
\frac{\Ga\left(\frac{m}2\right)}
{\Ga\left(\frac{m-1}2\right)\sqrt{\pi}}\la \f,u\ra.
\eqn 
Therefore, the contribution of the boundary term to $Z^{j0}u$ may be 
computed as follows using 
$C_0\w_{m-1} = T^{(b)}_0[1]=(m-2)^{-1}$ by 
substituting \refeq(tja1) with $k=j+1$ for 
$T^{(a)}_j[(1+2a)^{-(j+1)}]$:
\begin{align*} 
& \frac{2 i}{\w_{m-1}}
T^{(a)}_j T^{(b)}_0\left[
\int_{\R^m}\frac{(V\f)(y)dy}{|x-y|^{m-2}}
\left(-\frac{i^{j+1}j!}{(1+2a)^{j+1}} \right)
\right] 
\frac{\Ga\left(\frac{m}2\right)}
{
\Ga\left(\frac{m-1}2\right)\sqrt{\pi}
}
\la \f,u\ra \notag \\
& \hspace{1cm} = 2 i^{j+2}j! 
T^{(a)}_j[(1+2a)^{-(j+1)}] 
\frac{\Ga\left(\frac{m}2\right)}
{\Ga\left(\frac{m-1}2\right)\sqrt{\pi}}
\la \f,u\ra \f  
= -D_{m,j}\la \f,u\ra \f 
\lbeq(la) 
\end{align*}
This proves statement (2). We omit the proof of statement (3) 
which is similar to the corresponding part of the 
previous lemma. 
\edpf 

\bglm Define $\tilde D_m= \sum_{j=0}^\n {D}_{m,j}$. 
Then, $\tilde{D}_m =1$. 
\edlm 
\bgpf Use binomial formula for \refeq(djm). We have  
\[
\tilde{D}_m= 
2^{m}\frac{\Ga\left(\frac{m}2\right)}
{\Ga\left(\frac{m-1}2\right)\sqrt{\pi}}
\int_1^\infty \frac{x^{m-2}}{(x^2+1)^{m-1}}dx 
\]
Change of variable $x\to x^{-1}$ shows that the integral 
is equal to the same integral over the interval 
$0<x<1$. It follows after making the change of variable 
$x^2=t$ that the integral is equal to  
\[
\frac14 
\int_0^\infty \frac{t^{\n-\frac12}}{(t+1)^{m-1}}dt 
= \frac{\Ga\left(\frac{m-1}2\right)^2}{2^2 \Ga(m-1)}. 
\]
Thus, 
$\tilde{D}_m= 
2^{m-2}\Ga\left(\frac{m}2\right)\Ga\left(\frac{m-1}2\right)
\Ga(m-1)^{-1}\pi^{-\frac12}=1.$
\edpf 

In the next two sections we prove that  $Z^{\n\n}$ 
and $Z_{\log}$ are  bounded 
in $L^p(\R^m)$ for all $1<p<\infty$. These will complete 
the proof of \refth(theo6). 

\subsection{Estimate of $\|Z^{\n\n}u\|_p$ for 
$1<p<\infty$} \lbssec(5-2)

In this section we prove 
\bqn \lbeq(aim)
\|Z^{\n\n}u\|_p\leq C \|u\|_p, \quad 1<p<\infty.
\eqn 
The method of previous subsection does not apply for 
proving this and we exploit more direct method.  
By virtue of interpolation, it suffices to prove \refeq(aim) 
for arbitrarily small $p>1$ and large $p>m$.  

\subsubsection{The case for $1<p<\frac{2(m-1)}{m+1}$} 
\lbsssec(5-2-1) 

We first show \refeq(aim) for $1<p<\frac{2(m-1)}{m+1}$. 
After changing the variable $r$ to $(1+2a)r$ in \refeq(kjkab), 
we write $Q_{\n\n}^{a,b}(\rho)/\rho^\n$ in the form 
\bqn 
\frac{(-1)^{\n+1}}
{2\pi\rho^{\n}} \int_0^\infty 
e^{i(1+2b)\rho\lam} \lam^{m-3} F(\lam)
\left(\int_{\R} e^{-i(1+2a)r\lam} r^{\n+1}M(r)dr \right) 
d\lam. \lbeq(moto) 
\eqn 
Integration by parts implies that \refeq(moto) is equal to   
\begin{align*}
& \frac{i(-1)^{\n+1}}
{2\pi(1+2b)\rho^{\n+1}} \int_0^\infty 
e^{i(1+2b)\rho\lam}
(\lam^{m-3} F(\lam))' 
\left(
\int_{\R} e^{-i(1+2a)r\lam} r^{\n+1}M(r)dr \right) 
d\lam \\
& + \frac{(-1)^{\n+1}(1+2a)}
{2\pi(1+2b)\rho^{\n+1}} \int_0^\infty 
e^{i(1+2b)\rho\lam}
\lam^{m-3} F(\lam)
\left(
\int_{\R} e^{-i(1+2a)r\lam} r^{\n+2}Mdr \right) 
d\lam.
\end{align*}
The  first line becomes 
$i(1+2b)^{-1}Q_{\n(\n-1)}^{a,b}(\rho)/\rho^{m-2-(\n-1)}$ 
if we replace $(m-3)F(\lam)+\lam F'(\lam)$ by $F(\lam)$  
and the former function can play the same role as the 
latter does in the argument 
of previous sections and $\n-1\geq 1$ if $m\geq 6$. 
Thus, if we substutute it for $Q_{\n\n}^{a,b}(\rho)/\rho^\n$ 
in \refeq(zjkab-def) for $(j,k)=(\nu,\nu)$ and then the 
resulting function into \refeq(zdef) for 
$Z_{a,b}^{\n\n}(\f)u(x)$, it produces the operator 
which has the same $L^p$ property as $Z^{\n(\n-1)}$ 
which is bounded in $L^p(\R^m)$ for $1<p<\infty$. 
Hence, we need study only the operator produced by the 
second line. Once again we substitute it for 
$Q_{\n\n}^{a,b}(\rho)/\rho^\n$ in \refeq(zjkab-def) 
and the result into \refeq(zdef) for 
$Z_{a,b}^{\n\n}(\f)u(x)$. We denote the functiotn 
thus obtain by $Z^{\n\n}u(x)$, abusing notation. 
We want to show that this $Z^{\n\n}u(x)$ satisfies 
\refeq(aim) for $1<p<\frac{m}{m-1}$. 
Integrating with respect to $a, b$ first 
via Fubini's theorem shows    
\begin{align} \lbeq(ann)
& Z^{\n\n}u(x)= \frac{2i}{\w_{m-1}}
\int_{\R^m}(V\f)(x-y) X_{\n}(|y|) dy, \\
& X_\n(\rho)= \frac{2iC_{m,\nu}^2\w_{m-1}}{\rho^{\nu+1}} 
\int_0^\infty 
\left\{
e^{i\lam\rho} \lam^{m-3}
\left(
\int_0^\infty 
\frac{(1+2b)^{-1}e^{2i\lam\rho {b}}}
{(1+b)^{\nu+\frac12}}
\frac{db}{\sqrt{b}}
\right) 
\right.  \notag \\
& \times \left.
\int_{\R}e^{-i\lam{r}}
\left(
\int_0^\infty 
\frac{(1+2a)e^{-2ia\lam{r}}}{(1+a)^{\nu+\frac12}}
\frac{da}{\sqrt{a}}
\right) \lbeq(m-1)
r^{\nu+2}M(r)dr 
\right\}
F(\lam)d\lam.  
\end{align}
Let $\chi_\pm(r)= 1$ for $\pm r>0$ and 
$\chi_\pm(r)=0$ for $\pm r \leq 0$. Define, for $t>0$,  
\begin{gather}
g_\pm (t) = \int_0^\infty 
\left(1+\frac{a}{t}\right)
\left(1+ \frac{a}{2t}\right)^{-\n-\frac{1}2} 
e^{\pm ia}\frac{da}{\sqrt{a}},    \lbeq(f-2a) \\
h_\pm (t) = \int_0^\infty 
\left(1+\frac{b}{t}\right)^{-1}
\left(1+ \frac{b}{2t}\right)^{-\n-\frac{1}2} 
e^{\pm ib}\frac{db}{\sqrt{b}}    \lbeq(f-2b)
\end{gather}
and, with $C=iC_{m,\nu}^2\w_{m-1}$, 
write $X_\n(\rho)$ as follows: 
\begin{gather} \lbeq(f-3)
X_\n(\rho)= \frac{C}{\rho^{\n+\frac32}}
\int_{\R} (L_{+}(\rho,r)+ L_{-}(\rho,r))
r^{\n+2}|r|^{-\frac12}M(r)dr, \\ 
L_{\pm}(\rho,r)= \chi_\pm(r)
\int_0^\infty e^{i\lam(\rho-r)}
\lam^{m-4}h_{+}(\lam\rho)g_{\mp}(\pm r \lam)F(\lam) d\lam.  
\lbeq(f-1)
\end{gather} 

\bglm \lblm(gpm) Suppose that $f$ is of $C^\infty$ on 
$[0,\infty)$ and satisfies 
$|f^{(j)}(c)| \leq C_jc^{-(j+1)}$ for $c\geq 1$, 
$j=0,1, \dots$. Define 
\[
\ell_\pm(t) = \int_0^\infty e^{\pm ic}f(c/t)\frac{dc}{\sqrt{c}}.
\]
Then, $\ell_{\pm}(t)$ is $C^\infty$ for $t>0$ and 
satisfies the following properties. 
\ben 
\item[{\rm (1)}] 
$\ell_{\pm}(1/t)$ can be exteded to a $C^\infty$ function 
on $[0,1]$, hence, 
$\lim_{t\to \infty} \ell_{\pm}(t)= \alpha_\pm $ exists 
and for $t\geq 1$, 
$|\ell_\pm ^{(j)}(t)|\leq C_j t^{-j-1}$, $j=1,2,\dots$. 
\item[{\rm (2)}] For $0<t<1$, 
$|t^j \ell_{\pm}^{(j)} (t)|\leq C_j \sqrt{t} \leq C_j$, 
$j=0,1, \dots$.
\een  
\edlm 
\bgpf We prove the lemma for $\ell_{+}(t)$ only 
and omit the $+$-sign. 
It is evident that $\ell(t)$ is $C^\infty$ for $t>0$. 
Splitting the interval, we define   
\[
\ell(t)=\left(\int_0^1 + \int_1^\infty\right) 
f\left(\frac{c}{t}\right)e^{ic}\frac{dc}{\sqrt{c}}
\equiv {\ell}_{1}(t)+ {\ell}_{2}(t). 
\]
It is obvious that ${\ell}_{1}(1/t)$ is of 
$C^\infty[0,1]$. 
To see the same for ${\ell}_{2}(1/t)$,  
we perform integration by parts $n$ times for $t>0$:
\bqn 
i^n {\ell}_{2}(1/t)= B_n(t)+
(-1)^n \int_1^\infty \pa_c^n \left(\frac{f(ct)}{\sqrt{c}}
\right) e^{ic}dc. \lbeq(o-1)
\eqn
The boundary term $B_n(t)$ is a polynomial of order 
$n$ and Leibniz's formula implies  
$\pa_c^n \left(\frac{f(ct)}{\sqrt{c}}\right) 
=\sum_{j=0}^n C_{nj}
f^{(j)}(ct)(ct)^j c^{-\frac12-n} $.
Since $\pa_y^{k} (f^{(j)}(y)y ^j)$ 
is bounded for any $j,k=0,1, \dots$ and 
\[
\pa_t^k \left(\sum_{j=0}^n C_{nj}
f^{(j)}(ct)(ct)^j c^{-\frac12-n} \right) 
= \sum_{j=0}^n C_{nj}
\left. \pa_y^{k} (f^{(j)}(y)y ^j)\right\vert_{y=ct} 
c^{-\frac12-n+k}, 
\]
the integral of \refeq(o-1) is a function 
of class $C^{n-1}([0,1])$. Since $n$ is arbitray, 
this proves (1). 
For proving (2), after changing the variable we decompose:
\[
\ell(t)=\sqrt{t}\left(\int_0^1 + \int_1^\infty\right) 
f(c)e^{ict}\frac{dc}{\sqrt{c}}
\equiv \sqrt{t}(\tilde{\ell}_{1}(t)+ \tilde{\ell}_{2}(t)) 
\]
We obseve that $\sqrt{t}$ satisfies the property (2) 
and that, if $\alpha(t)$ satisfies (2) and 
$|t^j \beta^{(j)}(t)|\leq C_j$, then 
so does $\gamma(t)=\alpha(t)\beta(t)$. 
Hence, $\sqrt{t}\tilde{\ell}_{1}(t)$ satisfies (2) because 
$\tilde{\ell}_{1}(t)$ is entire. To prove the same for  
$\sqrt{t}(\tilde{\ell}_{2}(t)$, it suffices to show that 
$|(t^n \tilde{\ell}_{2}(t))^{(n)}| \leq C_n$ for $0<t<1$, 
$n=0, 1,2, \dots$.
By integration by parts we have 
\begin{align*}
&(it)^n \tilde{\ell}_2 (t)= \int_1^\infty (\pa_c^n e^{itc})f(c)  
\frac{dc}{\sqrt{c}} \\
& = \sum_{j=0}^{n-1}(-1)^{j+1}
\pa_c^j \left(\frac{f(c)}{\sqrt{c}}\right)
\left. \pa_c^{n-j-1}(e^{itc})\right\vert_{c=1} + 
\int_1^\infty e^{itc}
(f(c)c^{-\frac12})^{(n)}dc.
\end{align*}
The boundary term is a polynomial of $t$ and the integral is 
$n$ times continuously differentiable and a fortiori 
$(t^n \tilde{\ell}_2(t))^{(n)}\leq C$ for $0<t<1$. 
\edpf 

Generalizing ${L}_{\pm}(\rho,r)$ of \refeq(f-1), 
we define ${L}_{\pm,\sigma}(\rho,r)$ 
for an integer $\s\geq 0$ and functions $g_\pm$ and $h$ by 
\bqn \lbeq(dlt)
{L}_{\pm,\sigma}(\rho,r)= \chi_\pm(r)
\int_0^\infty e^{i\lam(\rho-r)}
\lam^{\s}h_{+}(\lam\rho)g_{\mp}(\pm r\lam)F(\lam) d\lam  
\eqn 
so that we have $L_\pm(\rho,r)={L}_{\pm,m-4}(\rho,r)$. 

\bglm \lblm(p-1) 
Suppose that $g_\pm (t)$ and $h_{+}(t)$ are 
$C^\infty$ functions of $t>0$ and they satisfy 
following properties replacing $f$: 
\ben 
\item[{\rm (a)}] The limit $\lim_{t\to \infty} f(t)$ exists. 
\item[{\rm (b)}]  
$ |t^j f^{(j)}(t)|\leq C_j \left\{
\br{lll} t^{-1} ,  \quad &  1< t, & j=1,2,\dots, \\ 
\sqrt{t}, \quad  & 0<t<1, & j=0,1, \dots.
\er \right.$. 
\een   
Then, ${L}_{\pm,\sigma}$ is $C^\infty$ 
with respect to $\rho>0$ and $r>0$ and, 
for a constant $C>0$,  
\bqn \lbeq(dlt-con)
|{L}_{\pm,\sigma} (\rho, r)| \leq C \la \rho-r  \ra^{-(\s+1)}
\eqn 
\edlm 
\bgpf We prove the lemma for ${L}_{+,\s}$. 
The proof for ${L}_{-,\s}$ is similar. It is obvious that 
${L}_{+,\s}(\rho,r)$ is smooth and is bounded 
for $\rho, r>0$ and it suffices to prove \refeq(dlt-con) 
for $|\rho-r|\geq 1$.  
We apply integration by parts $\s+1$ times to   
\[
{L}_{+,\s}(\rho,r)=\frac{(-i)^{\s+1}}{(\rho-r)^{\s+1}} \int_0^\infty 
\left(\pa_\lam ^{\s+1}
e^{i\lam(\rho-r)}\right)
\lam^{\s}h_{+}(\lam\rho)g_{-}(r\lam)F(\lam) d\lam.
\]
By Leibniz' rule, derivatives 
$(\lam^{\s}h_{+}(\lam\rho)g_{-}(r\lam)F(\lam))^{(\kappa)}$ for 
are linear combinations of 
over $(\alpha,\beta,\gamma,\delta)$ such that 
$\alpha+\beta+\gamma+\delta=\kappa $ and $\alpha \leq \s$ 
of 
\bqn \lbeq(derva)
\lam^{\s-\kappa + \delta}
(\lam\rho)^\beta h^{(\beta)}(\lam\rho)
(r\lam)^\gamma g_{-}^{(\gamma)}(r\lam) F^{(\delta)}(\lam)
\eqn 
and they converge to $0$ as $\lam \to 0$ if $\ka\leq \s$. 
It follows that no boundary terms appear and 
$(\rho-r)^{\s+1}{L}_{+,\s}(\rho,r)$ is a linear combination 
over the same set of $(\alpha,\beta,\gamma,\delta)$ 
as above with $\kappa=\s+1$ of 
\[
I_{\alpha\beta\gamma\delta}(\rho,r)=\int_0^\infty 
e^{i(\rho-r)\lam}\lam^{\delta-1}
(\lam\rho)^\beta h^{(\beta)}(\lam\rho)
(r\lam)^\gamma g_{-}^{(\gamma)}(r\lam) 
 F^{(\delta)}(\lam)d\lam. 
\]
It suffices to show that 
$I_{\alpha\beta\gamma\delta}(\rho,r)$ is 
bounded. If $\delta\not=0$, $F^{(\delta)}(\lam)= 0$ 
outside $0<c_0<\lam<c_1<\infty$ and it is clear that  
$I_{\alpha\beta\gamma\delta}(\rho,r)\absleq C$. 
Thus, we assume $\delta=0$ in what follows. We may also 
assume $0<r<\rho<\infty$ by symmetry. We split as 
$(0,\infty)=(0,1/\rho)\cup [1/\rho,1/r]\cup (1/r,\infty)$ 
and denote integrals over these intervals by 
$I_1$, $I_2$ and $I_3$ in this order so that 
$I_{\alpha\beta\gamma\delta}(\rho,r)= I_1+ I_2+ I_3$. 

\noindent 
(1) If $0<\lam<1/\rho$ then $0<r\lam<\rho\lam <1$ and 
$(\rho\lam)^\beta h^{(\beta)}(\rho\lam) 
\absleq C \sqrt{\rho\lam}$ and 
$(r\lam)^\gamma g_{-}^{(\gamma)}(r\lam) 
\absleq C \sqrt{r\lam}$. It follows that 
\bqn \lbeq(I-1)
I_1\absleq 
C\int_0^{1/\rho} 
\sqrt{\rho{r}} d\lam 
= C \sqrt{\frac{r}{\rho}}\leq C 
\eqn 
(2) If $1/\rho\leq  \lam \leq 1/r$, we have 
$0<r\lam \leq 1\leq \rho\lam$ and we estimate as   
$(\rho\lam)^\beta h^{(\beta)}(\rho\lam) \absleq C$ 
and $(r\lam)^\gamma g_{-}^{(\gamma)}(r\lam)
\absleq C\sqrt{r\lam}$. It follows that 
\bqn \lbeq(I-2)
I_2 \absleq 
C \int_{1/\rho}^{1/r} \lam^{-\frac12}\sqrt{r}d\lam 
= 2C \sqrt{r}\left(\frac1{\sqrt{r}}-\frac1{\sqrt{\rho}}\right) 
\leq 2C. 
\eqn 
(3) Finally if $1<r\lam<\rho\lam$, then we likewise estimate 
\[
(\lam\rho)^\beta h^{(\beta)}(\lam\rho)
(r\lam)^\gamma g_{-}^{(\gamma)}(r\lam) \absleq 
C \left\{\br{ll} (r\lam)^{-1}, \ & 
\mbox{if $\beta=0, \ \gamma\not=0$} \\
(\rho\lam)^{-1}, \ &
\mbox{if $\beta\not=0, \ \gamma=0$}, \\ 
(\rho\lam)^{-1}(r\lam)^{-1}, \ & 
\mbox{if $\beta, \gamma\not=0$}. \er 
\right. 
\]
The right  hand side is  bounded by $C r^{-1}\lam^{-1}$ 
and 
\[
I_3 \absleq C \int_{1/r}^\infty \lam^{-2} r^{-1} d\lam 
= C.    
\]
This completes the proof.  
\edpf 

\begin{proposition} \lbprp(last-prop)
Let $m\geq 6$ and $\f\in \Eg$. For 
$1\leq p \leq \frac{2(m-1)}{m+1}$, we have  
\bqn \lbeq(res-znn)
\|Z^{\n\n} u\|_p \leq C_p \|u\|_p .
\eqn 
\end{proposition} 
\bgpf We recall \refeq(ann). 
\reflm(p-1) implies 
$L_{\pm}(\rho,r)\absleq C \la \rho-r \ra^{-(m-3)}$. 
It follows by Young's inequality and \refeq(f-3) that  
\bqn \lbeq(g-13) 
\|Z^{\n\n} u\|_p \leq C\|V\f\|_1 \left( \int_0^\infty\left( 
\int_{\R} 
\frac{\rho^{\frac{m-1}{p}-\frac{m+1}2}|r^{\frac{m+1}2}M(r)|}
{\la \rho-r \ra^{m-3}}dr\right)^p 
d\rho \right)^{\frac1{p}} . 
\eqn 
Define $\ka= \frac{m-1}{p}-\frac{m+1}2$, then 
$\ka\geq 0$ for $1\leq p \leq \frac{2(m-1)}{m+1}$  
and $m-3-\ka\geq \frac{3}2$ for any $1\leq p<\infty$ 
if $m\geq 6$. Thus, we may estimate   
\[
\rho^{\ka} \la \rho-r \ra^{-(m-3)} 
\leq C \left\{\br{ll} 
\la \rho-r \ra^{-\frac32}  \quad & \mbox{if} \ \ 
|r|\leq 1 \\[5pt]
\la \rho-r \ra^{-\frac32} |r^{\ka}| \quad & \mbox{if} \ \ 
|r|\geq 1 \er \right.
\]
and Young's inequality implies  
\[
\|Z^{\n\n} u\|_p  \leq  C\|V\f\|_1 
\left(\int_0^1 |r^{\frac{m+1}2}M(r)|^p dr + 
\int_1^\infty |r^{\frac{m-1}p}M(r)|^p dr \right)^{\frac1{p}},
\]
which is bounded by 
$C(\|V\f\ast u \|_{\infty}+ \|V\f \ast u\|_p) 
\leq (\|V\f\|_{p'}+ \|V\f\|_1) \|u\|_p$.  
This completes the proof of the proposition. 
\edpf 

\subsubsection{The case $2\leq p<\infty$} \lbsssec(5-2-2) 

\bglm Let $m\geq 6$ and $\f \in \Eg$. 
Then, $Z^{\n\n}(\f)$ is bounded in $L^p(\R^m)$ 
for any $p\geq 2$. 
\edlm 
\bgpf we apply integration by parts to \refeq(moto) 
by using the identity that 
$\int_{\R} e^{-i(1+2a)r\lam} r^{\n+1}M(r)dr = 
i(1+2a)^{-1}\pa_\lam 
\left(\int_{\R} e^{-i(1+2a)r\lam} r^{\n}M(r)dr\right)'$. 
We see that 
$\rho^{-\n}Q_{\n\n}^{a,b}(\rho)$ is equal to 
\begin{multline*} 
\frac{(-1)^{\n+1}}
{2\pi\rho^{\n}(i(1+2a))}\int_0^\infty 
e^{i(1+2b)\rho\lam}
(\lam^{m-3} F(\lam))'
\left(
\int_{\R} e^{-i(1+2a)r\lam} r^{\n}M(r)dr \right) 
d\lam   
\\
+ \frac{(-1)^{\n+1}(1+2b)}
{2\pi\rho^{\n-1}(1+2a)} 
\int_0^\infty 
e^{i(1+2b)\rho\lam}
\lam^{m-3} F(\lam)
\left(
\int_{\R} e^{-i(1+2a)r\lam} r^{\n}M(r)dr \right) 
d\lam 
\end{multline*} 
The argument similar to the one at the beginning of 
the proof of \refprp(last-prop) shows that the operator 
produced by the first line has the same $L^p$ property 
as $Z^{(\n-1)\n}$ and, hence, is bounded in $L^p(\R^m)$ for 
any $1<p<\infty$. Thus, we need consider the operator 
produced by the second line, which we substitute for 
$Q_{\n\n}^{a,b}(\rho)/\rho^\n$ in \refeq(zjkab-def) 
and the resulting function into \refeq(zdef) for 
$Z_{a,b}^{\n\n}(\f)u(x)$. The result is given 
by \refeq(ann) where $X_{\n}(\rho)$ is replaced by 
\begin{align}
& \tilde{X}_\n(\rho)
= \frac{C}{\rho^{\nu-1}} 
\int_0^\infty 
\left\{
e^{i\lam\rho} \lam^{m-3}
\left(
\int_0^\infty 
\frac{(1+2b)e^{2i\lam\rho {b}}}
{(1+b)^{\nu+\frac12}}
\frac{db}{\sqrt{b}}
\right) 
\right.  \notag \\
& \times \left.
\int_{\R}e^{-i\lam{r}}
\left(
\int_0^\infty 
\frac{(1+2a)^{-1}e^{-2ia\lam{r}}}{(1+a)^{\nu+\frac12}}
\frac{da}{\sqrt{a}}
\right) \lbeq(m-2a)
r^{\nu}M(r)dr 
\right\}
F(\lam)d\lam,  
\end{align}
which can be simplified into the form \refeq(f-1) 
with the roles of $g$ and $h$ being replaced and 
the factors $\rho^{-\left(\nu+\frac32\right)}$ and 
$r^{\n+2}|r|^{-\frac12}$ being replaced 
by $\rho^{-\left(\nu-\frac{1}{2}\right)}$ and 
$r^{\n}|r|^{-\frac12}$ respectively. 
Then, \reflms(gpm,p-1) imply    
\[ % \lbeq(g-3)
X_\n(\rho) \absleq  \frac{C}{\rho^{\n-\frac12}}
\int_{\R} \la \rho -r \ra^{3-m} r^{\n}|r|^{-\frac12}M(r)dr.
\]
We estimate $\| X_{\nu}(|y|)\|_{L^p(|y|\geq 1)}$ for 
$p\geq 2$ and $\| X_{\nu}(|y|)\|_{L^1(|y|<1)}$.  
Let $\kappa=\frac{m-1}p -\n +\frac12$. If $p\geq 2$, 
then $\kappa\leq 0$ and $m-3+ \kappa \geq \frac32$ 
for $m\geq 6$ and   
\[
\rho^{\kappa}\la \rho -r \ra^{3-m} |r|^{\n-\frac12}
\leq C \la \rho -r \ra^{-\frac32}|r|^{\frac{m-1}{p}}, 
\quad\mbox{for}\ \ \rho\geq 1.
\]
It follows by Young's inequality that for any 
$2\leq p<\infty$,  
\begin{align}
\| X_{\nu}(|y|)\|_{L^p(|y|\geq 1)} 
& \leq C \left(\int_0^\infty \left|
\int_{\R} \la \rho -r \ra^{-\frac32} 
|r|^{\frac{m-1}{p}}M(r)dr
\right|^p d\rho\right)^{\frac1{p}} \notag \\
& \leq C \left(\int_{0}^\infty |M(r)|^p 
r^{m-1} dr\right)^{\frac1{p}} 
\leq C \|V\f\|_1 \|u\|_p . \lbeq(-1)
\end{align}
When $\rho\leq 1$, we have $\rho^{m-1-\n+\frac32}\leq 1$ 
and $\la \rho -r \ra^{3-m} \leq C \la r \ra^{3-m}$. 
Hence, 
\[
\|X_\n(|y|)\|_{L^1(|y|<1)} 
\leq C\int_{\R} \la r \ra^{3-m} r^{\n-\frac12}|M(r)|dr 
\leq C \|M\|_\infty \leq C \|V\f\|_{p'}\|u\|_p .
\] 
We therefore obtain by using Young's inequlity again 
after splitting the integral corresponding to 
\refeq(ann) into the ones over $|y|<1$ 
and $|y|\geq 1$ that 
\[
\|Z^{\n\n}u\|_p 
\leq C(\|V\f\|_1^2 + \|V\f\|_p \|V\f\|_{p'})\|u\|_p. 
\]
This completes the proof. 
\edpf

\subsection{Estimate of $\|Z_{\log}u\|_p$} 

In this section we study $Z_{\log}$ and prove the 
following lemma. 
The combination of the lemma with results of the 
previous subsections  
proves \refth(theo6) for even dimensions $m\geq 6$, 
the formal proof of which will be omitted.

\bglm \lblm(log) 
\ben 
\item[{\rm (1)}] If $m=6$, then $Z_{\log}$ is bounded 
in $L^p(\R^m)$ for all $1<p<m$. If $\Eg=\Eg_0$, then so is 
$Z_{\log}$ for all $1<p<\infty$. 
\item[{\rm (2)}] If $m\geq 8$, then $Z_{\log}$ is bounded 
in $L^p(\R^m)$ for all $1<p<\infty$.
\een 
\edlm 
\bgpf We prove the lemma for $m=6$ only. 
The proof for $m\geq 8$ is similar and easier. 
Out of three operators on the right of \refeq(loge) for $m=6$, 
we first study 
\bqn 
Z_{1,\log}= \int^\infty_0 
G_0(\lam)(V\ph \otimes V\ph) \lam
\log\lam 
(G_0(\lam)-G_0(-\lam))F(\lam)d\lam, \lbeq(1-log) 
\eqn 
where we have ignored the constant $\w_{m-1}/\pi(2\pi)^m$ 
which is not important. 
Since $Z_{1,\log}=0$, if $\Eg=\Eg_0$, it suffices to prove 
(1) for $1<p<\frac{m}{m-1}$ and $\frac{m}2<p<m$. 
By using \refeq(co-kerb) and \refeq(usformula) as previously,  
we express $Z_{1,\log}$ as the sum over $0\leq j,k\leq \n$ of 
\bqn 
Z_{1,\log}^{jk}u(x) = C_{jk} 
T^{(a)}_j  T^{(b)}_k \left[ \int_{\R^m} 
\frac{(V\f)(x-y) Q^{a,b}_{jk,\log}(|y|)}{|y|^{m-2-k}}dy 
\right], \lbeq(Wsm-rep-log) 
\eqn 
where $Q_{jk,\log}^{a,b}(\rho)$ are defined by 
\refeq(kjkab) or \refeq(k0kab) (for the case $j=0$) 
by replacing $\lam^{j+k-1}$ or $\lam^k$ respectively by 
$\lam^{j+k+1}\log\lam$ or $\lam^{k+2}\log\lam$. 
We prove 
\bqn \lbeq(logest)
\|Z_{1,\log}^{jk}u\|_p \leq C\|u\|_p 
\eqn 
separately for $(j,k)\not=(\n,\n)$ and $(j,k)=(\n,\n)$ 
by repeating the argument in corresponding subsections.  

Let $(j,k)\not=(\n,\n)$. We first observe that, if $j \geq 1$, 
Fourier inverse transforms of the derivatives upto 
the order $k+1$ of $\lam^{j+k+1}(\log \lam) F$ have the RDIM 
\[
\Fg^{\ast}(\lam^{j+k+1}(\log \lam) F)^{(l)})(\rho)
\absleq C(1+ \rho)^{-2}\la \log (1+ \rho) \ra, \quad 
0\leq l \leq k+1 
\]
and estimates corresponding to \refeq(devel-a) and 
\refeq(fg-1) are satisfied 
by $Q_{jk,\log}^{a,b}(\rho)$ respectively for $1\leq j \leq \n$ 
and for $j=0$ (without producing the boundary term).   
Then, the argument in \refsssec(6-1-1) goes through for 
$Z_{1,\log}^{jk}$ and produces estimate \refeq(logest) 
for $1<p<\frac{m}{m-1}$. By the same reason the 
estimate corresponding \refeq(kjkab-sub-mre) for $m/2<p<m$ 
is satisfied by $Q_{jk,\log}^{a,b}(\rho)$ for all $j,k$ 
and we likewise 
have \refeq(logest) for $m/2<p<m$ by using the argument of 
the first part of proof of \reflm(even-m/2m). It is then  
obvious that the same holds for $Z_{2,\log}$ which is 
obtained from $Z_{1,\log}$ by replacing $\lam \log\,\lam$ by 
$\lam^3 (\log\, \lam)$ and, that the operator 
\bqn 
Z_{3,\log}^{(a,b)}= \int^\infty_0 
G_0(\lam)(\ph_a \otimes \p_b)\lam^3 (\log\, \lam)^2 
(G_0(\lam)-G_0(-\lam))F(\lam)d\lam.  \lbeq(3-log)
\eqn 
produced by $\lam^2\log\lam F_2$ of \refeq(sing-1) 
satisfies \refeq(logest) for all $1<p<m$. 

We next prove \refeq(logest) when $(j,k)=(\n,\n)$. 
It suffices prove it for $1 < p<p_0$ for some $p_0>1$ and 
$p\geq 2$. 
The argument at the beginnings of \refsssec(5-2-1) 
and \refsssec(5-2-2) shows that respectively for 
$1 < p<p_0$ and $p\geq 2$, we have only to estimate 
operators obtained by replacing $Q^{a,b}_{jk,\log}(\rho)$ by   
\begin{gather*} 
\frac{1+2a}{(1+2b)\rho^{\n+1}} \int_0^\infty 
e^{i(1+2b)\rho\lam}n
\lam^{m-1}\log{\lam} F(\lam)
\left(
\int_{\R} e^{-i(1+2a)r\lam} r^{\n+2}Mdr \right) 
d\lam   \ \mbox{and} \\
\frac{1+2b}
{(1+2a)\rho^{\n-1}} \int_0^\infty 
e^{i(1+2b)\rho\lam}
\lam^{m-1}\log{\lam} F(\lam)
\left(
\int_{\R} e^{-i(1+2a)r\lam} r^{\n+2}Mdr \right) 
d\lam
\end{gather*}
in \refeq(Wsm-rep-log). We then repeat the argument of 
\refssec(5-2). We have $\lam^{m-2} \log\lam$ in place of 
$\lam^{m-4}$ in \refeq(f-1). If we change 
$\lam^\sigma$ by $\lam^{\sigma+2} \log \lam $ in the definition 
\refeq(dlt) of $\tilde{L}_\pm (\rho, r)$, then 
\refeq(dlt-con) is satisfied with faster decaying factor 
$\la \rho-r\ra^{-(\sigma+2)}$ in place of   
$\la \rho-r\ra^{-(\sigma+1)}$. 
Thus, $\|Z_{\log}^{\n\n}u\|_p$ is bounded $C\|V\f\|_1$ 
times \refeq(g-13) with $\la \rho-r\ra^{-(m-2)}$ in place of 
$\la \rho-r\ra^{-(m-3)}$ and this proves the lemma for 
$1<p<p_0$. The proof for $p\geq 2$ is similar and we omit 
further details.  
\edpf

%%%%%%%%%%%%%%%%%%%%%%%%%%%%%%%%%

%\end{footnotesize}


\begin{thebibliography}{99}

\bibitem{Ag} S. Agmon, 
\textit{Spectral properties of Schr\"odinger
operators and scattering theory}, Ann. Scuola 
Norm. Sup. Pisa Cl. Sci. (4)
 {\bf 2} (1975), 151--218. 

\bibitem{AS} M. Aizenman and B. Simon, 
{\it Brownian motion and Harnack inequality 
for Schr\"odinger equations}, 
Comm. Pure Appl. Math. {\bf 35} (1982). 


\bibitem{AY} G. Artbazar and K. Yajima, 
\textit{The $L^p$-continuity of 
wave operators for one dimensional Schr\"odinger 
operators}, J. Math. Sci. 
Univ. Tokyo {\bf 7} (2000), 221-240. 


\bibitem{Be} M. Beceanu, {\it Structure of wave operators 
for a scaling-critical class of potentials}, 
Amer. J. Math. {\bf 136} (2014), 255--308. 

\bibitem{BL} J. Bergh and J. L\"ofstr\"om, 
\textit{Interpolation spaces, an introduction}, 
Springer Verlag, 
Berlin-Heidelberg-New York (1976). 

\bibitem{DF} P. D'Ancona and L. Fanelli, 
{\it $L^p$--boundedness of the wave operator for the 
one dimensional Schr\"odinger operator}, 
Commun. Math. Phys. {\bf 268} (2006), 415--438. 

\bibitem{ES} M. B. Erdo\v{g}an and W. Schlag, 
\textit{Dispersive Estimates for Schr\"odinger Operators in the
Presence of a Resonance and/or an Eigenvalue at Zero
Energy in Dimension Three: I},  Dyn. Partial Differ. Equ. {\bf 1} 
(2004), 359--379. 



\bibitem{FY} D. Finco and K. Yajima, {\it The 
$L^p$ boundedness of wave operators for Schr\"odinger 
operators with threshold singularities. II. 
Even dimensional case}. J. Math. Sci. Univ. Tokyo 
{\bf 13} (2006), 277--346. 


\bibitem{GY} A. Galtbayer and K. Yajima 
{\it Resolvent estimates in amalgam spaces and 
asymptotic expansions for Schr\"odinger equations}, 
J. Math. Soc. Japan 


\bibitem{GG} M. Goldberg and W. Green, 
{\it The $L^p$ boundedness of wave operators for Schr\"odinger 
operators with threshold singularities}, arXiv:1508.06300. 


\bibitem{Gr} L. Grafakos, 
{\it Modern Fourier analysis}, Springer Verlag, 
New York (2009). 

\bibitem{GKP} R. L. Graham, D. E.Knuth and O. Patashnik, 
{\it Concrete matematics}, Addison-Wesley Pub. Co. 
New York (1988). 


\bibitem{Jensen} A. Jensen, 
{\it Spectral properties of Schr\"odinger operators and 
time decay of the wave functions, Results in $L_2(\R^m)$, 
$m\geq 5$}, Duke Math. J. {\bf 47} (1980), 57--80. 

\bibitem{JK} A. Jensen and T. Kato, 
{\it Spectral properties of Schr\"odinger operators 
and time-decay of the wave functions}, 
Duke Math. J. {\bf 46} (1979), 583--611. 

\bibitem{JY-1} A. Jensen and K. Yajima 
{\it A remark on $L^2$-boundedness of wave operators 
for two dimensional Schr\"odinger operators}, Commun. 
Math. PHys. {\bf 225} (2002), 633--637. 

\bibitem{JY-2} A. Jensen and K. Yajima, 
{\it On $L^p$-boundedness of wave operators for 
$4$-dimensional Schr\"odinger with threshold singularities}, 
Proc. Lond. Math. Soc. (3) {\bf 96} (2008), 136--162. 

\bibitem{Kato-e} T. Kato, {\it Growth 
properties of solutions 
of the reduced wave equation with a 
variable coefficient}, Comm. Pure 
Appl. Math {\bf 12} (1959), 403--425.  

\bibitem{KatoS1} T. Kato, {\it Wave 
operators and similarity for 
non-selfadjoint operators}, 
Ann. Math. {\bf 162} (1966), 258--279. 


\bibitem{Ku-0} S. T. Kuroda, {\it Introduction 
to Scattering Theory}, 
Lecture Notes, Aarhus University

\bibitem{Mu} M. Murata, \textit{Asymptotic 
expansions in time for solutions
of Schr\"{o}dinger-type equations}, J. Funct. Anal., 
{\bf 49} (1982), 10--56. 


\bibitem{RSi} M. Reed and B. Simon, 
{\it Methods of modern mathematical physics vol II, 
Fourier analysis, selfadjointness}, Academic Press, 
New-York, San Francisco, London (1975). 

\bibitem{RS3} M. Reed and B. Simon, 
{\it Methods of modern mathematical physics vol III, 
Scattering theory}, Academic Press, New-York, 
San Francisco, London (1979). 

\bibitem{RS4} M. Reed and B. Simon, 
{\it Methods of modern mathematical physics vol IV, 
Analysis of Operators}, Academic Press, New-York, 
San Francisco, London (1978). 



\bibitem{Stempak} K. Stempak{it Transplantation 
theorems, survey}, J. Fourier Anal. Appl. {\bf 17} 
(2011), 408--430. 


\bibitem{Stein-O} E. M. Stein, 
{\it Singular integrals and differentiability properties 
of functions}, Princeton Univ. Press, Princeton, NJ. (1970)


\bibitem{Stein} E. M. Stein, {\it Harmonic analysis, 
real-variable methods, orthogonality, and 
oscillatory integrals}, Princeton Univ. Press, 
Princeton, NJ. (1993). 

\bibitem{Weder} R. Weder, 
\textit{$L^p$-$L^{p'}$ estimates for the 
Schr\"odinger equations on the line and 
inverse scattering for the nonlinear 
Schr\"odinger equation with a potential}, 
J. Funct. Anal. {\bf 170} 
(2000), 37--68. 

\bibitem{WW} E. T. Whittaker and G. N. Watson, 
{\it A Course of Modern Analysis}. 
Cambridge University Press; 4th edition (1927). 

\bibitem{Y-d3} K. Yajima,  \textit{The 
$W^{k,p}$-continuity of wave operators for 
Schr\"odinger operators}, J. Math. Soc. 
Japan {\bf 47} (1995), 551-581.

\bibitem{Y-d4} K. Yajima,  
\textit{The $W^{k,p}$-continuity of wave operators 
for Schr\"odinger operators III}, J. Math. Sci. 
Univ. Tokyo {\bf 2} (1995), 311--346. 

\bibitem{Y-d2} K. Yajima, 
\textit{The $L^{p}$-boundedness of wave operators 
for two dimensional Schr\"odinger operators}, 
Commun. Math. Phys. {\bf 208} (1999), 125--152. 

\bibitem{Y-disp} K. Yajima, 
\textit{Dispersive estimate for Schr\"odinger equations 
with threshold singularities}, 
Comm. Math. Phys. 259 (2005), no. 2, 475--509. 

\bibitem{Y-odd} K. Yajima, {\it The $L^p$ boundedness of wave 
operators for Schr\"odinger operators with threshold 
singularities I, Odd dimensional case}, J. Math. Sci. 
Univ. Tokyo {\bf 7} (2006), 43--93. 

\bibitem{Y-arxiv} K. Yajima,
{\it Wave operators for Schr\"odinger operators with threshold 
singuralities,revisited}, arXiv:1508.05738. 

\end{thebibliography}
\end{document}